\documentclass[a4paper,11pt]{article}
\usepackage{latexsym,graphicx,amssymb,color,amsmath}
\usepackage{relsize,ccaption,url,latexsym,epsfig,a4wide}
\usepackage{lscape}
\usepackage{hyperref}
\usepackage{pdfsync}
\usepackage{epstopdf}
\numberwithin{equation}{section}
\newcommand{\be}{\begin{equation}}
\newcommand{\ee}{\end{equation}}
\newcommand{\ba}{\begin{eqnarray}}
\newcommand{\ea}{\end{eqnarray}}
\newcommand{\hf}{\frac{1}{2}}
 
\newcommand\req[1]{(\ref{#1})}

\newcommand\F{\mathbb{F}}
\newcommand\C{\mathbb{C}}
\newcommand\N{\mathbb{N}}
\newcommand\Q{\mathbb{Q}}
\newcommand\R{\mathbb{R}}
\newcommand\Z{\mathbb{Z}}

\newcommand\KKK{\mathcal{K}}
\newcommand\MMM{\mathcal{M}}
\newcommand\OOO{\mathcal{O}}
\newcommand\PPP{\mathcal{P}}
\newcommand\TTT{\mathcal{T}}

\newcommand\Aff{\mathop{\mathrm{Aff}}}
\newcommand\Aut{\mathop{\mathrm{Aut}}}
\newcommand\diag{\mathop{\mathrm{diag}}}
\newcommand\disc{\mathop{\mathrm{disc}}}
\newcommand\GL{\mathop{\mathrm{GL}}}
\newcommand\Hom{\mathop{\mathrm{Hom}}}
\newcommand\Pic{\mathop{\mathrm{Pic}}}
\newcommand\PSL{\mathop{\mathrm{PSL}}}
\newcommand\SO{\mathop{\mathrm{SO}}}
\newcommand\spann{\mathop{\mathrm{span}}\nolimits}
\newcommand\SU{\mathop{\mathrm{SU}}}

\newcommand\fa{\forall\,}
\newcommand\m{\mathcal}
\newcommand\qu{\overline}
\newcommand\rk{\mathop{\mathrm{rk}}\,}
\newcommand\wt{\widetilde}
\newcommand\wh{\widehat}
\newcommand\ord{\mathop{\mathrm{ord}}}

\newtheorem{definition}{Definition}[subsection]
\newtheorem{corollary}[definition]{Corollary}

\newtheorem{prop}[definition]{Proposition}
\newtheorem{theorem}[definition]{Theorem}
\newtheorem{remark}[definition]{Remark}

\newenvironment{proof}{\textsl{Proof:}}{\hspace*{\fill}$\blacksquare$\\}
\newenvironment{proofsketch}{\textsl{Proof (sketch):}}{\hspace*{\fill}$\blacksquare$\\}
\newenvironment{example}{\noindent\textsl{Example:}}{\hspace*{\fill}\\}

\captionnamefont{\bfseries}
\title{\mbox{}\hfill {\small DCPT-11/31}\\[40pt]
The overarching finite symmetry group of Kummer surfaces in the Mathieu group $M_{24}$}

\date{\vspace{-5ex}}
\author{\\ \Large{Anne Taormina\footnote{anne.taormina@durham.ac.uk}\;\; and Katrin Wendland\footnote{katrin.wendland@math.uni-freiburg.de}}\\ \\ \\
 \normalsize{$^*$Centre for Particle Theory \& Department of Mathematical Sciences,} \\ \normalsize{Durham University, Science Laboratories, South Road, Durham, DH1 3LE, U.K. }\\
 \normalsize{$^{\dagger}$Mathematics Institute, Albert-Ludwigs-Universit\"at Freiburg,}\\
 \normalsize{Eckerstra\ss e 1, Freiburg im Breisgau, D-79104, Germany.  }}
 
\begin{document}

\maketitle

\setlength{\parindent}{0pt}
\begin{abstract}
In view of a potential interpretation of the role of the Mathieu group $M_{24}$ in 
the context of strings compactified on K3 surfaces, we develop techniques 
to combine groups of symmetries from  different K3 surfaces to
larger `overarching' symmetry groups. 
We construct a bijection between 
the full integral homology lattice of K3 and the Niemeier lattice of type $A_1^{24}$,
which is simultaneously compatible with the finite symplectic automorphism
groups of  all  Kummer surfaces lying on {an appropriate path in moduli space 
connecting the square and the tetrahedral Kummer surfaces}. The Niemeier lattice
serves to express all these symplectic automorphisms as
elements of the Mathieu group $M_{24}$, generating the `overarching finite symmetry group'
$(\Z_2)^4\rtimes A_7$ of Kummer surfaces. This group has order $40320$, thus
surpassing the size of the largest finite symplectic automorphism
group of a K3 surface by orders of magnitude. For every Kummer surface this group
contains the group of symplectic automorphisms 
leaving the K\"ahler class invariant which is
induced from the underlying torus.
Our results are in line with the existence proofs of Mukai and Kondo,
that finite groups of symplectic automorphisms of K3 are subgroups of one of eleven
subgroups of $M_{23}$, and we extend their techniques of lattice embeddings
for all Kummer surfaces with K\"ahler class induced from the underlying torus.
\end{abstract}

\section{Introduction}
$M_{24}$ is the largest in a family of five sporadic groups - amongst the 26 appearing in the 
classification of finite simple groups - that has rekindled interest in the mathematical physics 
community following an intriguing remark published by Eguchi, Ooguri and Tachikawa \cite{eot10}. This remark stems from an 
expression for the elliptic genus of a K3 surface that uses knowledge of 2-dimensional $N=4$ superconformal field theory and Witten's construction of elliptic genera {\cite{akmw87,wi87}}. That the K3 elliptic genus, 
which is a weak Jacobi form of weight 0 and index 1, may be expanded in a linear combination of $N=4$ 
superconformal characters is not surprising. Indeed in the context of superstring theory, it has long been established that compactification on a K3 surface, which is a hyperk\"ahler manifold, yields a world-sheet theory that is invariant under $N=4$ superconformal transformations. The K3 elliptic genus may be calculated as a specialisation of the corresponding partition function, which is a sesquilinear expression in the $N=4$ characters \cite{eoty89}.

What {\em is} surprising and remains to be fully understood, is that the coefficients of the non-BPS $N=4$ characters in the elliptic genus 
decomposition coincide with the dimensions of some irreducible and reducible representations of the sporadic group $M_{24}$. Actually, on the basis of the information encoded in the K3 elliptic genus  alone, the dimensions could be those of representations of the Mathieu group $M_{23}$, the stabilizer in $M_{24}$ of an element in the set ${\cal I}=\{1,2,...,24\}$, when viewing $M_{24}$ as the group of permutations of 24 elements preserving the extended binary Golay code ${\cal G}_{24}$.  It turns out that this Mathieu group information is encoded in the K3 elliptic genus in the form of a weakly holomorphic mock modular form of weight $\hf$ on $SL(2,\mathbb{Z})$, suggesting the existence of a ``Mathieu Moonshine'' phenomenon
\cite{ch10,ghv10a,ghv10b,go10}.

Yet, an interpretation of the appearance of $M_{24}$ as a symmetry group within a theory of strings compactified on K3 is lacking. One difficulty is that the $M_{24}$ sporadic group does not act in a conventional way on the string states. Rather, the symmetry manifests itself when a specific subset of BPS states is considered. Another particularity is that by nature, 
the elliptic genus is an invariant on each irreducible component of the moduli space of $N=(2,2)$ superconformal
field theories. The moduli
space of  superconformal 
theories on K3 is such an irreducible component. Hence, the $M_{24}$-information the elliptic genus carries is unaltered when surfing the K3 moduli space between generic points and isolated points with enhanced symmetry, of which Gepner models are the main examples. This suggests that $M_{24}$, in a sense yet to be uncovered, `overarches' the symmetries associated with all superconformal field theories in the K3 moduli space.

{Although we have, at this stage, little new to say regarding the interpretation of an $M_{24}$-action within the framework of 
strings compactified on K3 surfaces, we show here how an overarching symmetry, smaller than $M_{24}$, emerges when considering 
the finite groups of symplectic automorphisms of {\em all}  K3 surfaces of a particular type.
Namely, we consider symmetry groups of Kummer surfaces, by which we mean the groups of 
those symplectic automorphisms that
preserve the K\"ahler class 
induced by that  of the underlying complex torus. 
The results 
presented in this work should be regarded as an attempt to set the scene for further investigations pinning down 
an `overarching' $M_{24}$-action as alluded to above. We revisit the results obtained by Mukai in \cite{mu88},  
stating that any finite group of symplectic automorphisms of a K3 surface is isomorphic to  a subgroup of the 
Mathieu group $M_{23}$ which has at least 5 orbits on the set ${\cal I}$ of 24 elements.} We develop a technique 
which allows to combine the symmetry
groups of several different Kummer surfaces. With this technique, we obtain the group
$(\Z_2)^4\rtimes A_7$ of order $2^7\cdot 3^2\cdot 5\cdot 7=40320$, which is  a maximal subgroup of $M_{23}$. 
This group contains as proper subgroups all symplectic automorphism groups of 
Kummer surfaces preserving the induced K\"ahler class. In this sense, we find an `overarching' 
symmetry group of all Kummer surfaces.

{In \cite{ko98}, Kondo rederives Mukai's result, using
ingenious lattice techniques. The framework we set up here is in this spirit. We
use the Niemeier lattice of type $A_1^{24}$, denoted $N$ hereafter, as a device which for all Kummer surfaces
encodes the action of the symmetry groups
as automorphisms of the extended binary Golay code, 
represented by permutations on 24 elements belonging to $M_{24}$. This encoding is linked to the construction 
of a bijection {$\theta$}
between the full integral homology lattice $H_\ast(X,\Z)$ of a Kummer surface $X$ 
and the negative definite version of the Niemeier lattice $N$, denoted $N(-1)$. 
Here, the lattice $H_\ast(X,\Z)$ is identified with the standard unimodular lattice of the appropriate
signature by an isometry which is induced by the Kummer construction and thus is naturally fixed and
compatible with all Kummer surfaces.
Although in general, 
the bijection {$\theta$} between $H_\ast(X,\Z)$ and $N(-1)$
depends on the complex structure and K\"ahler class of the Kummer surface considered, 
we are able to construct a unique bijection ${\Theta}$, up to a few choices of signs, 
which for two specific Kummer surfaces
denoted $X_0$ and $X_{D_4}$  yields 
an action of their  distinct  symmetry groups, ${\cal{T}}_{64}$ and ${\cal{T}}_{192}$ respectively, on the same 
Niemeier lattice $N$. 
{Moreover, $\Theta$ is compatible with the \textsl{generic symmetry group} 
$G_t\cong(\Z_2)^4$ of Kummer surfaces.}
In fact, we argue that we are naturally working in {a smooth connected cover} 
$\wt\MMM_{hk}$ of the moduli space of hyperk\"ahler structures on K3. 
Namely, the symplectic automorphism group of a K3 surface
which preserves a given K\"ahler class solely depends on the hyperk\"ahler structure determined by the invariant
complex structure and K\"ahler class. The choice of an isometric identification of $H_\ast(X,\Z)$ with a fixed standard
lattice amounts to the transition to the {smooth connected cover $\wt\MMM_{hk}$ of the moduli space}. 
We show that  
our bijection
${\Theta}$ is compatible with the {symmetry group of all K3 surfaces} 
along a particular path 
which connects  $X_0$ and $X_{D_4}$ in $\wt\MMM_{hk}$. 
{Indeed, there exists such a path consisting of Kummer surfaces, all of whose symmetry groups 
away from beginning and end of the path
restrict to the generic one, $G_t\cong(\Z_2)^4$.}
Our bijection ${\Theta}$
therefore effectively  reveals the presence of a larger, overarching and 
somehow hidden symmetry that could manifest itself  in an unsuspected way within  string theories 
compactified on K3 surfaces. In our eyes, the significance of our result is that it provides a framework 
allowing to pin down {overarching symmetry groups that transcend known symmetry groups of 
K3 surfaces given by distinct points in our moduli space.} In the specific case of $\wt\MMM_{hk}$, one 
ultimately wishes the overarching {symmetry group} to be $M_{24}$, in accordance with the 
information encoded in the K3 elliptic genus. Although the overarching symmetry group $(\Z_2)^4\rtimes A_7$ 
we find is two orders 
of magnitude larger than the biggest finite symplectic
automorphism group of a K3 surface, 
it is still by orders of magnitude smaller than the Mathieu group $M_{24}$. 
In order to obtain $M_{24}$,
 further techniques and different lattices are required, and this is beyond the scope of the present work.

Mathematically, one of our main results  is the following theorem,
which holds for all Kummer surfaces  and which we prove in section \ref{k3geometry},
see theorem  \ref{biglattice}:

{\bf Theorem.}
{\em Let $X$ denote a Kummer surface with underlying torus $T$, equipped with the 
complex structure and K\"ahler class which are induced from $T$.
By $E_{\vec a}, \vec a\in\F_2^4$, we denote the classes in $H_2(X,\Z)$  obtained 
in the Kummer construction by blowing up the $16$ singular points of
$T/\Z_2$.  Let $G$ denote the
group of symplectic automorphisms of $X$ which preserve the K\"ahler class. With 
$L_G= \left( H_\ast (X,\Z)^G \right)^\perp \cap H_\ast(X,\Z)$ and
$\upsilon_0\in H_0(X,\Z)$,
$\upsilon\in H_4(X,\Z)$ such that
$\langle\upsilon_0,\upsilon\rangle=1$, and $e:={1\over2}\sum_{\vec a\in\F_2^4} E_{\vec a}$,
the lattice $M_G(-1)$ with
$$ 
M_G:= L_G \oplus \mbox{span}_\Z \left\{ e, \upsilon_0-\upsilon \right\}
$$
can be primitively embedded in the Niemeier lattice $N$ of type $A_1^{24}$.}

The theorem generalizes Kondo's proof \cite{ko98} to Mukai's theorem \cite{mu88}
in the following sense:
For a K3 surface $X$, Kondo proves that there exists some Niemeier lattice
$\wt N$, such that the lattice $L_G$ can be primitively
embedded in $\wt N(-1)$. Although Kondo primitively embeds a lattice $L_G\oplus\langle-2\rangle$,
the additional direction $\langle-2\rangle$ is not identified geometrically within $H_\ast(X,\Z)$.
Our first contribution thus is the natural geometric interpretation
$\langle-2\rangle=\spann_\Z\{\upsilon_0-\upsilon\}$. Second, for every Kummer surface
with induced K\"ahler class, our theorem shows that in Kondo's proof, the special Niemeier
lattice $N$ of type $A_1^{24}$ can always be used\footnote{This nicely ties in with 
Mukai's result \cite[Appendix]{ko98} that there exists a symplectic action on the special Niemeier lattice $N$
for each
group $G$ in Mukai's classification.}.
Third, the lattice $M_G$ which we embed
contains Kondo's lattice $L_G$ and obeys ${\rm rk}\,(M_G)={\rm rk}\,(L_G)+2$. 
This embedding is the restriction to $M_G$ of the bijection $\Theta$ between
$H_\ast(X,\Z)$ and $N(-1)$
described above. The construction of such a bijection is an application of 
the ``gluing techniques"
due to Nikulin \cite{ni80b,ni80} that build on previous work by Witt \cite{wit41} and 
Kneser \cite{kn57}.  
Our specific bijection ${\Theta}$ isometrically  identifies two
{\em different} primitive sublattices $M_{\m T_{192}}$ and $M_{\m T_{64}}$  
in $H_\ast(X,\Z)$ with their images. 
Furthermore, we force these embeddings to be equivariant with respect to the two
corresponding symmetry groups
$G=\m T_{192}$ and $G=\m T_{64}$, thus obtaining  natural actions of these two groups
on $N$. This allows us to view $N$ as a device which carries both these
group actions, and in fact the actions of all symmetry groups along an appropriate
path in moduli space,
which  combine to the action of a larger, overarching
group. Interestingly, the resulting group already gives the overarching 
symmetry group of all Kummer surfaces.
{Since our techniques confine us to the study of symmetries
of Kummer surfaces, the fact that we can generate the entire overarching symmetry group
of Kummer surfaces may very well mean that on transition to more general techniques, employing
the Leech lattice instead of the Niemeier lattice $N$, say, one could obtain the entire group $M_{24}$.
We are currently working on such a generalisation.}

Our techniques were originally designed to study non-classical symmetries of superconformal
field theories compactified on K3 surfaces, particularly to develop a device which could maybe
distinguish between those symmetries which play a role for ``Mathieu Moonshine" and
those that don't. Indeed\footnote{We are grateful to Rob Curtis for explaining to us in 
April 2010 that the well-known symmetry  group   \break $C_4^3\rtimes S_4$
of the Gepner model $(2)^4$ is not a subgroup of $M_{24}$.},
in the beautiful work \cite{ghv11} it is proved that the symmetry groups of superconformal
field theories on K3 in general are not subgroups of $M_{24}$, but instead
that they are all contained in the Conway group $C\!o_1$. Our findings could in fact indicate
that the classical symmetries of K3 are the only ones that play a role for ``Mathieu Moonshine".

The paper is organised as follows. 

{Section \ref{lattice} starts with a short review of Nikulin's  lattice gluing prescription and  illustrates it in two cases of direct interest to us:  the gluing of the Kummer lattice to the lattice stemming from the underlying torus of a given Kummer surface, including an extension of this technique to recover the full integral homology, as well as the reconstruction of the Niemeier lattice $N$ associated with the root lattice of $A_1^{24}$ through the gluing of two different pairs of sublattices. 
This construction of the integral K3-homology provides an isometric identification of the homology
lattice $H_\ast(X,\Z)$ with a standard lattice which is induced by the Kummer construction and which
is used throughout this work.
Built in this section is  proposition \ref{piprim}, stating and proving  that 
the Niemeier lattice $N$ possesses a rank 16 primitive sublattice $\wt \Pi$ which, up to a total reversal of signature, is isometric to the Kummer lattice $\Pi$. That the primitive embedding of the Kummer lattice in $N(-1)$ is unique
up to automorphisms of $N$ is proven in proposition \ref{uniquepi}. Moreover, an explicit realisation of the primitive sublattice $\wt \Pi$ that proves extremely useful in our construction  is provided via the map \req{map}. These properties of the Niemeier lattice $N$ of type $A_1^{24}$, to our knowledge, 
had not been noticed before. }

{After reviewing  basic albeit crucial aspects of complex structures, K\"ahler forms and symplectic automorphisms of K3 surfaces, 
we argue that by our assumptions, we are in fact working on the smooth connected cover $\wt\MMM_{hk}$ of the
moduli space of hyperk\"ahler structures on K3. 
Next, after
providing a summary of the results obtained by Mukai and Kondo that are important for our analysis, Section \ref{k3geometry} generalizes Kondo's results to the full integral homology lattice to fit our purpose. Moreover, it introduces the lattice 
$M_G\supsetneq L_G$ that is so crucial for our construction, and proves that this lattice, after a total reversal of signature, may be embedded primitively in the Niemeier lattice $N$, as stated in theorem \ref{biglattice} and as discussed above.}

{Section \ref{overarching} builds on the previous sections to bring to light the overarching symmetry that is at the centre of this work. This is achieved by constructing a linear bijection $\Theta$ that extends the isometry between the Kummer lattice and $\wt \Pi(-1)$ to the full integral homology $H_\ast(X,\Z)$, which as mentioned above is explicitly identified with a fixed standard
lattice by means of the Kummer construction,
and $N(-1)$. We start in  subsection \ref{genericKummer} with the discussion of the translational automorphism group that is a normal subgroup of all symplectic automorphism groups of Kummer surfaces which preserve the induced K\"ahler class, and show how its action on $X$ induces an action of $(\Z_2)^4$ on $N$. Subsections \ref{tetrahedral} and \ref{square} are 
dedicated to the tetrahedral Kummer surface $X_{D_4}$ with symmetry group ${\cal T}_{192} \subset M_{24}$, and to the Kummer surface $X_{0}$ obtained from the square torus $T_0$, whose symmetry 
group is ${\cal T}_{64}\subset M_{24}$. In subsection \ref{overmen} we prove that these
two Kummer surfaces enjoy the same `overarching' linear bijection ${\Theta}\colon H_\ast(X,\Z)\longrightarrow N(-1)$, leading to the concept of overarching symmetry.}
Moreover, we construct a smooth path from $X_0$ to $X_{D_4}$ in the smooth connected cover $\wt\MMM_{hk}$ of
the moduli space of hyperk\"ahler structures on K3 such that the symmetry group of every K3 surface along the path
is compatible with   ${\Theta}$.



{We end by giving our conclusions and interpretations of our results, including some evidence
for the expectation
that our techniques may eventually show Mathieu moonshine to be linked to the group $M_{24}$
rather than $M_{23}$.  
{Appendix} \ref{MOG}
provides a brief review of some of the notions and techniques we borrowed from group theory in the course of our work, while appendix \ref{latticeappendix}
gathers notations for most lattices introduced and used in the main body of the text.}

\section{Lattices: Gluing and primitive embeddings}\label{lattice}
In this work, we investigate certain finite groups of symplectic 
automorphisms of Kummer surfaces in terms of subgroups
of the Mathieu group $M_{24}$. 
As we explain below, due to the Torelli Theorem \ref{torellithm} 
for K3 surfaces, on 
the one hand, and the natural action of $M_{24}$ on the Niemeier lattice $N$
of type $A_1^{24}$
by proposition \ref{M24onN}, 
on the other hand, both group actions are naturally described in terms of lattice isometries. 
Therefore, some of the main techniques of our work rest on lattice constructions, which shall be 
recalled in this section. As a guide to the reader, table \ref{latticetable} in appendix 
\ref{latticeappendix} lists a number of lattices 
that we introduce in this section and that we continue to use throughout this work.  

Let us first fix some standard terminology (see for example \cite{ni80,mo84}):
\begin{definition}
Consider a $d$-dimensional real vector space $V$ with  scalar product 
$\langle\cdot,\cdot\rangle$ of  signature $(\gamma_+,\gamma_-)$.
\begin{enumerate}
\item
A \textsc{lattice} $\Gamma$ in $V$ is a free $\Z$-module $\Gamma\subset V$ together
with the symmetric bilinear form induced by $\langle\cdot,\cdot\rangle$
on $\Gamma$. Its \textsc{discriminant} $\disc(\Gamma)$  
is the determinant of the associated bilinear form on $\Gamma$.  
The lattice $\Gamma$ is \textsc{non-degenerate} if $\disc(\Gamma)\neq0$, and it is 
\textsc{unimodular} if $|\!\disc(\Gamma)|=1$.  
By $\Gamma(N)$ with $N$ an integer we denote the same $\Z$-module as $\Gamma$, 
but with bilinear form rescaled by a factor of $N$.
\item
A lattice $\Gamma$ in $V$ is \textsc{integral} if its bilinear form has 
values in $\Z$ only. 
It is \textsc{even} if the associated quadratic form has values in $2\Z$ only.  
\item
If $\Gamma$ is a non-degenerate even lattice in $V$, then there is a natural 
embedding $\Gamma\hookrightarrow\Gamma^\ast=\Hom(\Gamma,\Z)$ by means 
of the bilinear form on $\Gamma$. Thus there is an induced $\Q$-valued
symmetric bilinear form on $\Gamma^\ast$. The 
\textsc{discriminant form} $q_\Gamma$  of $\Gamma$
is the map $q_\Gamma: \Gamma^\ast/\Gamma\rightarrow\Q/2\Z$ together 
with the symmetric bilinear form on the \textsc{discriminant group}
$\Gamma^\ast/\Gamma$ with values in $\Q/\Z$, 
induced by the quadratic form and symmetric bilinear form on $\Gamma^\ast$, respectively.
\item
A sublattice $\Lambda\subset\Gamma$ of a lattice $\Gamma$ 
in $V$ is a \textsc{primitive sublattice}
if $\Gamma/\Lambda$ is free.
\end{enumerate}
\end{definition}
By this definition, we can view every sublattice $\Lambda\subset\Gamma$ of $\Gamma\subset V$ 
as a lattice in the vector space $\Lambda\otimes\R$ with the induced quadratic form. In fact, 
since $V=\Gamma\otimes\R$ for non-degenerate $\Gamma$, we can dispense with the mention of the
vector space that any of our lattices is constructed in. Note that
$\Lambda$ is a primitive sublattice of $\Gamma$ if and only if
$\Lambda=\left(\Lambda\otimes\Q\right)\cap\Gamma$. 
Moreover,  for every non-degenerate even lattice $\Gamma$ we have
$|\disc(\Gamma)| = |\Gamma^\ast/ \Gamma|$, hence  such
$\Gamma$ is unimodular if and only if $\Gamma=\Gamma^\ast$.
The classification of even unimodular lattices is known to some extent:
\begin{theorem}\label{latticeclassification}
Let $V$ denote  a $d$-dimensional real vector space 
with scalar product $\langle\cdot,\cdot\rangle$ of signature $(\gamma_+,\gamma_-)$
with $d=\gamma_++\gamma_->0$.
\begin{enumerate}
\item
An even unimodular lattice $\Gamma\subset V$ exists if and only if $\gamma_+-\gamma_-\equiv 0\mod8$.
\item 
If $\gamma_+>0$ and $\gamma_->0$ and $\gamma_+-\gamma_-\equiv 0\mod8$,  then up to lattice isometries, there
exists a unique even unimodular lattice 
$\Gamma\subset V$.
\item
If $d=24$ and $\gamma_-=0$, then there are $24$ distinct isometry classes of 
even unimodular lattices
$\Gamma\subset V$. Each of these isometry classes is uniquely determined
by the sublattice
\be\label{rootlattice}
R:=\spann_\Z\left\{\delta\in\Gamma \mid \langle\delta,\delta\rangle =2 \right\}
\ee
which has rank $0$ or $24$ and is isometric to the
root lattice of some semi-simple complex Lie algebra. 
\end{enumerate}
\end{theorem}
Part 1.~and 2.~of the above theorem are known as
\textsc{Milnor's theorem} (see \cite{mi58} and \cite[Ch.~5]{se73}, \cite[Ch.~2]{mihu73} for a proof), 
while the classification 3.~of positive definite even unimodular lattices of rank $24$ is due to 
\textsc{Niemeier} \cite{ni73}, and such lattices are named after him:
\begin{definition}\label{niemeierdef}
A \textsc{Niemeier lattice} is a positive definite even unimodular lattice $\Gamma$
of rank $24$. Its \textsc{root sublattice} \mbox{\rm(\ref{rootlattice})} is the lattice $R\subset\Gamma$ 
that is generated by those elements of $\Gamma$ on which the quadratic form yields the
value $2$. If $R$ is isometric to the root lattice of the semi-simple Lie algebra
$\mathfrak g$, then we say that $\Gamma$ is a \textsc{Niemeier lattice of type $\mathfrak g$}.
\end{definition}
In our applications, we often find ourselves in a situation where a primitive sublattice 
$\Lambda\subset\Gamma$ of an even unimodular lattice $\Gamma$ is well understood, 
and where we need to deduce properties of the lattice $\Gamma$ from those of $\Lambda$. 
In such situations,   \textsc{gluing techniques} as well as \textsc{criteria for primitivity
of sublattices} that were developed by Nikulin in \cite{ni80b,ni80} prove tremendously useful, 
see also \cite{mo84}. We recall these techniques in the next subsection.
%
%
\subsection{Even unimodular lattices from primitive sublattices}\label{glue}
If $\Gamma$ is an even unimodular lattice and $\Lambda\subset\Gamma$ 
is a primitive sublattice, 
then according to \cite[Prop.~1.1]{ni80b} the discriminant forms $q_\Lambda$ and $q_{\m V}$ 
of $\Lambda$ and its orthogonal complement $\m V:=\Lambda^\perp\cap\Gamma$ obey
$q_\Lambda=-q_{\m V}$. Moreover,
\begin{prop}\label{glueperp}
Let $\Gamma$ denote an even unimodular lattice and $\Lambda\subset\Gamma$ 
a non-degenerate primitive sublattice. Then 
the embedding $\Lambda\hookrightarrow\Gamma$ is uniquely determined by an isomorphism 
$\gamma:\Lambda^\ast/\Lambda\rightarrow {\m V}^\ast/{\m V}$, where
$\m V\cong\Lambda^\perp\cap\Gamma$ and 
the discriminant forms  
obey $q_\Lambda=-q_{\m V}\circ\gamma$. Moreover,
\be\label{glueconstruction}
\Gamma\cong\left\{ (\lambda,v)\in \Lambda^\ast\oplus {\m V}^\ast\mid
\gamma(\qu\lambda)=\qu v \right\},
\ee
where for $L$ a non-degenerate even lattice,
$\qu l$ denotes the projection of $l\in L^\ast$  to $L^\ast/L$. 
\end{prop}
Note that \req{glueconstruction} allows us to describe $\Gamma$ entirely by means of its 
sublattices $\Lambda$ and $\Lambda^\perp\cap\Gamma$ along with the isomorphism 
$\gamma$.
\vspace{0.5em}

\begin{example}
Recall the \textsc{hyperbolic lattice}, i.e.\ the even unimodular lattice of signature $(1,1)$  with quadratic form
\be\label{hyperbolic}
\left( \begin{array}{cc} 0&1\\1&0 \end{array}\right)
\ee
with respect to generators $\upsilon_0,\,\upsilon$ over $\Z$. 
We generally denote the hyperbolic lattice by $U$. 
As a useful exercise, the reader should convince herself that  the gluing procedure described 
above allows a reconstruction of $U$ from the two definite sublattices $A^\pm$ generated by 
$a_\pm:=\upsilon_0\pm\upsilon$, respectively, with
$$
U = \left\{ \left.{\textstyle{1\over2}}(n_+ a_+ + n_-a_-) \right| n_\pm\in\Z,\; n_+ + n_-\in 2\Z\right\}
$$
by \req{glueconstruction}.
\end{example}

To apply proposition \ref{glueperp} one needs to know that 
the lattice $\Lambda$ can be primitively embedded in some even, unimodular lattice. 
Nikulin has also developed powerful techniques which determine
whether or not this is the case for abstract non-degenerate lattices 
$\Gamma$ and $\Lambda$. To recall this, we need some 
additional terminology:
\begin{definition}
Consider a finite abelian group $A$. 
The minimum number of generators of $A$ is called 
the \textsc{length $\ell(A)$ of $A$}. 
If
$A=A_q$ carries a quadratic form $q\colon A_q\rightarrow \Q/2\Z$,
then for every prime $p$ let $A_{q_p}$ denote the
\textsc{Sylow $p$-group in $A_q$}, that is, 
$A_{q_p}\subset A_q$ is the maximal subgroup  whose order is a power of $p$.
Then by  $q_p$ we denote the restriction of $q$ to $A_{q_p}$. 
\end{definition}
Note that for every finite abelian group $A$ and for every prime $p$, there exists
a unique Sylow $p$-group in $A$. Moreover, if
$A=A_q$ carries a quadratic form $q$ with values in $\Q/2\Z$, then  $q$ decomposes into an
orthogonal direct sum of all the $q_p$ with prime $p$. 
As a special case of \cite[Thm.~1.12.2]{ni80} used in
precisely this form in Kondo's proof of \cite[Lemma 5]{ko98}, we now have:
\begin{theorem}\label{nik}
Let $\Lambda$ denote a non-degenerate even lattice of signature $(l_+,$ $l_-)$ and
$A_q=\Lambda^\ast/\Lambda$ its discriminant group, 
equipped with the induced quadratic form $q$ with values in $\Q/2\Z$.
Furthermore assume that all the following conditions hold
for the integers $\gamma_+,\gamma_-$ with $\gamma_+\equiv \gamma_-\mod 8$:
\begin{enumerate}
\item $\gamma_+\geq l_+$ and  $\quad \gamma_-\geq l_-$,
\item  $\ell(A_{q})\leq \gamma_++\gamma_--l_+-l_-$,
\item  i. $q_2=({1\over2}) \oplus q_2^\prime$, where $({1\over2})$ is the quadratic form
on the discriminant of a root lattice $\Z(2)$ of type $A_1$ and $q_2^\prime$ is an arbitrary quadratic form
with values in $\Q/2\Z$

or 
ii. $\ell(A_{q_2})< \gamma_++\gamma_--l_+-l_-$,
\item for every prime $p\neq2$, $\ell(A_{q_p})< \gamma_++\gamma_--l_+-l_-$.
\end{enumerate}
Then $\Lambda$ can be primitively embedded in some even unimodular lattice $\Gamma$
of signature $(\gamma_+,\gamma_-)$.
\end{theorem}
%
\subsection{Example: The K3-lattice for Kummer surfaces}\label{exampleK3}
A classical application of the gluing technique summarized in proposition \ref{glueperp} is the description 
of the integral homology of a Kummer surface in terms of the integral homology of its underlying torus 
and the contributions  from the blow-up of singularities. We review this
construction in the present subsection, 
following \cite{pss71,ni75}. 
First recall the definition of K3 surfaces:
\begin{definition}\label{k3definition}
A \textsc{K3 surface} is a compact complex  
surface with trivial canonical bundle and vanishing first Betti number.
\end{definition}
By a seminal result of Kodaira's \cite[Thm. 19]{ko64}, 
as real four-manifolds all K3 surfaces are diffeomorphic. 
The integral homology of a K3 surface $X$ has the following properties 
(see e.g.\ \cite{mi58,sa65}): First,
$H_2(X,\Z)$ has no torsion,
and equipped with the intersection form it is
an even  unimodular lattice of signature $(3,19)$. 
By theorem \ref{latticeclassification} this means that $H_2(X,\Z)$ is uniquely determined
up to isometry. In fact, $H_2(X,\Z)\cong U^3\oplus E_8^2(-1)$, where $U$ is the
hyperbolic lattice with quadratic form \req{hyperbolic} and
$E_8$ is isometric to
the root lattice of the Lie algebra $\mathfrak e_8$. The lattice $U^3\oplus E_8^2(-1)$
 is often called the \textsc{K3-lattice}.
Furthermore, $H_\ast(X,\Z)= H_0(X,\Z)\oplus H_2(X,\Z)\oplus H_4(X,\Z)$ is an even 
unimodular lattice of signature $(4,20)$ and thus $H_\ast(X,\Z)\cong U^4\oplus E_8^2(-1)$
by theorem \ref{latticeclassification} and the above.
A choice of isometries $H_2(X,\Z)\cong U^3\oplus E_8^2(-1)$,
$H_\ast(X,\Z)\cong U^4\oplus E_8^2(-1)$ is called a \textsc{marking} of our K3 surface $X$.
{Such a marking is specified by fixing a standard basis of $H_2(X,\Z)$, $H_\ast(X,\Z)$,
which does not necessarily need to exhibit the structure of $U^3\oplus E_8^2(-1)$,
$U^4\oplus E_8^2(-1)$ in the first instance.}

Let us now recall the \textsc{Kummer construction}, which yields a K3 surface by
a $\mathbb{Z}_2$-orbifold construction from every
complex torus $T$ of dimension $2$. 
Let $T=T(\Lambda)=\mathbb{C}^2/\Lambda$, with $\Lambda\subset\C^2$ a lattice of 
rank $4$ over $\Z$, {and with generators $\vec{\lambda_i}, i=1,\ldots,4$}. The group $\Z_2$ acts 
naturally on $\C^2$ by $(z_1,z_2)\mapsto (-z_1,-z_2)$ and thereby on $T(\Lambda)$. Using Euclidean 
coordinates $\vec{x}=(x_1,x_2,x_3,x_4)$, where $z_1=x_1+ix_2$ and $z_2=x_3+ix_4$, points on the 
quotient $T(\Lambda)/\Z_2$ are identified according to
\ba
\vec{x}&\sim& \vec{x}+\sum_{i=1}^4n_i {\vec\lambda_i},\quad n_i \in \mathbb{Z},\nonumber\\
\vec{x}&\sim&-\vec{x}.\nonumber
\ea
Hence $T(\Lambda)/\Z_2$ has $16$ singularities of type $A_1$, located at the fixed points of 
the $\Z_2$-action. These fixed points are conveniently labelled by the \textsc{hypercube} 
$\mathbb {F}_2^4\cong{1\over2}\Lambda/\Lambda$, where $\mathbb {F}_2=\{0,1\}$ is the finite field with two elements, as
\be\label{labels}
\vec{F}_{\vec a}:=\left[{\textstyle\hf} \sum_{i=1}^4a_i \vec{\lambda_i}\right]\in T(\Lambda)/\Z_2,\quad
\vec{a}=(a_1,a_2,a_3,a_4) \in \mathbb {F}_2^4.
\ee
It is known that the complex surface obtained by minimally
resolving $T(\Lambda)/\Z_2$ is a K3 surface
(see e.g.\ \cite{ni75}):
\begin{definition}\label{defKummer}
Let $T=T(\Lambda)$ denote a complex torus of dimension $2$ with underlying lattice
$\Lambda\subset\C^2$. The
\textsc{Kummer surface with underlying torus $T$} is the complex surface $X$
obtained by  minimally resolving each of the singularities 
$\vec F_{\vec a},\, \vec a\in\F_2^4\cong{1\over2}\Lambda/\Lambda$,
of $T/\Z_2$ given in \mbox{\rm\req{labels}}. The blow-up of $\vec F_{\vec a}$ yields
an exceptional divisor whose class in $H_2(X,\Z)$  we denote by $E_{\vec a}$. 
The natural rational map of degree $2$ from $T$ to $X$, which is defined outside the fixed points of 
$\Z_2$ on $T$, is denoted $\pi\colon T\dashrightarrow X$, and  the  linear map it induces on homology
is denoted $\pi_\ast\colon H_2(T,\Z)\rightarrow H_2(X,\Z)$,
$\pi_\ast\colon H_2(T,\R)\rightarrow H_2(X,\R)$, respectively.
\end{definition}
Note that by definition, the Kummer surface $X$ with underlying torus $T(\Lambda)=\C^2/\Lambda$
carries the  complex structure which is induced from the universal cover $\C^2$ of $T(\Lambda)$. Moreover,
the integral classes $E_{\vec a}\in H_2(X,\Z)$ with $\vec a\in\F_2^4$ arise from blowing up
singularities of type $A_1$ and thus represent rational two--cycles on $X$. Hence by construction,
they generate a sublattice of $H_2(X,\Z)$ of type $A_1^{16}(-1)$, i.e.\ a sublattice of rank $16$ with 
quadratic form $\diag(-2,\ldots,-2)$. This lattice, however, is not primitively embedded
in $H_2(X,\Z)$. Instead, by \cite{pss71,ni75} we have
\begin{prop}\label{Kummerform}
Let $X$ denote a Kummer surface with underlying torus $T$ as in definition 
\mbox{\rm\ref{defKummer}}.
Furthermore, let $\Pi$ denote the smallest primitive sublattice of $H_2(X,\Z)$ containing 
all the $E_{\vec a}, \vec a\in\F_2^4$, and let $K:=\pi_\ast (H_2(T,\Z))$. Then the following holds:
\begin{enumerate}
\item
The  lattice $\Pi$ is the \textsc{Kummer lattice}
\begin{equation}\label{Kummer}
\Pi = \spann_\Z\left\{ E_{\vec a}\mbox{ with } \vec a\in\F_2^4; \;
{\textstyle\hf} \sum_{\vec a\in H} E_{\vec a} \mbox{ with } H\subset \F_2^4 \mbox{ an affine hyperplane} \right\}.
\end{equation}
The embedding of $\Pi$ in $H_2(X,\Z)$ is unique up to automorphisms of $H_2(X,\Z)$.
\item
The lattice $K$ is a primitive sublattice of $H_2(X,\Z)$ with $K \cong U^3(2)$
and $K=\Pi^\perp\cap H_2(X,\Z)$.
\end{enumerate}
\end{prop}
Note that $K \cong U^3(2)$ follows by construction, since $H_2(T,\Z)\cong U^3$ and
$\pi$ has degree $2$. Moreover, each element of $K$ 
represents the image of a two--cycle on the torus which is in general 
position and thus does not contain any fixed points of the $\Z_2$-action, hence
$K\perp\Pi$.
Then $K=\Pi^\perp\cap H_2(X,\Z)$ since $K$ is a primitive sublattice of
$\Pi^\perp\cap H_2(X,\Z)$ of same rank $6=22-16$.

For the remainder of this work, we continue to use the notations introduced
in definition \ref{defKummer} and proposition \ref{Kummerform} above. 
According to this proposition, $K$ and $\Pi$ are orthogonal complements of one another
in the even unimodular lattice $H_2(X,\Z)$.  
By the gluing construction of proposition \ref{glueperp}, $H_2(X,\Z)$ can therefore
be reconstructed from its sublattices $K$ and $\Pi$. 

Indeed, one first checks $K^\ast/K \cong \Pi^\ast/\Pi \cong (\Z_2)^6$:
With\footnote{The generators $\vec\lambda_i$, $i=1,\ldots,4$, of the lattice $\Lambda$ are 
naturally identified with generators $\lambda_i$, $i=1,\ldots,4$, of $H_1(T,\Z)$, such that 
$H_2(T,\Z)$ is generated by the $\lambda_i\vee\lambda_j$.}
\be\label{lambdaij}
\lambda_{ij}:=\lambda_i\vee\lambda_j \in H_2(T,\Z)\mbox{ for } i,j\in\{1,2,3,4\},
\ee
standard generators of
$K^\ast/K$ and the discriminant form with respect to these generators are given by
$$
\textstyle
\qu{{1\over2}\pi_\ast\lambda_{ij}},\;
ij=12,\, 34,\, 13,\, 24,\, 14,\, 23,\quad
q_K = \left( \begin{array}{cc}0&\hf\\\hf&0\end{array}\right)^3.
$$
Analogously, with
\be\label{planes}
P_{ij} := \left\{ \vec a=(a_1,a_2,a_3,a_4)\in\F_2^4 \mid a_k=0\; \fa k\neq i,j\right\}
\mbox{ for } i,j\in\{1,2,3,4\},
\ee
standard generators of $\Pi^\ast/\Pi$ and the discriminant form with respect to these generators are given by
$$
\textstyle\qu{\hf\sum\limits_{\vec a\in P_{ij}}E_{\vec a}},\;
ij=12,\, 34,\, 13,\, 24,\, 14,\, 23,\quad
q_\Pi = -\left( \begin{array}{cc}0&\hf\\\hf&0\end{array}\right)^3.
$$
The bilinear forms associated with $q_K,\, q_\Pi$ take values in $\Q/\Z$ 
and thus $q_K=-q_K$, $q_\Pi=-q_\Pi$. Hence we obtain a natural isomorphism $\gamma$ 
between the two discriminant groups which  obeys $q_K=-q_\Pi\circ\gamma$:
\be\label{Kummerglue}
\gamma\colon K^\ast/K \longrightarrow \Pi^\ast/\Pi,\quad
\gamma\left(\textstyle\qu{{1\over 2}\pi_\ast\lambda_{ij}}\right)
:=\textstyle\qu{\hf\sum\limits_{\vec a\in P_{ij}}E_{\vec a}}.\\
\ee
Now proposition \ref{glueperp} implies that the K3-lattice $H_2(X,\Z)$ is generated by the
$\pi_\ast\lambda_{ij}\in\pi_\ast (H_2(T,\Z))$, the elements of the Kummer lattice $\Pi$, and
two--cycles of type ${1\over2}\pi_\ast\lambda_{ij}+\hf\sum_{\vec a\in P_{ij}}E_{\vec a}\,\in K^*\oplus \Pi^*$.
In this case, the gluing procedure can be visualised geometrically as follows: Consider a real 
$2$-dimensional subspace of $\C^2$ which on the torus $T$ yields a 
$\Z_2$-invariant submanifold $\kappa$ 
containing the four fixed points labelled by a plane $P\subset\F_2^4$. Then $\kappa\rightarrow\kappa/\Z_2$ 
is a $2\colon\!\!1$ cover of an $\mathbb S^2$ with branch points $\vec F_{\vec a},\, \vec a\in P$, which under
blow-up are replaced by the corresponding exceptional divisors representing the 
$E_{\vec a}\in H_2(X,\Z)$.
Hence {$\left(\kappa\setminus\{\vec F_{\vec a}\mid \vec a\in P\}\right)$} is a $2\colon\!\!1$ unbranched covering of a 
two--cycle on the Kummer surface $X$ representing $\pi_\ast[\kappa]-\sum_{\vec a\in P}E_{\vec a}\in H_2(X,\Z)$. 
In other words,
$\hf \pi_\ast[\kappa]\mp\hf\sum_{\vec a\in P}E_{\vec a}\in H_2(X,\Z)$.
Indeed, note that for $P$ as above and $P^\prime\subset\F_2^4$ a plane parallel to $P$,
$\hf\sum_{\vec a\in P}E_{\vec a}\mp\hf\sum_{\vec a\in P^\prime}E_{\vec a}\in\Pi$ according to \req{Kummer}.

For later use, instead of restricting our attention to the K3-lattice $H_2(X,\Z)$ of a Kummer surface $X$, 
we need to work on the full integral homology $H_\ast(X,\Z)=H_0(X,\Z)\oplus H_2(X,\Z)\oplus H_4(X,\Z)$. 
Since $H_0(X,\Z)\oplus H_4(X,\Z)\cong U$ is an even unimodular lattice, 
we can use the gluing prescription (\ref{Kummerglue}) either replacing $K$ by $K\oplus U$ with 
$(K\oplus U)^\ast/(K\oplus U)\cong K^\ast/K$, or replacing $\Pi$ by $\Pi\oplus U$ with 
$(\Pi\oplus U)^\ast/(\Pi\oplus U)\cong \Pi^\ast/\Pi$. However, yet another option will turn out to 
be even more useful: We combine the gluing prescription (\ref{Kummerglue}) with the exercise 
posed in section \ref{glue} to obtain, as another application of proposition \ref{glueperp},
\begin{prop}\label{fullk3}
Consider a Kummer surface $X$ with underlying torus $T$ and notations as in 
definition \mbox{\rm\ref{defKummer}} and proposition \mbox{\rm\ref{Kummerform}}. Furthermore, let
$\upsilon_0\in H_0(X,\Z)$ and $\upsilon\in H_4(X,\Z)$ denote generators of
$H_0(X,\Z)\oplus H_4(X,\Z)\cong U$ such that the quadratic form 
with respect to $\upsilon_0, \upsilon$ is \mbox{\rm(\ref{hyperbolic})}. Let
$$
\m K:=K\oplus\spann_\Z\{\upsilon_0+\upsilon\}, \quad
\m P:=\Pi\oplus\spann_\Z\{\upsilon_0-\upsilon\}.
$$
Then $\m K^\ast/\m K\cong \m P^\ast/\m P\cong(\Z_2)^7$ under an isomorphism $g$ with
\begin{eqnarray*}
&g\colon& \m K^\ast/\m K\longrightarrow \m P^\ast/\m P, \quad
g(\qu\kappa):=\gamma(\qu\kappa) \;\fa\kappa\in K^\ast,\quad
g(\qu{\textstyle\hf (\upsilon_0+\upsilon)}) := \qu{\textstyle\hf (\upsilon_0-\upsilon)};\nonumber\\[10pt]
&&H_\ast(X,\Z)
\cong \left\{ (\kappa,\pi)\in \m K^\ast\oplus {\m P}^\ast\mid
g(\qu\kappa)=\qu \pi \right\}.
\end{eqnarray*}
\end{prop}
By the above, the Kummer construction yields a natural marking $H_2(X,\Z)\cong U^3\oplus E_8^2(-1)$,
$H_\ast(X,\Z)\cong U^4\oplus E_8^2(-1)$. Indeed, the $\lambda_{ij}$, $i, j\in \{1,\ldots,4\}$, $i < j$,
form a basis
of $H_2(T,\Z)$, where each of the pairs $\{\lambda_{12},\lambda_{34}\}, \{\lambda_{13},\lambda_{24}\}, 
\{\lambda_{14},\lambda_{23}\}$ generates a sublattice which is isometric to the
hyperbolic lattice $U$. Furthermore, \req{Kummer} gives an
abstract construction of the Kummer lattice $\Pi$ in terms of the roots $E_{\vec a}, \vec a\in\F_2^4$, with
$\langle E_{\vec a},E_{\vec b}\rangle=-2\delta_{\vec a,\vec b}$. Then the above gluing prescription 
of $H_2(X,\Z)$ from the lattices $\pi_\ast H_2(T,\Z)$ and $\Pi$ specifies an isometry between $H_2(X,\Z)$
for $X=\wt{T/\Z_2}$ with a standard even, unimodular
lattice of signature $(3,19)$ (though not written in the form $U^3\oplus E_8^2(-1)$).
Denoting generators of $H_0(X,\Z)$ and $H_4(X,\Z)$ by $\upsilon_0,\, \upsilon$, respectively, where
$\langle\upsilon_0,\upsilon\rangle=1$ as in proposition \ref{fullk3}, we obtain 
$H_\ast(X,\Z)\cong U^4\oplus E_8^2(-1)$.

In all examples studied below, we use this fixed marking, as it is natural for all Kummer surfaces.
Note that the marking allows us to smoothly vary the {generators $\vec\lambda_1,\ldots,\vec\lambda_4$
of the}
defining lattice $\Lambda\subset\C^2$ of $T=T(\Lambda)$ 
for $X=\wt{T(\Lambda)/\Z_2}$; the marking is thus compatible with the deformation of any Kummer
surface into any other one. 
%
\subsection{Example: The Niemeier lattice of type $A_1^{24}$}\label{niemeier}
A second example for the application of Nikulin's gluing techniques from proposition \ref{glueperp}, 
which we find extremely useful, involves the Niemeier lattice $N$ of type 
$A_1^{24}$ (see  theorem \ref{latticeclassification} and definition \ref{niemeierdef}). In other words, 
$N$ is the Niemeier lattice with root sublattice $R\subset N$ of rank $24$, where $R$  
has quadratic form $\diag(2,\ldots,2)$. 
Since $N$ is unimodular, we
have
$$
R\subset N\subset R^\ast,
$$
where $R^\ast={1\over2}R$ and 
thus $R^\ast/R\cong\F_2^{24}$. Hence $N/R$ can be viewed as a subspace of $\F_2^{24}$,
and in fact $N/R\cong\m G_{24}$, the \textsc{extended binary Golay code} 
\cite[Ch.\ 16, 18]{cosl88}. 
Up to isometry, the extended binary Golay code is uniquely determined by the fact that it is
a $12$-dimensional subspace of $\F_2^{24}$ over $\F_2$ such that every $v\in\m G_{24}$
has \textsc{weight}\footnote{that is, the number of non-zero entries}  zero, $8$ (\textsc{octad}), 
$12$ (\textsc{dodecad}), $16$ (\textsc{complement octad}), or $24$. 
For further details concerning the extended binary Golay code, which for brevity we simply 
call the \textsc{Golay code} from now on, see  {appendix \ref{MOG}}. 

In terms of the Golay code $\m G_{24}\subset\F_2^{24}$, the Niemeier lattice $N$ can be constructed
from its root sublattice $R$ as a sublattice of $R^\ast={1\over2}R$:
\begin{prop}\label{niefromroot}
Consider the Niemeier lattice of type $A_1^{24}$, which for the remainder of this work is
denoted by $N$, and its root sublattice $R$.  Then
$$
N = \left\{ v\in R^\ast \mid \qu v\in \m G_{24}\right\},
$$
where $\qu v$ denotes the projection of $v\in R^\ast$  to $\F_2^{24}\cong R^\ast/R$.
\end{prop}
As stated in theorem \ref{latticeclassification}, every Niemeier lattice $\wt N$ is
uniquely determined by its root sublattice $\wt R$ up to isometry. In fact, if $\wt R$ has rank $24$, then
$\wt N$ can always be constructed from $\wt R$ analogously to proposition \ref{niefromroot}.
In this paper, the Niemeier lattice $N$ of type $A_1^{24}$ plays a special role, though,
since the \textsc{Mathieu group $M_{24}$} acts so naturally on it\footnote{By giving an extensive list of examples,
in \cite{ni11} Nikulin emphasizes that all $24$ Niemeier lattices are important for the study
of K3 surfaces; while this may be true, our objective is the clarification of the role
of the Mathieu group $M_{24}$, justifying our preference for the lattice $N$.}. Namely, $M_{24}$ is the automorphism group
of the Golay code \cite{tod59,tod66,co71}, and thus we determine the automorphisms of $N$ in accord
with \cite[Ch.~4.3, 16.1]{cosl88}:
\begin{prop}\label{M24onN}
Consider the Niemeier lattice $N$ of type $A_1^{24}$ 
with root sublattice $R$
generated by the roots $f_1,\ldots, f_{24}$. 
The automorphism group of $N$ is 
$$
\Aut(N)=(\Z_2)^{24}\rtimes M_{24}, 
$$
where the action of the $n^{th}$ factor of $(\Z_2)^{24}$
is induced by $f_n\mapsto -f_n$, while the 
\textsc{Mathieu group $M_{24}$} is viewed as a subgroup
of the permutation group $S_{24}$ on $24$ elements whose action on $N$ is
induced by permuting the roots $f_1,\ldots, f_{24}$.
\end{prop}
\begin{proof}
Since every non-trivial codeword in the Golay code has at least weight $8$, 
from proposition \ref{niefromroot} we see that
the set $\Delta$ of roots in $N$ is
\be\label{nieroot}
\Delta=\left\{ \delta\in N\mid \langle\delta,\delta\rangle=2\right\}
=\left\{\pm f_1, \ldots, \pm f_{24}\right\}. 
\ee
Hence every lattice isometry $\gamma\in\Aut(N)$ acts as
a permutation on $\Delta$. Proposition \ref{niefromroot} implies
$N\subset\hf R$ and thus that $\gamma$
is uniquely determined by its action on $\Delta$. Thus $\gamma=\iota\circ\alpha$ 
with $\iota\in (\Z_2)^{24}\subset\Aut(N)$ induced by
a composition of involutions $f_n\mapsto -f_n$, and with
$\alpha\in \Aut(N)$ induced by a permutation of $f_1,\ldots,f_{24}$. 
Now proposition \ref{niefromroot}
implies that $\alpha$ induces an action on $\F_2^{24}\cong R^\ast/R$ by permuting
the binary coordinates of $\F_2^{24}$, which must leave the Golay code 
$\m G_{24}\subset \F_2^{24}$
invariant. In other
words, $\alpha\in M_{24}$. Vice versa, every involution 
$f_n\mapsto -f_n$ and every permutation of the binary coordinates of $\F_2^{24}$
which preserves the Golay code induces an automorphism of $N$.
From this the claim follows.
\end{proof}

The above proposition allows us to view the Niemeier lattice $N$ 
as a device which yields a geometric interpretation of 
the Mathieu group $M_{24}$. On the other hand, this lattice turns out to share a 
number of properties
with the integral homology of Kummer surfaces that we discussed in section \ref{exampleK3}. 
First we observe
\begin{prop}\label{piprim}
The Niemeier lattice $N$ of type $A_1^{24}$ possesses a primitive sublattice
$\wt\Pi\subset N$ which up to a total reversal of signature is isometric to
the Kummer lattice  \mbox{\rm\req{Kummer}}.

No other Niemeier lattice possesses such a sublattice.
\end{prop}
\begin{proof}
Given the Kummer lattice $\Pi$ of \req{Kummer}, consider the lattice $\wt\Pi:=\Pi(-1)$.
We first show that $\wt\Pi$ can be primitively embedded in some Niemeier lattice
$\widetilde N$ by applying theorem \ref{nik} to the lattice $\Lambda=\wt\Pi$.
 
Recall  from definition \ref{niemeierdef}
that a Niemeier lattice is an even unimodular lattice of signature 
$(\gamma_+,\gamma_-)=(24,0)$. Furthermore, by proposition \ref{Kummerform} and
the explanations following it, the lattice $\wt\Pi$ is a non-degenerate even
lattice with signature
$(l_+,l_-)=(16,0)$ and discriminant group $A_q=\wt\Pi^\ast/\wt\Pi\cong(\Z_2)^6$.
In particular, the length of this group is $\ell(A_q)=6$, and for the reductions
modulo primes we find $\ell(A_{q_2})=6$ and $\ell(A_{q_p})=1$ for all primes $p\neq2$.
Since $\gamma_++\gamma_--l_+-l_-=\gamma_+-l_+=8>6>1$, 
the conditions 1., 2., 3.ii.~and 4.~of theorem \ref{nik} hold, and thus $\widetilde\Pi$
can be primitively embedded in some Niemeier lattice $\widetilde N$.

For every Niemeier lattice $\widetilde N$, the root lattice $\widetilde R$
is known  (see e.g.\ {\cite[Ch.\ 16.1]{cosl88}}), and one checks that the 
Niemeier lattice $N$ of type $A_1^{24}$ is the only one containing $16$
pairwise perpendicular roots that generate a sublattice $R$ of type $A_1^{16}$,
such that $(R\otimes\Q)\cap N$ contains no further
roots,
as $\widetilde\Pi$ does. From this
the claim follows. 
\end{proof}

In \cite[Case 23]{ni11}, Nikulin remarks that $N$ is the only Niemeier lattice which can contain
a primitive sublattice that up to a total signature reversal is isomorphic to the Kummer 
lattice.\footnote{Nikulin however fails to mention that the existence
of a primitive sublattice $\wt\Pi$ in $N$ as in proposition \ref{piprim}
was first observed and proved in our first
installment arXiv:1008.0954 of this work. In fact he fails to observe that existence needs to be 
proved, whatsoever.} 

For our investigations, we need an \textsc{explicit realisation of a primitive sublattice}
$\wt\Pi$ as in proposition \ref{piprim}. Our construction of the lattice $\wt\Pi$ 
and its orthogonal complement $\wt K$ in $N$ depends on the choice of an
arbitrary \textsc{special octad} in the Golay code, i.e.\ a vector of weight $8$ 
in $\m G_{24}$. For historical reasons, our choice of special octad is 
the codeword $o_9:=(0,0,1,0,1,1,0,0,1,0,0,0,0,0,1,0,0,0,1,0,0,0,1,1)\in\m G_{24}$
whose non-zero entries are the binary coordinates with labels 
$\left\{ 3,\,5,\,6,\,9,\,15,\,19,\,23,\,24\right\}\subset\m I$
with $\m I:=\left\{1,\ldots,24\right\}$.
For ease of notation, we regularly denote 
a codeword $v\in \m G_{24}\subset\F_2^{24}$ of the Golay code by listing the set 
$A_v\subset\m I$ of coordinate labels with non-zero entries. With this 
notation, calculating the sum of codewords 
$v,\, w\in\m G_{24}\subset\F_2^{24}$ amounts to taking the \textsc{symmetric 
difference} of sets $A_v+A_w=(A_v\setminus A_w)\cup (A_w\setminus A_v)$.
Our special octad $o_9$ is thus described by
\ba\label{o9octad}
\m O_9 &:=& \left\{ 3,\,5,\,6,\,9,\,15,\,19,\,23,\,24\right\}, \\
o_9 &=& (0,0,1,0,1,1,0,0,1,0,0,0,0,0,1,0,0,0,1,0,0,0,1,1)\in\m G_{24}\subset\F_2^{24}.\nonumber
\ea
This is the octad corresponding to the standard MOG configuration
described in appendix \ref{MOG}, where the two first 
columns have entries $1$, and the others are $0$.
We now claim the following
\begin{prop}\label{piinn}
Consider the Niemeier lattice $N$ of type $A_1^{24}$ with root sublattice $R$ generated
by roots $f_1,\ldots, f_{24}$. With our choice \mbox{\rm\req{o9octad}} of a special octad 
$\m O_9$ let
\be \label{specialchoice}
\wt K := \left\{ \nu\in N\mid  \fa n\not\in \m O_9\colon \langle \nu, f_n\rangle=0 \right\}, \quad
\wt\Pi := \left\{ \nu\in N\mid \fa n\in \m O_9\colon \langle \nu, f_n\rangle=0 \right\}.
\ee
Then $\wt K$ and $\wt\Pi$ are primitive sublattices of $N$ which are orthogonal complements
of one another. 
{Moreover}, $\wt\Pi(-1)$ is isometric to the Kummer lattice $\Pi$ as in \mbox{\rm\req{Kummer}}.
\end{prop}
\begin{proof}
That $\wt K$ and $\wt\Pi$ are primitive sublattices of $N$
follows immediately from \req{specialchoice}. For $n\in\m I$, by construction, $f_n\in\wt K$ if 
and only if $n\in\m O_9$, while $f_n\in\wt\Pi$ if and only if $n\not\in\m O_9$. 
Thus $\wt K$ has rank at least $8$ and $\wt\Pi$ has rank at least $16$ 
with $\wt K^\perp\cap N\subset\wt\Pi$, $\wt\Pi^\perp\cap N\subset\wt K$. Thus $8+16=24=\rk(N)$
implies that these lattices are orthogonal complements of one another.

To prove that $\wt\Pi(-1)$ is isometric to the Kummer lattice $\Pi$ given in \req{Kummer}, first 
observe that there is a $5$-dimensional subspace of the Golay code 
$\m G_{24}$, defined as the space of all those codewords which have 
no intersection with the octad $\m O_9$. A basis of this space is
\be \label{hyperplanes}
\begin{array}{rcl}
{\m H}_1&:=&\{1,2,4,12,13,14,17,18\},\\[5pt]
{\m H}_2&:=&\{1,2,8,11,14,16,17,22\},\\[5pt]
{\m H}_3&:=&\{1,8,10,11,13,17,18,21\},\\[5pt]
{\m H}_4&:=&\{1,4,11,13,14,16,20,21\},\\[5pt]
{\m H}_5&:=&\{2, 7, 8, 10,12, 17, 18,  22\}.
\end{array}
\ee
Hence we have
\be\label{GolayKummer}
\wt\Pi = \spann_\Z\left\{ f_n,\, n\not\in\m O_9;\;
\textstyle\hf\sum\limits_{n\in\m H_i} f_n,\, i=1,\ldots,5\right\}.
\ee
Now consider the map 
$I\colon\m I\setminus\m O_9\longrightarrow\F_2^4$ with
\be
\label{map}
I\colon\left\{
 \begin{array}{rcrcrcr}
1 \mapsto (0, 0, 0, 0),&&8 \mapsto (1, 0, 0, 1),&&13 \mapsto (0, 1, 0, 0),&&18 \mapsto (0, 1, 0, 1),\\
2 \mapsto (0, 0, 1, 1),&&10 \mapsto (1, 1, 0, 1),&&14 \mapsto (0, 0, 1, 0),&&20 \mapsto (1, 1, 1, 0),\\
4 \mapsto (0, 1, 1, 0),&&11 \mapsto (1, 0, 0, 0),&&16 \mapsto (1, 0, 1, 0),&&21 \mapsto (1, 1, 0, 0),\\
7 \mapsto (1, 1, 1, 1),&&12 \mapsto (0, 1, 1, 1),&&17 \mapsto (0, 0, 0, 1),&&22 \mapsto (1, 0, 1, 1).
\end{array}\quad
\right.
\ee
Under this map, the elements of $\m H_i$ with $i=1,\ldots,4$ correspond to the 
hypercube points $\vec a=(a_1,a_2,a_3,a_4)\in\F_2^4$ with $a_i=0$, while 
the hypercube points
corresponding to elements of $\m H_5$ are obtained from those corresponding to 
$\m H_4$ by a shift by $(1,1,1,1)\in\F_2^4$. In other words, in terms of the hypercube 
labels, each $\m H_i$ contains the labels corresponding to an affine hyperplane 
$H_i\subset\F_2^{4}$, such that every hyperplane in $\F_2^{4}$ can be 
obtained from $H_1,\ldots, H_5$ by means of symmetric differences. 
Hence (\ref{Kummer}) and (\ref{GolayKummer}) show that the map 
$f_n\mapsto E_{I(n)}$ for $n\not\in\m O_9$ induces an isometry of lattices
$\wt\Pi(-1)\longrightarrow\Pi$. 
\end{proof}

To our knowledge, the relation of the Golay code to the Kummer lattice 
found in proposition \ref{piinn} is a new 
observation. It is certainly crucial for our analysis below.
Our map $I$ in \req{map} induces the structure of a $4$-dimensional
vector space over $\F_2$
on the $16$ labels of the Golay code in $\m I\setminus\m O_9$. 
This linear structure
is known to group theorists\footnote{as we learned after the previous installment
of our work arXiv:1107.3834}: In \cite[Thm.~2.10]{co71}, Conway proves 
the previously known fact (see e.g. \cite{jo70}) that the general
linear group $L_4(2)=\PSL(4,\F_2)=\GL(4,\F_2)$ of $\F_2^4$ is isomorphic to the permutation group
$A_8$ of even permutations on $8$ elements. In our notations,
the crucial step in the proof identifies the action of the stabilizer subgroup
of $\m O_9\cup\{7\}$ in $M_{24}$ with the action of a subgroup of 
$L_4(2)$ on the
$4$-dimensional subspace of $\m G_{24}$ generated by the octads
$\m H_1,\ldots,\m H_4$ as in \req{hyperplanes}. By our construction, this vector space is the
dual of our hypercube $\F_2^4$ built on $\m I\setminus\m O_9$. The vector space 
structure of the latter,
to our knowledge, is first mentioned by Curtis in \cite{cu89}. Hence altogether our
proposition \ref{piinn} gives a novel geometric meaning  in terms of Kummer geometry
to the known vector space
structure on $\m I\setminus\m O_9$.

Note that both lattices $\wt K$ and $\wt\Pi$ in proposition \ref{piinn}
are contained in the $\Q$-span of their root sublattices. 
The gluing techniques of proposition \ref{glueperp} apply and allow us to 
reconstruct the Niemeier lattice $N$ from these lattices.
Indeed, first note that $\wt K^\ast/\wt K\cong \wt\Pi^\ast/\wt \Pi\cong(\Z_2)^6$ 
with associated discriminant forms obeying $q_{\wt K}=-q_{\wt\Pi}$. Namely, 
as representatives $q_{ij}\in\wt K^\ast$ of a minimal set of generators of 
$\wt K^\ast/\wt K$ we identify, for example,
\be\label{kijchoice}
\begin{array}{rclrcl}
q_{12}&:=&\hf\left( f_3+f_6-f_{15}-f_{19}\right),\;&
q_{34}&:=&\hf\left(f_6+f_9-f_{15}-f_{19}\right),\\[5pt]
q_{13}&:=&\hf\left( -f_6+f_{15}-f_{23}+f_{24}\right),\; &
q_{24}&:=&\hf\left( -f_{15}+f_{19}+f_{23}-f_{24}\right),\\[5pt]
q_{14}&:=&\hf\left( f_3-f_{9}-f_{15}+f_{24}\right),\;&
q_{23}&:=&\hf\left( f_3-f_{9}-f_{15}+f_{23}\right),\;
\end{array}
\ee
where the choices of signs at this stage are arbitrary but will come useful later on. 
The resulting quadratic form is thus calculated to
\be\label{quadform}
q_{\wt K} = \left( \begin{array}{cc}0&\hf\\\hf&0\end{array}\right)^3
\ee
with the associated bilinear form taking values in $\Q/\Z$. An analogous analysis yields 
representatives $p_{ij}\in\wt\Pi^\ast$ of generators  of $\wt\Pi^\ast/\wt \Pi$ which are 
glued to the $q_{ij}$ under an appropriate isomorphism 
\be\label{Niemeierglue}
\wt\gamma\colon \wt K^\ast/\wt K\longrightarrow \wt\Pi^\ast/\wt \Pi, 
\quad\wt\gamma(\qu{q_{ij}})=\qu{p_{ij}}, 
\ee
such that $q_{\wt K}=-q_{\wt\Pi}\circ\wt\gamma$ for the associated quadratic forms. 
In fact, we can use
\be\label{pdef}
\begin{array}{rclrcl}
p_{12}&:=&\hf\left( f_1+f_{11}+f_{13}+f_{21}\right),\; &
p_{34}&:=&\hf\left( f_1+f_{2}+f_{14}+f_{17}\right),\\[5pt]
p_{13}&:=&\hf\left( f_1+f_{11}+f_{14}+f_{16}\right),\; &
p_{24}&:=&\hf\left(f_1+f_{13}+f_{17}+f_{18}\right),\\[5pt]
p_{14}&:=&\hf\left( f_1+f_{8}+f_{11}+f_{17}\right),\; &
p_{23}&:=&\hf\left( f_1+f_{4}+f_{13}+f_{14}\right).
\end{array}
\ee
Denoting by $\wt P_{ij}\subset\m I$ the sets of labels such that
$$
{p_{ij}} = {\textstyle\hf} \sum\limits_{n\in\wt P_{ij}} f_n\quad\in\wt \Pi^\ast
$$
in (\ref{pdef}) above, we find that the map $I$ 
of \req{map} maps $\wt P_{ij}$ to the plane $P_{ij}\subset \F_2^{4}$ given in (\ref{planes}). 
In other words, the isometry between $\Pi$ and $\wt\Pi(-1)$ induced by $I$ is compatible
with the gluing prescriptions described in section \ref{exampleK3} and above.

As to uniqueness of the embedding of the Kummer lattice in the Niemeier lattice $N(-1)$, 
we have
\begin{prop}\label{uniquepi}
Consider the Niemeier lattice $N$ of type $A_1^{24}$. The primitive embedding
of the Kummer lattice $\Pi$ in $N(-1)$ found in proposition \mbox{\rm\ref{piprim}} is unique up to 
automorphisms of $N$.
\end{prop}
\begin{proof}
Assume that $i\colon\Pi\hookrightarrow N(-1)$ is a primitive embedding of the 
Kummer lattice $\Pi$ in $N(-1)$. 
Then from \req{Kummer} we deduce that the
$i(E_{\vec a})$, $\vec a\in\F_2^4$, yield 
$16$ pairwise perpendicular vectors in $\wh\Pi:=i(\Pi)\subset N(-1)$ on which 
the quadratic form takes value $-2$. Hence by \req{nieroot} and with the notations
used there, for every $n\in\m I\setminus\m O_9$
we find some $\wh I(n)\in\m I$ such that
$i(E_{I(n)}) = \pm f_{\wh I(n)}\in\wh\Pi$, where $I$ is the map \req{map}
and the indices $\wh I(n)$ are pairwise distinct. In particular, $i$
induces the hypercube structure of $\F_2^4$ on the labels 
$\{\wh I(n)\mid n\in\m I\setminus\m O_9\}$, and the map induced by
$i(E_{I(n)})\mapsto f_n$ for all $n\in\m I\setminus\m O_9$ is an
isometry $\wh\Pi\cong\wt\Pi$. 

Again by \req{Kummer} we have $e:={1\over2}\sum_{\vec a\in\F_2^4} E_{\vec a}\in\Pi$
and thus ${1\over2} \sum_{n\in\m I\setminus\m O_9} f_{\wh I(n)}\in N$. 
However, by proposition \ref{niefromroot}
this implies that $\check{\m O}:=\{ \wh I(n) \mid n\in\m I\setminus\m O_9\}$ 
gives a codeword of the Golay code. It follows that
$\m O:=\m I\setminus\check{\m O}$ is an octad in the Golay code. This means that the lattice
$\wh K:=\wh\Pi^\perp\cap N$ is isometric to the lattice $\wt K$ defined in \req{specialchoice}.
Note that $\wh K$ is generated by the $f_n$ with $n\in\m O$ along with $\hf\sum_{n\in\m O}f_n$.

Using proposition \ref{glueperp} we can recover $N$ either  from gluing $\wh K$
to $\wh\Pi$, or  from gluing $\wt K$
to $\wt\Pi$. On the level of discriminant groups, our isometry $\wh\Pi\cong\wt\Pi$ 
together with the respective gluing isometries yields an
isometry of the discriminant groups $\wh K^\ast/\wh K\longrightarrow \wt K^\ast/\wt K$.
We need to show that this isometry can be lifted to an isometry  
$\wh K\longrightarrow \wt K$ in order to yield the isometry
$\wh K\oplus\wh\Pi\cong\wt K\oplus\wt\Pi$
compatible with gluing. This can be done by explicit calculation:

Given the quadratic form \req{quadform} on our discriminant groups, 
it follows from the gluing prescription \req{Niemeierglue} that the  preimages under 
$\wh K^\ast/\wh K\longrightarrow \wt K^\ast/\wt K$ of  the $\qu{q_{ij}}$ 
with $q_{ij}$ given by 
\req{kijchoice} 
are of the form $\qu{\hf\sum_{n\in\wh Q_{ij}} f_n}$,
where $\hf\sum_{n\in\wh Q_{ij}} f_n\in\wh K^\ast$ is glued to $\hf\sum_{n\in\wh P_{ij}} f_n\in\wh\Pi^\ast$
with $\wh P_{ij}=\wh I(\wt P_{ij})$ and $I(\wt P_{ij})=P_{ij}\subset\F_2^4$ in
\req{planes}. In other words, $\wh Q_{ij}\subset\m O$ is a quadruplet of labels 
such that $\wh P_{ij}\cup\wh Q_{ij}$ is an octad in the Golay code.
Since $\m O$ is an octad, by replacing 
$\wh Q_{ij}$ by $\m O\setminus \wh Q_{ij}$ where necessary, we may assume that
all $\wh Q_{ij}$ share a common label, say $a$. Comparing to the $Q_{ij}\subset\m O_9$
which in \req{kijchoice} yield $q_{ij}=\hf\sum_{n\in Q_{ij}}(\pm f_n)$, we map $f_{a}\mapsto f_{15}$.
By means of symmetric differences, for example $\wh P_{12}+\wh P_{13}+\wh Q_{12}+\wh Q_{13}$
must yield an octad, where $\wh P_{12}+\wh P_{13}=\wh I(\wt P_{12}+\wt P_{13})$ is a quadruplet of labels.
Thus $\wh Q_{12}$ and $\wh Q_{13}$ share precisely two entries, $a, b$, say.
Again comparing to \req{kijchoice}, we map $f_{b}\mapsto f_{6}$. Continuing 
in this fashion one obtains the desired isometry $\wh K\longrightarrow \wt K$
which by construction induces $\wh K^\ast/\wh K\longrightarrow \wt K^\ast/\wt K$.
In other words, it is compatible with gluing and we obtain an isometry
$\wh K\oplus\wh\Pi\cong\wt K\oplus\wt\Pi$ which can be extended
to a lattice automorphism $\alpha$ of $N$. Then $\alpha\circ i$
gives the primitive embedding
of the Kummer lattice $\Pi$ in $N(-1)$ found in proposition \ref{piprim}.
\end{proof}

Finally, we give yet another description of the Niemeier lattice $N$ in terms 
of the gluing techniques of proposition \ref{glueperp}, which resembles the 
description of the full integral K3-homology  in proposition \ref{fullk3}:
\begin{prop}\label{fullniemeierprop}
Consider the Niemeier lattice $N$ of type $A_1^{24}$ with root sublattice $R$ generated
by roots $f_1,\ldots, f_{24}$. With our choice \mbox{\rm\req{o9octad}} of a special octad 
$\m O_9$ let $n_0\in \m O_9$ and 
$$
\begin{array}{rcl}
\wt{\m K}_{n_0}&:=&\left\{ \nu\in N\mid  \fa n\not\in\m O_9\setminus\{n_0\}\colon \langle \nu, f_n\rangle=0 \right\},\\[5pt]
\wt{\m P}_{n_0}&:=&\left\{ \nu\in N\mid \fa n\in\m O_9\setminus\{n_0\}\colon \langle \nu, f_n\rangle=0 \right\}.
\end{array}
$$
Then $\wt{\m K}_{n_0}$ is a root lattice of type $A_1^{7}$ and 
$\wt{\m P}_{n_0}=\wt\Pi\oplus\spann_\Z\{f_{n_0}\}$. 
Both lattices are primitive sublattices of $N$ which are
orthogonal complements of one 
another. Furthermore, $\wt{\m K}_{n_0}^\ast/\wt{\m K}_{n_0}\cong 
\wt{\m P}_{n_0}^\ast/\wt{\m P}_{n_0}\cong(\Z_2)^7$ under an isomorphism 
$\wt g_{n_0}\colon \wt{\m K}_{n_0}^\ast/\wt{\m K}_{n_0}{\longrightarrow} \wt{\m P}_{n_0}^\ast/\wt{\m P}_{n_0}$ with
\ba\label{fullniemeier}
\wt g_{n_0}(\qu{q_{ij}})&:=&\wt\gamma(\qu{q_{ij}}) \mbox{ for }
ij=12,\, 34,\, 13,\, 24,\, 14,\, 23,\nonumber\\
\wt g_{n_0}(\qu{\textstyle\hf \sum\limits_{n\in\m O_9\setminus\{n_0\}} f_n}) &:=& \qu{\textstyle\hf f_{n_0}};\nonumber\\[5pt]
\wt{\m P}_{n_0}\cong {\m P}(-1),\quad \quad\quad
N &\cong& \left\{ (k,p)\in \wt{\m K}_{n_0}^\ast\oplus \wt{\m P}_{n_0}^\ast\mid\wt g_{n_0}(\qu k)=\qu p \right\}
\ea
with notations as in \mbox{\rm\req{kijchoice}--\req{pdef}}.
\end{prop}
\begin{proof}
The lattices
$\wt{\m K}_{n_0},\, \wt{\m P}_{n_0}$  are perpendicular primitive sublattices of $N$ of rank $7$ and $17$
and thus orthogonal complements of one another by construction.

By proposition \ref{niefromroot} the lattice
$\wt{\m K}_{n_0}\subset\wt{\m K}$ is generated by the $f_n$ with $n\in \m O_9\setminus\{n_0\}$ along 
with linear combinations $\hf \sum_{n\in A} f_n$, if $A\subset \m O_9\setminus\{n_0\}$
 corresponds to a codeword of the Golay code. However, since $O_9\setminus\{n_0\}$ 
 contains only $7$ elements, while the shortest non-trivial codeword in the Golay code 
 has weight $8$, we find $\wt{\m K}_{n_0}=\spann_\Z\{ f_n\mid n\in \m O_9\setminus\{n_0\}\}$, 
 which is a root lattice of type $A_1^{7}$ as claimed. 
 
 Similarly, by \req{specialchoice} the lattice
 $\wt{\m P}_{n_0}$ is generated by the elements of $\wt\Pi$ along with $f_{n_0}$ and 
 any linear combination $\hf \sum_{n\in A} f_n$ with $A\cap\m O_9=\{n_0\}$, if 
 $A\subset\m I$ corresponds to a codeword of the Golay code. 
 However, since $\m O_9$ corresponds to a codeword in the Golay code, 
 any two codewords of which intersect in an even number of labels, 
 no such $A$ can exist. In other words, 
 $\wt{\m P}_{n_0}=\wt\Pi\oplus\spann_\Z\{f_{n_0}\}$, and thus $\wt{\m P}_{n_0}\cong {\m P}(-1)$
 by proposition \ref{piinn} along with
 the very definition of $\m P$ in proposition \ref{fullk3}. 
 
Using proposition \ref{glueperp}, the Niemeier lattice $N$ can be glued from
$\wt{\m K}_{n_0}$ and $\wt{\m P}_{n_0}$ as claimed in \req{fullniemeier}.
\end{proof}
\section{The complex geometry and the symmetries of K3 surfaces}\label{k3geometry}
In the preceding section \ref{exampleK3}, we have already addressed some properties 
of K3 surfaces and in particular of Kummer surfaces. 
Recall that the objective of this work is the investigation of finite groups of 
symplectic automorphisms of 
Kummer surfaces, and the role of the Mathieu group $M_{24}$ in their realisation. 
We need a number of additional techniques to describe and investigate 
symplectic automorphisms of Kummer surfaces. This section reviews some of 
these techniques and introduces new ones.

\subsection{Complex structures and dual K\"ahler classes}\label{complex}
Recall from definition \ref{k3definition} that we view a K3 surface $X$ as a complex
surface, which thus in particular comes with a choice of complex structure.
Moreover, $X$ has trivial canonical bundle and hence there exists a 
holomorphic $(2,0)$-form on $X$ which never vanishes and which 
represents a Hodge-de Rham class $\wh\Omega\in H^2(X,\C)$. 
Having worked in homology, so far, we introduce the $2$-cycle 
$\Omega\in H_2(X,\C)$ which is Poincar\'e dual to $\wh\Omega$. 
By construction, it obeys $\Omega\vee\Omega=0$, and 
$H_4(X,\R)\ni\Omega\vee\qu\Omega$ is positive with respect to the intersection
form $\langle\cdot,\cdot\rangle$ on $H_\ast(X,\R)$. 
Decomposing $\Omega$ into its real and its imaginary part,
$$
\Omega = \Omega_1 + i\Omega_2,\quad \Omega_k\in H_2(X,\R),
$$
the above conditions on $\Omega$ immediately imply
$$
\langle \Omega_1,\Omega_2\rangle = 0, \quad
\langle \Omega_1,\Omega_1\rangle = \langle \Omega_2,\Omega_2\rangle >0.
$$
In other words, $\Omega_1,\, \Omega_2\in H_2(X,\R)$ form an orthogonal basis of a 
positive definite oriented $2$-dimensional subspace of $H_2(X,\R)$, which is 
traditionally denoted by $\Omega$, too. While the choice of complex structure on $X$
obviously determines the position of $\Omega$  
relative to the lattice of integral homology $H_2(X,\Z)$, it is a deep theorem, which is 
equivalent to the 
\textsc{global Torelli theorem for K3 surfaces} \cite{ku77,lo81,na83,si81,to80}, that the
converse is also true:
\begin{theorem}[Torelli Theorem]\label{torellithm}
Consider a K3 surface $X$ and the $2$-di\-men\-sio\-nal oriented subspace $\Omega$ of $H_2(X,\R)$
whose basis is represented by the real and imaginary part of the Poincar\'e dual of 
a holomorphic $(2,0)$-form on $X$ which
vanishes nowhere.
The complex structure of  $X$ is uniquely determined by the 
position of $\Omega$ relative to the lattice $H_2(X,\Z)$ of integral 
homology.
\end{theorem}
By the Torelli Theorem \ref{torellithm}, the moduli space {$\MMM_{cpx}$}
of complex structures on a K3 surface is the Grassmannian $\wt\MMM_{cpx}$ of
positive definite oriented $2$-dimensional subspaces $\Omega$ of
$\R^{3,19}=\left(U^3\oplus E_8^2(-1)\right)\otimes\R$, modulo the action of
the automorphism group of $U^3\oplus E_8^2(-1)$. This moduli space is equipped with
its natural topology, which however does not have the Hausdorff property \cite{at58}. By choosing
a marking $H_2(X,\Z)\cong U^3\oplus E_8^2(-1)$, $H_\ast(X,\Z)\cong U^4\oplus E_8^2(-1)$,
we work in the smooth connected cover $\wt\MMM_{cpx}$ of the moduli space. This 
allows us to explicitly specify the complex structure
of a given K3 surface in our applications by writing out the basis 
$\Omega_1,\,\Omega_2$ of $\Omega$ in terms of lattice vectors in $H_2(X,\Z)$.
Indeed, thereby we specify the very location of $\Omega$ relative to $H_2(X,\Z)$
and thus the complex structure, by theorem \ref{torellithm}. 

In terms of local holomorphic coordinates $z_1,\, z_2$, the procedure works as follows:
Locally, a  holomorphic $(2,0)$-form representing $\wh\Omega$ has the form 
$dz_1\wedge dz_2$, which with respect to  real coordinates 
$\vec{x}=(x_1,x_2,x_3,x_4)$, $z_1=x_1+ix_2$, $z_2=x_3+ix_4$, 
as before, yields
\be\label{cohomologyinv}
dz_1 \wedge dz_2=[dx_1 \wedge dx_3 - dx_2\wedge dx_4]
+i\,[dx_1 \wedge dx_4 + dx_2\wedge dx_3] =: \wh\Omega_1+i\wh\Omega_2.
\ee
Here, the real valued two-forms $\wh\Omega_1,\,\wh\Omega_2$ 
represent the Poincar\'e duals of $\Omega_1,\,\Omega_2$. 
Hence with respect  
to standard real coordinate vector fields $\vec e_1,\dots, \vec e_4$, 
the latter are readily identified as
\be \label{homologyinv1}
\Omega_1 = e_1\vee e_3-e_2\vee e_4,\quad\quad
\Omega_2 = e_1\vee e_4+e_2\vee e_3.
\ee
Recall from definition \ref{defKummer} that a Kummer surface
$X$ with underlying torus $T(\Lambda)$ carries the complex structure
induced from the universal cover $\C^2\cong\R^4$ of $T(\Lambda)$. 
Moreover, as remarked at the end of section \ref{exampleK3}, the
Kummer construction provides us with a natural marking
$H_2(X,\Z)\cong U^3\oplus E_8^2(-1)$, $H_\ast(X,\Z)\cong U^4\oplus E_8^2(-1)$,
which we use throughout our discussion of Kummer surfaces. In other words,
by the {above-mentioned} procedure we
calculate the 
complex structure of the Kummer surface $X=\wt{T(\Lambda)/\Z_2}$ in 
terms of the lattice data $\Lambda\subset\C^2\cong\R^4${, where the generators
$\vec\lambda_1,\ldots,\vec\lambda_4$ of $\Lambda$ are part of the data. Indeed, 
$\vec\lambda_1,\ldots,\vec\lambda_4$  are expressed} in terms of 
the standard basis vectors $\vec e_1,\ldots,\vec e_4$ of $\R^4$. Thus we obtain expressions 
for the Poincar\'e duals of $dx_1 \wedge dx_3 - dx_2\wedge dx_4$ and 
$dx_1 \wedge dx_4 + dx_2\wedge dx_3$ in terms of the $\lambda_{ij}=\lambda_i\vee\lambda_j$, 
our standard generators of $H_2(T(\Lambda),\Z)$. 
Now recall from definition \ref{defKummer} the linear map 
$\pi_\ast\colon H_2(T(\Lambda),\R)\longrightarrow H_2(X,\R)$ which is induced by the
rational map $\pi\colon T(\Lambda)\dashrightarrow X$.
The images of the Poincar\'e duals of $dx_1 \wedge dx_3 - dx_2\wedge dx_4$ and 
$dx_1 \wedge dx_4 + dx_2\wedge dx_3$ under $\pi_\ast$ yield the two--cycles 
$\Omega_1,\, \Omega_2$ specifying the complex structure of the Kummer surface $X$. 
One thus immediately obtains expressions for the $\Omega_k$ in terms of the lattice 
$H_2(X,\Z)$, uniquely specifying the complex structure of $X$.
Note that this procedure allows us to vary the underlying lattice $\Lambda$ of $T(\Lambda)=\C^2/\Lambda$
and thereby the induced complex structure of $X=\wt{T(\Lambda)/\Z_2}$, giving the relative position
of $\Omega=\spann_{\R}\{\Omega_1,\Omega_2\}$ in $H_2(X,\R)=H_2(X,\Z)\otimes\R$ in terms of 
$H_2(X,\Z)$ with our choice of marking. Here, the lattice $H_2(X,\Z)$ remains fixed, while the position
of $\Omega$ varies with $\Lambda$, thus describing a path in the smooth connected cover
$\wt\MMM_{cpx}$ of the moduli space of complex structures on K3.
\vspace{0.5em}

\begin{example}
Consider the standard \textsc{square torus} 
$T_0:=T(\Z^4)=\C^2/\Z^4$, where we simply have $\vec e_i=\vec\lambda_i,\, i=1,\ldots,4$, 
and thus $\Omega_1=\pi_\ast\lambda_{13}-\pi_\ast\lambda_{24},\, 
\Omega_2=\pi_\ast\lambda_{14}+\pi_\ast\lambda_{23}\in H_2(X,\Z)$. 
Hence the Kummer surface $X_0$ with underlying torus $T_0$ has the special property that the 
$2$-dimensional space $\Omega\subset H_2(X_0,\R)$ which specifies its complex structure contains 
a sublattice of $H_2(X_0,\Z)$ of (the maximal possible) rank $2$. For such K3 surfaces, by a seminal 
result of Shioda and Inose \cite{shin77}, the quadratic form of the \textsc{transcendental lattice} 
$\Omega\cap H_2(X,\Z)$ already uniquely determines the complex structure of $X$. 
In other words, the complex structure of the Kummer surface $X_0$ with underlying torus 
$T_0$ is uniquely determined by the following quadratic form of its transcendental lattice:
\be\label{standardKummer}
\left( \begin{array}{cc}4&0\\0&4\end{array}\right).
\ee
According to the final remark of \cite{shin77}, this means that $X_0$ agrees with the 
so--called \textsc{elliptic modular surface of level $4$} defined over $\Q\left(\sqrt{-1}\right)$ of \cite[p.\ 57]{sh72}.
\end{example}

By a deep result due to Siu \cite{si81}, every K3 surface is \textsc{K\"ahler}. In fact, in
addition to a complex structure $\Omega\subset H_2(X,\R)$, we always fix a 
\textsc{K\"ahler class} on each of our K3 surfaces $X$. By definition, a K\"ahler class is the 
cohomology class of the two-form which is associated to a K\"ahler metric on $X$. 
By \cite{to83} this amounts to choosing a real, positive, effective element of 
$H^{1,1}(X,\C)$. Under Poincar\'e duality, this translates into the choice of some 
$\omega\in\Omega^\perp\cap H_2(X,\R)$ with $\langle\omega,\omega\rangle>0$, 
ensuring effectiveness by replacing $\omega$ by $-\omega$ if necessary. 
\begin{definition}\label{pol}
Consider a K3 surface $X$, and
let $\Omega\subset H_2(X,\R)$ denote the oriented $2$-dimensional subspace which
specifies the complex structure of $X$ according to the Torelli Theorem \mbox{\rm\ref{torellithm}}. 
A choice of \textsc{dual K\"ahler class} on $X$ is the choice of some 
$\omega\in\Omega^\perp\cap H_2(X,\R)$ with $\langle\omega,\omega\rangle>0$.
If the dual K\"ahler class obeys $\omega\in H_2(X,\Z)$, then $\omega$
is called a \textsc{polarization} of $X$.

If $X$ is a Kummer surface with underlying torus $T(\Lambda)$, let $\omega_T$ denote
the Poincar\'e dual of the standard K\"ahler class induced from  the standard 
Euclidean metric on $\C^2$. Then with $\pi_\ast$ as in definition \mbox{\rm\ref{defKummer}},
we call $\pi_\ast\omega_T\in\Omega^\perp\cap H_2(X,\R)$ the \textsc{induced dual
K\"ahler class on $X$}.
\end{definition}
By the above, a choice of a dual K\"ahler class is equivalent to the choice
of a K\"ahler structure on $X$, and the choice of a polarization is equivalent
to the choice of a K\"ahler structure 
which is represented by an integral K\"ahler form.
It is known that a K3 surface $X$ is \textsc{algebraic}, that is, $X$ can be viewed
as a complex subvariety of some complex projective space $\mathbb P^N$,
if and only if it admits the choice
of a polarization (see e.g.\ \cite[pp.~163, 191]{grha78}). 

We equip all Kummer surfaces with the dual K\"ahler class
induced from its underlying torus. That this 
special type of dual K\"ahler class 
has been chosen is an important assumption which we make throughout our
work. Without loss of generality we are also assuming a coordinate description
for all our tori $T=T(\Lambda)=\C^2/\Lambda$ where the dual K\"ahler class $\omega_T$
of $T$ is induced from the standard K\"ahler structure of $\C^2$.
The K\"ahler class dual to  $\pi_\ast\omega_T$ is actually
located on a wall of the K\"ahler cone of $X$. In other words, it represents a degenerate
(or \textsc{orbifold limit}  of a) K\"ahler metric on $X$. In this sense, the induced (dual)
K\"ahler classes
on our Kummer surfaces yield \textsc{degenerate K\"ahler structures}.

For later convenience, note that in terms of the standard local holomorphic coordinates 
on $T(\Lambda)=\C^2/\Lambda$ the standard K\"ahler class is represented by
\ba\label{standardKaehler}
{1\over2i}(dz_1\wedge d\qu z_1+dz_2\wedge d\qu z_2)
= dx_1\wedge dx_2+dx_3\wedge dx_4
\ea
and hence 
\be \label{homologyinv2}
\omega=\pi_\ast\omega_T=e_1\vee e_2+e_3\vee e_4
\ee
with notations as above. The induced dual K\"ahler class  on the
Kummer surface obtained from $T(\Lambda)$ can thus be immediately calculated in terms 
of the lattice $H_2(X,\Z)$ and our fixed marking, given the lattice $\Lambda$ of the underlying torus. 
\vspace{0.5em}

\begin{example}
For our square torus $T_0=T(\Z^4)$ above we argued that we have 
$\omega=\pi_\ast\omega_T=\pi_\ast\lambda_{12}+\pi_\ast\lambda_{34}\in H_2(X_0,\Z)$ for the associated 
Kummer surface $X_0$. Hence this Kummer surface is algebraic. 
The real $3$-dimensional subspace $\Sigma$ of  $H_2(X_0,\R)$ 
containing $\Omega$ and $\omega$ has the property that 
$\Sigma\cap H_2(X_0,\Z)$ yields a lattice of (the maximal possible) 
rank $3$, with quadratic form
\be\label{standardPolar}
\left( \begin{array}{ccc}4&0&0\\0&4&0\\0&0&4\end{array}\right).
\ee
%
\end{example}
If $X$ denotes a K3 surface with complex structure specified by $\Omega\subset H_2(X,\R)$
as in the Torelli Theorem \ref{torellithm} and with dual K\"ahler class $\omega$ according to
definition \ref{pol}, then the oriented $3$-dimensional subspace $\Sigma$ of $H_2(X,\R)$ containing
$\Omega$ and $\omega$ uniquely determines a real Einstein metric on $X$, up to its volume,
or equivalently a hyperk\"ahler structure on $X$ (see for example \cite{asmo94,nawe00} for
a review). {Vice versa, $\Sigma$ is uniquely determined by such a hyperk\"ahler structure.}
The moduli space $\MMM_{hk}$ of hyperk\"ahler structures on $X$ is thus the Grassmannian
$\wt\MMM_{hk}$ of positive definite oriented $3$-dimensional subspaces $\Sigma$
of $\R^{3,19}=\left(U^3\oplus E_8^2(-1)\right)\otimes\R$, modulo the action of 
the automorphism group of $U^3\oplus E_8^2(-1)$. This moduli space carries a natural
Hausdorff topology. By choosing a fixed marking, as we are doing in this work, we are 
effectively working in the smooth connected cover $\wt\MMM_{hk}$ of this moduli space.
A smooth variation of the {generators $\vec\lambda_1,\ldots,\vec\lambda_4$
of the underlying lattice $\Lambda$} for Kummer surfaces
$\wt{T(\Lambda)/\Z_2}$ thus amounts to the description of a smooth path in
$\wt\MMM_{hk}$, which we simply call a \textsc{Kummer path}.
%
\subsection{Holomorphic symplectic automorphisms of K3 surfaces}\label{autos}
In this subsection, we consider a K3 surface $X$ with a complex structure that is 
encoded in terms of a real $2$-dimensional oriented positive definite subspace 
$\Omega\subset H_2(X,\R)$ according to the Torelli Theorem \ref{torellithm}. 
We discuss the notion of \textsc{symplectic 
automorphisms}\footnote{Here, we follow the slightly misleading terminology which has 
become standard, by now. Note however that in Nikulin's original work such 
automorphisms are called \textsc{algebraic}  \cite[Def.\ 0.2]{ni80b},
and that the definition of symplectic automorphisms does not refer to a symplectic
structure on $X$.} and \textsc{holomorphic symplectic automorphisms} of $X$:
\begin{definition}\label{symmgrp}
Consider a K3 surface $X$.
A map $f\colon X\longrightarrow X$ of finite order is called a
\textsc{symplectic automorphism} if and only if $f$ is biholomorphic
and the induced map 
$f_\ast\colon H_\ast(X,\R)\longrightarrow H_\ast(X,\R)$ leaves 
the complex structure $\Omega\subset H_\ast(X,\R)$
invariant.

If $\omega$ is a dual K\"ahler class on $X$ and $f_\ast\omega=\omega$,
then $f$ is a \textsc{holomorphic symplectic automorphism}
with respect to $\omega$. 

When a dual K\"ahler class $\omega$ on $X$ has been specified, then the group
of holomorphic symplectic {automorphisms} of $X$ with respect to $\omega$ is called
the \textsc{symmetry group} of $X$.
\end{definition}
Consider a map $f\colon X\longrightarrow X$ of finite order
which is bijective, such that
 both $f$ and $f^{-1}$ respect the complex structure of $X$. 
The induced linear map $f_\ast$ on $H_\ast(X,\R)$
restricts to a lattice automorphism on $H_2(X,\Z)$ and, according to our description of complex
structures by means of the Torelli Theorem \ref{torellithm}, it
acts as a multiple of the identity on $\Omega$. 
By definition \ref{symmgrp},
$f$ is a symplectic automorphism if and only
if $f_\ast$ acts as  the identity on $\Omega$.
In fact, by another version of the Torelli Theorem
(see e.g.\ \cite[Thm.~2.7$^\prime$]{ni80b}) the 
{restriction of $f_\ast$ to $H_2(X,\Z)$ (also denoted $f_\ast$)}
uniquely determines the symplectic automorphism $f$. The following version of the Torelli
Theorem is most adequate for our applications to symplectic automorphisms, and it
is obtained from 
\cite[Thms.~2.7$^\prime$\&4.3]{ni80b} in conjunction with \cite{ku77} and \cite[Prop.~VIII.3.10]{bpv84}:
\begin{theorem}[Torelli Theorem] \label{niktorelli}
Consider a K3 surface $X$ and a
lattice automorphism $\alpha$ of $H_2(X,\Z)$, i.e.\ a linear map which respects the 
intersection form. Assume that after linear extension to $H_2(X,\R)$, $\alpha$ 
leaves $\Omega$ invariant. 

Then $\alpha$ 
is induced by a symplectic automorphism of $X$ if and only if the 
following holds: $\alpha$
preserves effectiveness for every $E\in \Omega^\perp\cap H_2(X,\Z)$ with $\langle E,E\rangle=-2$, 
the invariant sublattice $L^\alpha:=H_2(X,\Z)^\alpha
=\{v \in H_2(X,\mathbb{Z})\,\vert \alpha v=v\}$ 
has a negative definite orthogonal complement 
$L_\alpha:=(L^\alpha)^\perp\cap H_2(X,\Z)$, and for all $v\in L_\alpha$, $\langle v,v\rangle\neq-2$.

In this case the symplectic automorphism $f$ with $\alpha=f_\ast$ is uniquely determined.
\end{theorem}
This theorem allows us to carry out the entire discussion of symplectic automorphisms
of a K3 surface $X$ in terms of lattice automorphisms on $H_\ast(X,\Z)$.
To treat finite symplectic automorphism groups, we
\begin{remark}\label{finitetorelli}
If $G$ is a finite group of symplectic automorphisms of a K3 surface $X$, let
$L_G:=(L^G)^\perp\cap H_2(X,\Z)$ denote the orthogonal complement of the lattice
$L^G\subset H_\ast(X,\Z)$ which is pointwise invariant under the induced action of
$G$ on $H_\ast(X,\Z)$, generalizing the lattice {$L_\alpha$} in the above version
of the Torelli Theorem \mbox{\rm\ref{niktorelli}}. By \mbox{\rm\cite[Lemma 4.2]{ni80b}}, the lattice
{$L_G$} is negative definite. To see why this only follows when $G$ is finite, note that
the proof given in \mbox{\rm\cite{ni80b}} uses the resolution of the quotient $X/G$. Indeed, by
our assumptions on $G$, all singularities of $X/G$ are isolated, and
the minimal resolution of all its singularities yields a K3 surface $Y$. 
Moreover, {in \mbox{\rm\cite{ni80b}} Nikulin proves $L_G\otimes\Q\cong M_Y\otimes\Q$} with $M_Y$ the sublattice of 
$H_2(Y,\Z)$ generated by the exceptional classes in the resolution. Since
$M_Y$ is negative definite by \mbox{\rm\cite{mu62}}, the claim follows.
\end{remark}
Now the
restriction to finite symplectic automorphism groups of 
K3 surfaces conveniently translates into the restriction to
symplectic automorphism groups which preserve some dual
K\"ahler class:
\begin{prop}\label{mukaifinite}
Consider a K3 surface $X$, and denote by $G$ a non-trivial group of symplectic 
automorphisms of $X$. 

Then $G$ is finite if and only if $X$ possesses a dual K\"ahler class which is invariant under $G$.
Equivalently, the sublattice $L^G\subset H_\ast(X,\Z)$  which is pointwise invariant under
the induced action of $G$ on 
$H_\ast(X,\Z)$ has signature $(4,m_-)$ with $m_-\in\N$ and $m_-\geq1$.

If $X$ is an algebraic K3 surface, then $G$ is finite if and only if $X$ possesses a 
polarization which is invariant under $G$.
\end{prop}
\begin{proof}
Consider the \textsc{Picard lattice} 
$\Pic(X):=\Omega^\perp\cap H_2(X,\Z)$ of rank
$\rho\in\N, 0\leq\rho\leq20$, and {its orthogonal complement
$T_X:=\Pic(X)^\perp\cap H_2(X,\Z)$. Then by \cite[Lemma 4.2]{ni80b},
the lattice $H_0(X,\Z)\oplus T_X\oplus H_4(X,\Z)$ is a sublattice of
$L^G$, so in particular  $\Omega\subset \left( L^G\cap H_2(X,\Z)\right)\otimes\R$}.

Assume first that $G$ is finite. 
Then by Remark \ref{finitetorelli}, the 
lattice 
obtained as $L_G:=\left(L^G\right)^\perp\cap H_2(X,\Z)$ is negative definite, 
such that $L^G$ is a lattice of signature $(4,m_-)$
for some $m_-\in\N$. Since $H_0(X,\Z)\oplus H_4(X,\Z)\subset L^G$, we have
$m_-\geq1$ as claimed. 

Next assume that $L^G$ has signature $(4,m_-)$ with $m_-\geq1$.
Recall from the above that $\Omega\subset \left( L^G\cap H_2(X,\Z)\right)\otimes\R$,
where $H_0(X,\Z)\oplus H_4(X,\Z)\subset L^G$ implies that $L^G\cap H_2(X,\Z)$
has signature $(3,m_--1)$, while $\Omega$ is positive definite of dimension $2$.
Hence there exists some $\omega\in \left( L^G\cap H_2(X,\Z)\right)\otimes\R$ with
$\omega\in\Omega^\perp$ and
$\langle\omega,\omega\rangle>0$. In other words, $X$ possesses a dual K\"ahler
class which is invariant under $G$. 

Conversely, assume that $X$ possesses a dual K\"ahler class $\omega$ which is invariant under $G$.
Hence the induced action of $G$ on $H_\ast(X,\R)$ leaves the $3$-dimensional vector space
$\Sigma\subset H_\ast(X,\R)$ generated by $\Omega$ and $\omega$ invariant. Thus the non-trivial
action of $G$ on $H_\ast(X,\Z)$ amounts to an action on the negative definite lattice 
$\Sigma^\perp\cap H_2(X,\Z)$. It is therefore given by a permutation group on every set of 
vectors of same length in that lattice, a finite set for every length. Thus $G$ is represented on $H_\ast(X,\R)$
by the action of a finite group. Moreover,
the uniqueness statement in the Torelli Theorem \ref{niktorelli} implies that $G$ acts
faithfully on $H_\ast(X,\R)$. It follows that $G$ is finite.

Now assume that $X$ is algebraic.
Since by definition \ref{pol} every polarization of $X$ is a dual K\"ahler class,
it remains to be shown that there exists a polarization $\omega$ fixed by $G$ if
$L^G$ has signature $(4,m_-)$ with $m_-\geq1$.
Indeed, if $X$ is algebraic, then
$\Pic(X)$ has signature $(1,\rho-1)$
with $1\leq\rho\leq20$, and thus {its orthogonal complement}
$T_X$ has signature $(2,20-\rho)$. Moreover, {by the above
the lattice $H_0(X,\Z)\oplus T_X\oplus H_4(X,\Z)$
of signature $(3,21-\rho)$ is a sublattice of
$L^G$}. Hence for $L^G$ of signature $(4,m_-)$, $m_-\geq1$,
there exists some $\omega\in L^G$ with
$\langle\omega,\omega\rangle>0$ and 
$\omega\in T_X^\perp\cap H_2(X,\Z)\subset\Omega^\perp$. In other words,
there exists a polarization of $X$ which is invariant under $G$. 
\end{proof}

In this work, we restrict our attention to the investigation of finite symplectic automorphism
groups of K3 surfaces $X$. 
Hence by theorem \ref{mukaifinite}
there exists a dual K\"ahler class $\omega$ of $X$ which is invariant under the action of $G$. 
Once a complex structure $\Omega$ and dual K\"ahler class $\omega$ of our K3 surface
have been chosen, the group $G$ of all holomorphic symplectic automorphisms with respect
to $\omega$ is uniquely determined. By {the above}, it solely depends on the relative
position of the $3$-dimensional space $\Sigma\subset H_\ast(X,\R)$ generated by $\Omega$
and $\omega$ with respect to the lattice $H_\ast(X,\Z)$. As explained at the end of section
\ref{complex}, $\Sigma$ uniquely determines a hyperk\"ahler structure on $X$ and is uniquely
determined by such a hyperk\"ahler structure. Working with a fixed marking, our study of finite
symplectic automorphism groups of K3 surfaces thus is naturally carried out on the smooth
connected cover $\wt\MMM_{hk}$ of the moduli space of hyperk\"ahler structures on K3.

Mukai shows in 
\cite{mu88} that the action of a finite symplectic automorphism group
$G$ on $H_\ast(X,\Q)$ is a \textsc{Mathieu representation}, that is: 
The character $\mu$ of this representation is given by
$$
\mu(g) = 24 \left( \ord(g) \prod_{p\mid \ord(g)} \left( 1+{1\over p}\right)\right)^{-1}
\quad\fa g\in G.
$$
Furthermore, with $M_{23}\subset M_{24}$ the stabilizer group in the Mathieu group 
$M_{24}$ of one label in $\m I=\{1,\ldots,24\}$, 
Mukai found the following
seminal result \cite{mu88}:
\begin{theorem}\label{mukai}
Consider a K3 surface $X$, and a finite group $G$
of symplectic automorphisms of $X$. Then $G$ is isomorphic to a subgroup of one of those 
$11$ subgroups of $M_{23}$ which decompose $\m I$ into at least five orbits. 
\end{theorem}
Following Mukai, two further independent proofs of the above theorem
were given, namely by Xiao \cite{xi96} and by Kondo \cite{ko98}, see also \cite{ma86}. 
We use a number of ideas from Kondo's proof throughout this work. 
Therefore,  we briefly review its main steps, adjusting it slightly to fit our purposes:
\vspace{0.5em}

\noindent
\begin{proofsketch}
Assume that $X$ is a K3 surface and that $G$ is a non-trivial finite group acting as 
symplectic automorphism group on $X$. Let $L^G\subset H_\ast(X,\Z)$ denote the invariant sublattice 
of the integral homology\footnote{Here, we slightly modify Kondo's conventions: First, we work in 
homology instead of cohomology, which by Poincar\'e duality is equivalent. Second, instead of 
restricting to $H_2(X,\Z)$ we consider the total integral K3-homology, such that our lattice 
$L^G$ differs from the one in Kondo's work by a summand $H_0(X,\Z)\oplus H_4(X,\Z)\cong U$, 
a hyperbolic lattice. Since the latter is unimodular, the arguments carry through identically.} 
of $X$ and $L_G:=\left(L^G\right)^\perp\cap H_\ast(X,\Z)$. 
By proposition \ref{mukaifinite}, $L_G$ is negative definite of rank at most $19$, while 
$L^G$ has at least rank $5$. Moreover, denoting by $\upsilon_0,\, \upsilon$  a choice of generators of 
$H_0(X,\Z),\, H_4(X,\Z)$ with $\langle\upsilon_0,\upsilon\rangle=1$, we have a lattice vector 
$\upsilon_0- \upsilon\in L^G$ on which the quadratic form takes value $-2$. 

As an application of theorem \ref{nik}, Kondo proves that a lattice $N_G\oplus\langle2\rangle$ isometric
to $(L_G\oplus\langle-2\rangle)(-1)$ can be primitively embedded in some Niemeier lattice $\wt N$
(see definition \ref{niemeierdef}), 
where $\langle2\rangle$ denotes a lattice of rank $1$ with quadratic form $(2)$ on a generator $f$, 
and $\langle-2\rangle=\langle2\rangle(-1)$. In fact, since $G$ acts trivially on $L^G$, it acts trivially 
on the discriminant group $(L^G)^\ast/L^G$. This in turn implies a trivial action of $G$ on the 
discriminant group of $L_G$, since $H_\ast(X,\Z)$ is obtained by the gluing techniques of
proposition  \ref{glueperp} from $L^G$ and $L_G$. 
Hence on the Niemeier lattice $\wt N$, which can be obtained by the gluing techniques from 
$N_G\cong L_G(-1)$ and its orthogonal complement $N^G:=(N_G)^\perp\cap\wt N$, 
the action of $G$ on $N_G\cong L_G(-1)$ can be extended to $\wt N$, leaving $N^G$ 
invariant (see \cite[Prop.~1.1]{ni80b}). Note that in particular, by construction, the invariant sublattice 
$N^G$ of $\wt N$ has rank $\rk(N^G)=\rk(L^G)\geq5$, and it contains the vector $f$ on which the 
quadratic form takes value $2$. While $N^G$ and $L^G$ in general have little in common, apart 
from their ranks and their discriminant groups, note that we can naturally 
identify\footnote{To obtain such an interpretation for $f$, our modification of Kondo's conventions is 
crucial.} $f\in N^G$ with $\upsilon_0-\upsilon\in L^G\subset H_\ast(X,\Z)$.

Next, for each Niemeier lattice $\wt N$ with root sublattice $\wt R$, Kondo shows that the induced 
action on $\wt N/\wt R$ 
gives an injective image of the $G$-action and that it yields an embedding of $G$ in $M_{23}$. 
The latter is readily seen 
in the case of the Niemeier lattice $N$ of type $A_1^{24}$: Here, 
$N/R\cong \m G_{24}\subset\F_2^{24}$ with $\m G_{24}$ the Golay code by proposition \ref{niefromroot}. 
Hence the action of $G$ yields a group of automorphisms of the Golay code. Since $M_{24}$ is the 
automorphism group of $\m G_{24}$, 
this yields an embedding of $G$ in the Mathieu group $M_{24}$. Moreover, 
the invariant part $N^G$ of $N$ by 
construction contains the root $f$. Hence the induced action of $G$ on the Golay code 
stabilizes the corresponding label in 
$\m{I}$. Therefore, Kondo's construction indeed embeds $G$ in 
the subgroup $M_{23}\subset M_{24}$ which stabilizes that label. 
\end{proofsketch}

In fact, by Mukai's appendix to Kondo's 
paper, the Niemeier lattice $N$  of type $A_1^{24}$ can be used to construct a symplectic 
action of each of the $11$ groups $G$ in Mukai's classification. Note that this does not
imply that the Niemeier lattices $\wt N$ can be replaced by $N$ within the above proof. 
Mukai also proves in \cite{mu88} that each
of the $11$ groups in his classification actually
occurs as the symplectic automorphism group of some algebraic 
K3 surface. The largest of these subgroups of $M_{23}$
is the polarization-preserving symplectic automorphism
group $M_{20}$ of  a particular deformation of the Fermat quartic, 
a group of order $960$.
%
\subsection{Holomorphic symplectic automorphisms of Kummer surfaces}\label{Kummerautos}
%
Throughout this subsection we assume that $X$ is a Kummer surface with
underlying torus $T=T(\Lambda)$, $\Lambda\subset\C^2$. According to definition \ref{defKummer},
$X$ carries the complex structure which is induced from the universal cover $\C^2$ of $T$.
Furthermore, according to definition \ref{pol}, we use the induced dual K\"ahler class on $X$. We
are interested in the group of holomorphic symplectic automorphisms of $X$ with respect
to that dual K\"ahler class. It is important
to keep in mind that we are making a very special choice of dual K\"ahler class,
which severely restricts the types of symmetry groups that are accessible to our methods.
In this subsection, we determine the generic structure of such 
holomorphic symplectic automorphism groups of Kummer surfaces. While
the description of these groups themselves, which we
give in proposition \ref{sdproduct}, along with the ideas that lead to
it are known to the experts, our observation about their action on the K3-homology in
theorem \ref{biglattice} is new.

We begin by specifying a group of holomorphic symplectic automorphisms 
that all our Kummer surfaces share:
\begin{prop}\label{findgeneric}
Consider a Kummer surface $X$ with underlying torus $T$, equipped
with the induced dual K\"ahler class. Let $G$ denote the holomorphic symplectic
automorphism group of $X$ with respect to that dual K\"ahler class. 
Then $G$ contains an abelian subgroup $G_t$ which is 
isomorphic to $(\Z_2)^4$.  With notations as in definition
\mbox{\rm\ref{defKummer}} and proposition \mbox{\rm\ref{Kummerform}},
$$
G_t = \left\{ t^{\vec b} \mid \vec b\in\F_2^4\right\},
$$
where each $\alpha^{\vec b}:=t^{\vec b}_\ast$ acts trivially on $K=\pi_\ast (H_2(T,\Z))$, and its action on 
the Kummer lattice $\Pi$ is induced by
$E_{\vec a}\mapsto E_{\vec a+\vec b}$ for all $\vec a\in\F_2^4$.
\end{prop}
\begin{proof}
From proposition \ref{Kummerform} it immediately follows that the action of each 
$\alpha^{\vec b}$
described above induces a lattice automorphism on $H_\ast(X,\Z)$, such that for
the sublattice $L^{\alpha^{\vec b}}$ of $H_2(X,\Z)$ which is invariant under $\alpha^{\vec b}$ we have 
$L^{\alpha^{\vec b}}\supset K$. It follows that $L_{\alpha^{\vec b}}:=(L^{\alpha^{\vec b}})^\perp\cap H_2(X,\Z)$ obeys
$L_{\alpha^{\vec b}}\subset\Pi$. Thus $L_{\alpha^{\vec b}}$ is negative definite. Moreover, by construction 
${\alpha^{\vec b}}$ preserves effectiveness, and 
since $e:=\hf\sum_{\vec a\in\F_2^4} E_{\vec a}\in L^{\alpha^{\vec b}}$ and 
$L_{\alpha^{\vec b}}= (L^{\alpha^{\vec b}})^\perp\cap\Pi$,
$\langle v,v\rangle\neq -2$ for all $v\in L_{\alpha^{\vec b}}$.
Hence by the Torelli Theorem \ref{niktorelli}, ${\alpha^{\vec b}}$ is indeed induced by a uniquely
determined holomorphic symplectic automorphism $t^{\vec b}$ of $X$.
\end{proof}

With $T=T(\Lambda)$, recall from \req{labels} and definition \ref{defKummer} that
each $E_{\vec a}$, $\vec a\in\F_2^4$, is obtained by blowing up a singular point 
$\vec{F}_{\vec{a}}$
in $T/\Z_2$. Hence for every lattice vector 
$\vec\lambda\in\Lambda$, the shift symmetry $\vec x\mapsto \vec x+\hf\vec\lambda$  
for $\vec x\in\R^4$ induces a symmetry on $T/\Z_2$ which permutes the singular points 
by the corresponding shift on the hypercube $\F_2^4$.  If 
$\lambda=\sum_{i=1}^4 b_i\vec\lambda_i$ with generators $\vec\lambda_1,\ldots,\vec\lambda_4$
of $\Lambda$ as in \req{labels}
and $b_i\in\{0,1\}$ for all $i$, then this symmetry induces the holomorphic
symplectic automorphism $t^{\vec b}$ of the above proposition. 
This motivates our terminology of
\begin{definition}\label{defgeneric}
Consider a Kummer surface $X$ with underlying torus $T$, equipped
with the induced dual K\"ahler class. We call the group $G_t\cong (\Z_2)^4$ of holomorphic
symplectic automorphisms $t^{\vec b}$, $\vec b\in\F_2^4$, obtained in proposition \mbox{\rm\ref{findgeneric}}
the \textsc{translational automorphism group of $X$}.
\end{definition}
We now proceed to determine the structure of the symmetry group 
$G$ of a Kummer surface. First, the action of $G$ preserves 
the special structure of the homology of Kummer surfaces that we described in 
proposition \ref{Kummerform}; the statements of the following proposition follow
immediately from the results of \cite{ni75}:

\begin{prop}\label{kummeraffine}
Consider a Kummer surface $X$ with underlying torus $T$, equipped
with the induced dual K\"ahler class. Let $f$ denote a holomorphic symplectic automorphism
of $X$ and $\alpha:=f_\ast$ its induced action on $H_\ast(X,\Z)$. 

With notations as in definition
\mbox{\rm\ref{defKummer}} and proposition \mbox{\rm\ref{Kummerform}}, 
for the Kummer lattice $\Pi$ we have $\alpha(\Pi)=\Pi$, and
for $K=\pi_\ast (H_2(T,\Z))$ we have $\alpha(K)=K$. 

There is an affine linear map $A$ on $\F_2^4$ such that the action
of $\alpha$ on $\Pi$ is induced by $\alpha(E_{\vec a})=E_{A(\vec a)}$, $\vec a\in\F_2^4$. 
In fact, $A$ uniquely determines $f$.
\end{prop}
%
%
%

Note that the translational automorphisms $t^{\vec b}$, $\vec b\in\F_2^4$, found in
proposition \ref{findgeneric} are precisely the holomorphic symplectic
automorphisms of a Kummer surface which are given by translations
$A(\vec a)=\vec a+\vec b$
on the hypercube $\F_2^4$ according to proposition \ref{kummeraffine}.
\vspace*{0.5em}
We are now ready to describe the generic form of every symmetry 
group of a Kummer surface. As mentioned above, this
description along with the ideas which lead to it are known to the experts.
However, for the reader's convenience and since we have not found an appropriate reference 
in the literature, we recall its derivation:
\begin{prop}\label{sdproduct}
Consider a Kummer surface $X$ with underlying torus $T=T(\Lambda)$, equipped
with the induced dual K\"ahler class. Let $G$ denote its symmetry group. Then the following holds:
\begin{enumerate}
\item
The group $G$ is a semi-direct product $G_t\rtimes G_T$ with $G_t$ the 
translational automorphism group of definition \mbox{\rm\ref{defgeneric}} and 
$G_T$ the group of those symplectic automorphisms of $X$ that are induced by the 
holomorphic symplectic automorphisms of $T$ that fix $0\in\C^2/\Lambda=T$. 
In other words\footnote{The
notion of holomorphic symplectic automorphisms on complex tori is 
defined completely analogously to our definition for K3 surfaces.},
$G_T\cong G_T^\prime/\Z_2$ with $G_T^\prime$ the group of 
non-translational holomorphic
symplectic automorphisms of $T$.
\item 
The group $G$ is isomorphic to a subgroup of $(\Z_2)^4\rtimes A_7$ with
$A_7$ the group of even permutations on $7$ elements.
\end{enumerate}
\end{prop}
\begin{proof}
Throughout this proof we use the notations introduced in definition
\ref{defKummer} and proposition \ref{Kummerform}.
\begin{enumerate}
\item
Let $f\in G$ and $\alpha:=f_\ast$ as before. 
Proposition \ref{kummeraffine} implies that $\alpha$
induces an automorphism of $K=\pi_\ast (H_2(T,\Z))\cong H_2(T,\Z)(2)$
and thus an automorphism $\alpha_T^\prime$ of $H_2(T,\Z)$. Then $\alpha_T^\prime$ 
acts as a Hodge isometry on $H_2(T,\Z)$, since
$\alpha$ does so on $H_2(X,\Z)$ and complex structure and dual K\"ahler class of $X$ are induced
by those on $T$. Hence by the global Torelli Theorem for complex tori, this means that
$\alpha_T^\prime=(f^\prime_T)_\ast$ for  some holomorphic symplectic
automorphism $f^\prime_T$ of $T$.
Let $f_T\in G_T$ denote the automorphism which $f^\prime_T$ induces on $X$
(see for example the proof of \cite[Thm.~2]{ni75}). 
We claim that $f = t^{\vec b}\circ f_T$ for some $\vec b\in\F_2^4$
and $t^{\vec b}\in G_t$ as in proposition \ref{findgeneric}. 
Indeed, by construction,
$t_\ast:=\alpha\circ ((f_T)_\ast)^{-1}$ acts trivially on $K=\pi_\ast (H_2(T,\Z))$. 
Under gluing $\Pi$ to
$K$ according to proposition \ref{glueperp}, this trivial
action must be compatible with the action of $t_\ast$ on $\Pi$. 
In particular, $t_\ast$ must induce an action on the hypercube
$\F_2^4$ of labels $\vec a$ of the $E_{\vec a}$ which maps every affine plane in $\F_2^4$
to a parallel plane. One checks that this implies that {$t_\ast\in G_t$}.

The uniqueness of the decomposition $f=t^{\vec b}\circ f_T$ with $t^{\vec b}\in G_t$
and $f_T\in G_T$ and then $G\cong G_t\rtimes G_T$
now follow from proposition \ref{kummeraffine}.
\item
Since we already know that $G\cong G_t\rtimes G_T$ with $G_t\cong (\Z_2)^4$, 
it remains to show that $G_T\subset A_7$. However, this immediately follows
from Fujiki's classification of holomorphic symplectic automorphism groups
of complex tori \cite[Lem\-ma 3.1\&3.2]{fu88}. Indeed, Fujiki proves that $G_T^\prime$
is isomorphic to a subgroup of one of the following groups: The cyclic groups $\Z_4$
or $\Z_6$ with $\Z_4/\Z_2\cong\Z_2$ and $\Z_6/\Z_2\cong\Z_3$, the binary dihedral
groups $\m O$ or $\m D$ of orders $8$ or $12$ with $\m O/\Z_2\cong\Z_2\times\Z_2$
and $\m D/\Z_2\cong S_3$ (the permutation group on $3$ elements), or
the binary tetrahedral group $\m T$ with $\m T/\Z_2\cong A_4$ (the group of even permutations
on $4$ elements). Hence $G_T$ is isomorphic to a subgroup of $A_6$, the group of even permutations
on $6$ elements. With $A_6\subset A_7$ the claim follows.\vspace*{-1.5em}
\end{enumerate}
\end{proof}

Note that by our proof of proposition \ref{sdproduct}, every holomorphic symplectic
automorphism group of a Kummer surface is in fact isomorphic to a subgroup of
$(\Z_2)^4\rtimes A_6$. However, for reasons that will become clear later, we prefer working
with the bigger group $(\Z_2)^4\rtimes A_7$.
{The proposition implies that the translational group $G_t\cong(\Z_2)^4$ is the symmetry
group of generic Kummer surfaces with induced (dual) K\"ahler class.}
\begin{definition}\label{overarch}
We call the group $(\Z_2)^4\rtimes A_7$, where $A_7$ denotes the group of even permutations
on $7$ elements,
the \textsc{overarching finite symmetry group of Kummer surfaces}.
\end{definition}
Let us now take a closer look at the lattices that are involved in the action of a
symmetry group of  a Kummer surface (see table \ref{latticetableG} for a snapshot of lattices used in the following):

\begin{prop}\label{latticechar}
Consider a Kummer surface $X$ with underlying torus $T$, equipped
with the induced dual K\"ahler class. Let $G$ denote its symmetry group. 
Then the following holds, with notations as in
proposition \mbox{\rm\ref{fullk3}}:
\begin{enumerate}
\item
For the Kummer lattice $\Pi$ and the induced action of $G$ on it, 
$$
\Pi^G=\Pi\cap H_\ast(X,\Z)^G=\mbox{span}_\Z\{e\} 
\quad \mbox{ with }\quad
e:={1\over2}\sum_{\vec a\in\F_2^4} E_{\vec a}.
$$
\item
Let $L_G:= \left( H_\ast (X,\Z)^G \right)^\perp \cap H_\ast(X,\Z)$ as before
and $\upsilon_0\in H_0(X,\Z)$, $\upsilon\in H_4(X,\Z)$ such that
$\langle\upsilon_0,\upsilon\rangle=1$. With $G_T^\prime$ the group of holomorphic
symplectic automorphisms of $T$ that fix $0\in\C^2/\Lambda$, the lattice 
$$ 
M_G:=L_G
\oplus \mbox{span}_\Z \left\{ e, \upsilon_0-\upsilon \right\}
$$
is the orthogonal complement of the lattice $M_G^\prime:=\pi_\ast(H_2 (T,\Z)^{G_T^\prime})
 \oplus \mbox{span}_\Z \left\{  \upsilon_0+\upsilon \right\}$ in $H_\ast(X,\Z)$.
\end{enumerate}
\end{prop}
\begin{proof}
\begin{enumerate}
\item
This follows immediately from the fact that $G$ contains the translational group
$G_t\cong(\Z_2)^4$ by proposition \ref{findgeneric} {along with proposition \ref{kummeraffine}}.
\item
From propositions \ref{kummeraffine} and \ref{sdproduct} 
it follows that $H_2(X,\Z)^G$ is obtained by gluing $\pi_\ast (H_2(T,\Z)^{G_T^\prime})$ 
and $\Pi^G$ according to proposition \ref{glueperp}. Hence
the claim follows from $H_\ast(X,\Z)^G = H_2(X,\Z)^G \oplus \mbox{span}_\Z\{ \upsilon_0, \upsilon\}$
together with $\Pi^G=\mbox{span}_\Z\{e\} $ along with proposition \ref{glueperp}.\vspace*{-1.5em}
\end{enumerate}
\end{proof}

We are now ready to prove and appreciate the following result, which is crucial to our investigations:
\begin{theorem}\label{biglattice}
Let $X$ denote a Kummer surface with underlying torus $T$, equipped with the 
complex structure and dual K\"ahler class which are induced from $T$. Let $G$ denote the
symmetry group of $X$. With notations
as in proposition \mbox{\rm\ref{latticechar}}, in particular 
$L_G= \left( H_\ast (X,\Z)^G \right)^\perp \cap H_\ast(X,\Z)$ and
$\upsilon_0\in H_0(X,\Z)$,
$\upsilon\in H_4(X,\Z)$ such that
$\langle\upsilon_0,\upsilon\rangle=1$, and $e:={1\over2}\sum_{\vec a\in\F_2^4} E_{\vec a}$,
the lattice $M_G(-1)$ with
$$ 
M_G:= L_G \oplus \mbox{span}_\Z \left\{ e, \upsilon_0-\upsilon \right\}
$$
can be primitively embedded in the Niemeier lattice $N$ of type $A_1^{24}$.
\end{theorem}
\begin{proof}
We apply Nikulin's theorem \ref{nik} to the 
lattice $\Lambda:=M_G(-1)$ of signature $(\ell_+,\ell_-)$ and to the pair of integers $(\gamma_+, \gamma_-)=(24,0)$.
Throughout the proof we use the notations introduced in definition
\ref{defKummer}, proposition \ref{Kummerform} and proposition \ref{latticechar}. 

Notice first that by proposition \ref{latticechar} our lattice $M_G$ contains the lattice
\begin{eqnarray*}
\wh M_G &:=& \wh K_G \oplus \Pi \oplus \mbox{span}_\Z\{\upsilon_0-\upsilon\}\\ 
&&\quad\quad\mbox{ with }
 \wh K_G:= \pi_\ast\left( \left( H_2(T,\Z)^{G_T^\prime} \right)^\perp\cap H_2(T,\Z)\right).
\end{eqnarray*}
In fact property 2. in
proposition \ref{latticechar} implies $\rk(M_G)=\rk(\wh M_G)$. 
Hence the lattice $M_G(-1)$ has signature $(l_+,l_-)=(17+k,0)$,
where $17=\rk(\Pi\oplus\mbox{span}_\Z\{\upsilon_0-\upsilon\})$
and $k:=\rk(\wh K_G)$.  In particular, we have $k\leq3$, since
$H_2(T,\R)^{G_T^\prime}=H_2(T,\Z)^{G_T^\prime}\otimes\R$ contains the positive definite 
$3$-dimensional subspace 
$\Sigma_T\subset H_2(T,\R)$ 
yielding complex structure and dual K\"ahler class of $T$,
such that $H_2(T,\Z)^{G_T^\prime}$ has at least rank $3$. Hence condition 1.~in theorem \ref{nik} holds. 

By condition 2.~in proposition \ref{latticechar}, $M_G$ is the orthogonal complement of the primitive
sublattice $M_G^\prime=\pi_\ast(H_2 (T,\Z)^{G_T^\prime})
 \oplus \mbox{span}_\Z \left\{  \upsilon_0+\upsilon \right\}$ in the unimodular lattice $H_\ast(X,\Z)$.
Hence in particular $A_q=(M_G)^\ast/M_G \cong (M_G^\prime)^\ast/M_G^\prime$ and 
$$
\ell( A_q) = \ell( (M_G^\prime)^\ast/M_G^\prime ) \leq \rk( M_G^\prime ) = 7-k = \gamma_++\gamma_--l_+-l_-,
$$
proving condition 2.~of theorem \ref{nik}.
 
Since $M_G(-1)$ contains the lattice $\mbox{span}_\Z\{\upsilon_0-\upsilon\}(-1)$
of type $A_1$ as direct summand, condition
3.i.~of Nikulin's theorem is also immediate.

It hence remains to be shown that
\begin{equation}\label{lengths}
\ell(A_{q_p})< 7-k = \gamma_++\gamma_--l_+-l_- \mbox{ for every prime } p\neq2.
\end{equation}
Using $\wh M_G\subset M_G \subset M_G^\ast \subset \wh M_G^\ast$, we find
$M_G^\ast / M_G \cong (M_G^\ast/\wh M_G) / (M_G/\wh M_G)$ and therefore
$$
A_q = M_G^\ast/M_G \subset \left (\wh M_G^\ast / \wh M_G\right) \slash 
\left(M_G/\wh M_G\right).
$$
Since $\wh M_G^\ast / \wh M_G = \wh K_G^\ast/\wh K_G \times (\Z_2)^7$
with $\ell ( \wh K_G^\ast/\wh K_G )\leq \rk( \wh K_G)\leq 3$, for
every prime $p\neq2$ this implies
$$
\ell ( A_{q_p} ) \leq \rk( \wh K_G)\leq 3 < 7-k = \gamma_++\gamma_--l_+-l_-.
$$
Equation (\ref{lengths}) is proved, so the assumptions of Nikulin's theorem \ref{nik} hold,
which implies that $M_G(-1)$ can be primitively embedded into some Niemeier lattice $\widetilde N$.
By the above, 
$\Pi(-1) \subset \widehat{M}_G (-1) \subset M_G (-1)$  and thus $\widetilde N\cong N$ by proposition \ref{piprim}.
\end{proof}

To fully appreciate the theorem, note that it immediately implies the following
\begin{corollary}\label{biglatticeinterpretation}
Let $X$ denote a Kummer surface with underlying torus $T$, equipped with the induced
dual K\"ahler class. Let $G$ denote the
symmetry group of $X$. Then in Kondo's proof
of Mukai's theorem \mbox{\rm\ref{mukai}}, one can use the Niemeier lattice $N$ of type $A_1^{24}$. 
Moreover, there is an isometry $i_G$ of a primitive sublattice $M_G$ of $H_\ast(X,\Z)$ to a
primitive sublattice of $N(-1)$ which is
equivariant with respect to the natural actions of $G$,
where $M_G$ contains the lattice $L_G=\left(H_\ast(X,\Z)^G\right)^\perp\cap H_\ast(X,\Z)$, the Kummer lattice
$\Pi$ of $X$, and the vector $\upsilon_0-\upsilon$. Here, $G$-equivariance of $i_G$ means that 
$i_G(g_\ast(v))=g_\ast(i_G(v))$ for every $g\in G$ and $v\in M_G$. 
\end{corollary}
This means that our theorem \ref{biglattice} offers an improvement of Kondo's 
techniques for all Kummer surfaces with induced dual K\"ahler class, since
our lattice $M_G$  contains the sublattice $L_G$
which Kondo identifies in his Niemeier lattices and has rank
$\rk(M_G)=\rk(L_G)+2$.

Recall that by means of the Kummer construction, we are using a fixed marking for all
our Kummer surfaces, as was explained at the end of section \ref{exampleK3}. In particular,
we may {smoothly vary the generators $\vec\lambda_1,\ldots,\vec\lambda_4$ of}
the underlying lattice $\Lambda\subset\C^2$ of $T(\Lambda)$ and thus vary 
between distinct Kummer surfaces $\wt{T(\Lambda)/\Z_2}$. This amounts to a variation
along a Kummer path in the smooth connected cover $\wt\MMM_{hk}$ of 
{the moduli space of} hyperk\"ahler
structures, as was explained at the end of section \ref{complex}. Along a generic such
path, both the symmetry group $G$ and the lattice $M_G$ change, but 
as $G$ varies all the lattices $M_G$
share the same sublattice $\Pi\oplus\spann_\Z\{\upsilon_0-\upsilon\}$. Since according to 
proposition \ref{uniquepi} the Kummer lattice $\Pi$ allows a unique primitive embedding $\Pi\hookrightarrow N(-1)$
up to automorphisms of $N$, we can always find an embedding $i_G\colon M_G\hookrightarrow N(-1)$
as in theorem \ref{biglattice} such that $(i_G)_{\mid\Pi}$ is induced by $E_{\vec a}\mapsto f_{I^{-1}(\vec a)}$
for all $\vec a\in\F_2^4$ with $I$ as in \req{map}. For brevity, we say that we choose $i_G$ 
\textsc{with constant $i_G(\Pi)=\wt\Pi$ along Kummer paths in $\wt\MMM_{hk}$}. 
{Let us briefly discuss some properties of $i_G(\Pi)\hookrightarrow M_G\hookrightarrow N(-1)$.
To this end we remark that 
the only automorphisms $\gamma\in\Aut(N)$ which act as the identity 
on the sublattice $\wt\Pi\subset N$
in \req{specialchoice} are compositions of sign flips $f_n\mapsto -f_n$
with $n\in\OOO_9$. Indeed, without loss of generality assume that $\gamma$ acts as
permutation on $\{f_1,\ldots, f_{24}\}$. Then $\gamma_{|\wt\Pi}=\mbox{id}_{\wt\Pi}$ implies that
$\gamma$ is induced by an element of the maximal subgroup
$(\Z_2)^4\rtimes A_8$ of $M_{24}$ which {stabilizes the octad $\OOO_9$}. This group acts
as $\Aff(\F_2^4)=(\Z_2)^4\rtimes \GL_4(\F_2)\cong (\Z_2)^4\rtimes A_8$ on the hypercube
$\F_2^4$ underlying the complement octad of $\OOO_9$ (see \cite[Thm.~2.10]{co71}), thus
in particular its induced action on $\wt\Pi$ is faithful. Hence $\gamma_{|\wt\Pi}=\mbox{id}_{\wt\Pi}$ implies
$\gamma=\mbox{id}$.}
As a consequence, we cannot expect
$i_G\left( \Pi\oplus\spann_\Z\{\upsilon_0-\upsilon\}\right)\subset N(-1)$ to be constant
along Kummer paths in $\wt\MMM_{hk}$. In other words, $i_G(\upsilon_0-\upsilon)$ cannot 
be chosen freely among the roots $\pm f_n$ with $n\in\OOO_9$. We will come
back to this observation in the next section and in the {conclusions}. 
%
\section{The overarching finite symmetry group of Kummer surfaces}\label{overarching}
 %
By Mukai's theorem \ref{mukai}, certain comparatively small subgroups $G$ of 
the Mathieu group $M_{24}$ occur as holomorphic symplectic automorphism
groups of K3 surfaces $X$. The correlation with $M_{24}$
is made more explicit by Kondo's proof of this 
theorem, by which for the smallest primitive sublattice $L_G\subset H_\ast(X,\Z)$ with the 
property that $G$ acts trivially on $(L_G)^\perp$ there also exists a primitive embedding 
of $L_G$
into some Niemeier lattice $\wt N(-1)$, such that the embedding is equivariant with 
respect to a natural action of $G$ on 
$\wt N$. Restricting our attention to Kummer surfaces with their
induced dual K\"ahler classes, our theorem \ref{biglattice}
allows us to replace the lattice $L_G$ by a lattice $M_G$ of higher rank 
(namely, $\rk(M_G)=\rk(L_G)+2$), and to replace the Niemeier lattice $\wt N$, which is not further
specified by Kondo's construction, by the Niemeier lattice $N$ of type $A_1^{24}$.  
In other words, a part of the
K3-homology can be isometrically identified with a sublattice of the Niemeier lattice $N(-1)$,
where the group $M_{24}$, whose role we would like to understand better, acts naturally
(see  proposition \ref{M24onN}). This section is devoted to developing a technique by which we
extend this identification even further, namely to a linear bijection between the two lattices 
$H_\ast(X,\Z)$ and $N(-1)$. This bijection identifies two different primitive sublattices 
$M_{\m T_{192}}$ and $M_{\m T_{64}}$ of
rank $20$ in $H_\ast(X,\Z)$, corresponding to two distinct Kummer surfaces, 
isometrically with their images. Moreover, 
these embeddings are equivariant with respect to the
actions of the respective symmetry groups $\m T_{192}$ on $M_{\m T_{192}}$ and $\m T_{64}$ on
$M_{\m T_{64}}$, which allows us to combine these groups to a bigger one,
yielding the overarching finite symmetry 
group of Kummer surfaces (see definition \ref{overarch}).
This group is by orders of magnitude larger
than the largest holomorphic symplectic automorphism group of any K3 surface.
In fact, we find a Kummer path in $\wt\MMM_{hk}$ with constant $i_G(\Pi)=\wt\Pi$
along it, which connects the two above mentioned Kummer surfaces smoothly
and such that $\Theta$ restricts to a $G$-equivariant isometry
$i_G\colon M_G\hookrightarrow N(-1)$ for each 
Kummer surface along the path.
 %
 \subsection{The translational automorphism group $G_t\cong(\Z_2)^4$}\label{genericKummer}
As a first step of our construction, we consider the translational automorphism
group $G_t$ of definition \ref{defgeneric}, which all Kummer surfaces share. 
We proceed analogously to Kondo's proof of Mukai's theorem \ref{mukai}.
 
If $G=G_t$ is the symmetry group of a Kummer surface $X$
with underlying torus $T$, 
then with notations as in proposition \ref{Kummerform} and theorem \ref{biglattice},
in particular $K=\pi_\ast (H_2(T,\Z))$, we have 
$$
H_\ast(X,\Z)^G=H_0(X,\Z)\oplus K \oplus \spann_\Z\{ e\} \oplus H_4(X,\Z)
$$
and thus $M_G=\Pi\oplus \spann_\Z\{ \upsilon_0-\upsilon\}$, which can be primitively
embedded in $N(-1)$ by theorem \ref{biglattice}. In fact, by proposition \ref{fullk3}, 
we have $M_G=\m P$, where we know how to 
glue $H_\ast(X,\Z)$ from the lattices $\m P$ and $\m K=\m P^\perp\cap H_\ast(X,\Z)$, and 
\req{fullniemeier} shows $\m P\cong\wt{\m P}_{n_0}(-1)$ for a primitive sublattice $\wt{\m P}_{n_0}$ of
$N$, $n_0\in\m O_9$ arbitrary,
along with the relevant gluing prescription for $N$. By the proof of proposition \ref{piinn},
an isometry $\m P\longrightarrow\wt{\m P}_{n_0}(-1)$ can be induced by 
$E_{\vec a}\mapsto f_{I^{-1}(\vec a)}$
for every $\vec a\in\F_2^4$, with $I$ as in \req{map}, along with $\upsilon_0-\upsilon\mapsto f_{n_0}$.
Since $G_t$ acts on $\m P$ by automorphisms, we obtain an induced action
of $G_t\cong(\Z_2)^4$ on $\wt{\m P}_{n_0}$
by enforcing our isometry $\m P\longrightarrow\wt{\m P}_{n_0}(-1)$ to
be $G_t$-equivariant. By construction, $G_t$ acts trivially on 
$\m P^\ast/\m P\cong \wt{\m P}_{n_0}^\ast/\wt{\m P}_{n_0}$,
and thus the action can be extended to one on $N$ acting trivially on the orthogonal
complement of $\wt{\m P}_{n_0}$ (see \cite[Prop.~1.1]{ni80b}). 
The construction yields the following
\begin{prop}\label{translationonN}
Consider a Kummer surface $X$ with underlying torus $T$, equipped with
the induced dual K\"ahler class. Let $G$ denote the symmetry
group of $X$ and $G_t\subset G$ the translational automorphism group according to 
proposition \mbox{\rm\ref{findgeneric}}.

There exists a primitive embedding $i_G\colon M_G\hookrightarrow N(-1)$
of the lattice $M_G\subset H_\ast(X,\Z)$ 
defined in theorem \mbox{\rm\ref{biglattice}} into the Niemeier lattice $N(-1)$ of type $A_1^{24}$ 
with the following properties:
Consider the lattices $\m P$ and
$\wt{\m P}_{n_0}$ with the notations 
of propositions \mbox{\rm\ref{fullk3}} and \mbox{\rm\ref{fullniemeierprop}}, where 
$\m P\subset M_G$ by corollary \mbox{\rm\ref{biglatticeinterpretation}}.
Then 
$$
i_G\colon E_{\vec a}\longmapsto f_{I^{-1}(\vec a)}\quad\mbox{ for every }\quad \vec a\in\F_2^4, 
$$
where $I$ is the map \mbox{\rm\req{map}}, and 
$i_G(\m P)=\wt{\m P}_{n_0}(-1)$ for some $n_0\in\m O_9$.

Moreover, by means of the embedding $i_G\colon M_G\hookrightarrow N(-1)$
the action of the translational automorphism group
$G_t$ on $X$ induces an action of $(\Z_2)^4\subset M_{24}$ on $N$ which is generated
by the following four involutions: 
\be\label{genericc24}
\begin{array}{rcl}
\iota_1&:=& (1,11)(2,22)(4,20)(7,12)(8,17)(10,18)(13,21)(14,16),\\[5pt]
\iota_2&:=& (1,13)(2,12)(4,14)(7,22)(8,10)(11,21)(16,20)(17,18),\\[5pt]
\iota_3&:=& (1,14)(2,17)(4,13)(7,10)(8,22)(11,16)(12,18)(20,21),\\[5pt]
\iota_4&:=& (1,17)(2,14)(4,12)(7,20)(8,11)(10,21)(13,18)(16,22),
\end{array}
\ee
where as in proposition \mbox{\rm\ref{M24onN}} the elements of $M_{24}$ are viewed as permutations
in $S_{24}$ whose action on $N$ is induced by a permutation of the $f_1,\ldots, f_{24}$.
\end{prop}
\begin{proof}
With notations as above, theorem \ref{biglattice}
implies the existence of a primitive embedding $i\colon M_G\hookrightarrow N(-1)$. In particular,
we find a primitive sublattice $\wh{\m P}(-1)\subset i(M_G)$
such that $\wh\PPP$ is isometric to any $\wt{\m P}_{n_0}$ with $n_0\in\m O_9$. 
By proposition \ref{fullniemeierprop}
we have $\wt{\m P}_{n_0}=\wt\Pi\oplus\spann_\Z\left\{f_{n_0}\right\}$ with 
$\wt\Pi\cong\Pi(-1)$ for the Kummer lattice $\Pi$ of \req{Kummer}.
Since the $\pm f_n$ are the 
only vectors in $N$ on which the quadratic form takes value $2$ by \req{nieroot}, 
this implies $\wh{\m P}=\wh\Pi\oplus\spann_\Z\left\{f_{\wh n_0}\right\}$ with $\wh\Pi\cong\Pi(-1)$
and for some $\wh n_0\in\m I$.  
Our uniqueness result in proposition \ref{uniquepi} implies that there is a lattice
automorphism $\alpha$ of $N$ with $\alpha(\wh{\Pi})=\wt{\Pi}$ and
$\alpha( i(E_{I(n)}) )= f_{n}$ for every
$n\in\m I\setminus \m O_9$. In other words, $\alpha\circ i$ gives
a primitive embedding of $M_G$ in $N(-1)$ which embeds the lattice $\Pi$
as claimed. Moreover, $\alpha\circ i(\upsilon_0-\upsilon)=\pm\alpha(f_{\wh n_0})=\pm f_{n_0}$
for some $n_0\in\m O_9$ due to \req{nieroot} and $\upsilon_0-\upsilon \in\Pi^\perp$.
Hence $\alpha\circ i$ or its composition with the automorphism of $N$ induced
by $f_{n_0}\mapsto -f_{n_0}$ (see proposition \ref{M24onN}) yields an embedding
$i_G\colon M_G\hookrightarrow N(-1)$ with the properties claimed.

We now apply the methods described before the statement of our proposition to induce
an action of $G_t\cong(\Z_2)^4$ on $N$ with respect to which the embedding $i_G$
is equivariant. In other words,
with $\vec b_1,\ldots, \vec b_4$ the standard basis of $\F_2^4$,
for each $k\in\{1,\ldots,4\}$ we impose
$I(\iota_k(n)) = I(n)+\vec b_k$ for all $n\in\m I\setminus\m O_9$,
where $I$ is the map \req{map}. This uniquely determines the permutations $\iota_1,\ldots,\iota_4$
of \req{genericc24}, which as one confirms as a 
cross-check\footnote{See appendix \ref{MOG} for a definition of the 
extended binary Golay code $\m G_{24}$ and for the description of a technique that may 
be used to prove that a given permutation of $\m I$ preserves $\m G_{24}$,
for instance by verifying  that its action on the Golay code basis 
\req{Golaybasis}
yields Golay codewords.} are elements of $M_{24}$.
\end{proof}

The precise form of the permutations in \req{genericc24} depends on our specific
choice of primitive embedding for the Kummer lattice $\Pi$ in $N(-1)$, described
in proposition \ref{piinn} and in its proof. Indeed, the image $\wt\Pi(-1)$ of $\Pi$ under
the primitive embedding solely depends on the choice of the special octad $\m O_9$ 
used in \req{specialchoice}. The isometry from $\Pi$ to $\wt\Pi(-1)$ induced by
this embedding then solely depends on the choice of the octads
$\m H_1,\ldots,\m H_5$ in \req{hyperplanes}. As explained in the discussion
of proposition \ref{piinn}, the subspace of the Golay code generated by 
$\m H_1,\ldots,\m H_4$ is dual to the hypercube $\F_2^4$ built on
$\m I\setminus\m O_9$ by means of our map $I$. Then the choice of 
$\m H_1,\ldots,\m H_5$ amounts to the choice of an affine basis of the
underlying affine $4$-dimensional space.

Note that in general, the symmetry group $G$ of a Kummer
surface contains the translational automorphism group $G_t$ as a proper subgroup.
The primitive embedding of $M_G$ in $N(-1)$ found in proposition \ref{translationonN},
where $\wt{\m P}_{n_0}(-1)$ is a primitive sublattice of the image of $M_G$, can then
be used to induce a faithful $G$-action on the Niemeier lattice $N$ by enforcing
the embedding $i_G\colon M_G\hookrightarrow N(-1)$ to be $G$-equivariant, 
by the very same
ideas that lead to proposition \ref{translationonN}. 

We claim that the embedding
$i_G\colon M_G\hookrightarrow N(-1)$ can be extended to a linear bijection {$\theta$}
between
the K3-homology $H_\ast(X,\Z)$ and $N(-1)$. Indeed, $H_\ast(X,\Z)$ 
can be glued from $\m P$ and its orthogonal complement $\m K$, while
$N$ can be glued from $\wt{\m P}_{n_0}$ and its orthogonal complement $\wt{\m K}_{n_0}$,
by means of propositions \ref{fullk3} and \ref{fullniemeierprop}, respectively. 
Since $\m P$ is already isometrically
identified with $\wt{\m P}_{n_0}(-1)$, the extension to a bijection 
${\theta}\colon H_\ast(X,\Z)\longrightarrow N(-1)$ 
amounts to the construction of
a bijection $\m K\rightarrow \wt{\m K}_{n_0}(-1)$ which induces an isometry from
$M_G\cap\m K$ to its image, and which
is compatible with gluing. 
The latter in turn amounts to compatibility with the isometries (up to signature inversion)
$g\colon {\m K}^\ast/{\m K}\longrightarrow {\m P}^\ast/{\m P}$ and 
$\wt g_{n_0}\colon  \wt{\m K}_{n_0}^\ast/\wt{\m K}_{n_0}\longrightarrow\wt{\m P}_{n_0}^\ast/\wt{\m P}_{n_0}$ 
which are so fundamental to the 
gluing procedure, and which in particular, by means of $\m P\cong \wt{\m P}_{n_0}(-1)$,
yield an isometry ${\m K}^\ast/ {\m K}\longrightarrow \wt{\m K}_{n_0}^\ast/\wt{\m K}_{n_0}$. 
In other words, we need to find a lift of 
${\m K}^\ast/ {\m K}\longrightarrow \wt{\m K}_{n_0}^\ast/\wt{\m K}_{n_0}$
to a bijection ${\m K}^\ast\rightarrow \wt{\m K}_{n_0}^\ast(-1)$ which induces
a $G$-equivariant  isometry from
$M_G\cap\m K$ to its image. This is always possible,
and our compatibility conditions severely restrict the number of choices. Indeed,
by construction $M_G\cap \m K$ can be isometrically identified with $(\wt M_G\cap \wt{\m K}_{n_0})(-1)$, where $\wt M_G$
denotes the image of $M_G$ under its primitive embedding into $N$. This  determines the
lift on a sublattice whose rank $7-r_G$ depends on the group $G$. 
Hence we only need to extend the isometry
$M_G\cap \m K\longrightarrow (\wt M_G\cap \wt{\m K}_{n_0})(-1)$ to a lift
of ${\m K}^\ast/ {\m K}\longrightarrow \wt{\m K}_{n_0}^\ast/\wt{\m K}_{n_0}$.
Since
${\m K}^\ast/ {\m K}\cong(\Z_2)^7$ according to proposition \ref{fullk3}, and thus ${\m K}^\ast/ {\m K}$
can be generated by no less than $7$ elements, while the lattice $\m K$ has rank $7$
as well, we are certain to find $r_G=\rk\left( (M_G)^\perp\cap \m K\right)$ 
compatibility conditions for the lift to $(M_G)^\perp\cap \m K$.
{Since the translational group $G_t\cong(\Z_2)^4$ is the symmetry group of a generic 
Kummer surface, any bijection $\theta$ as above is compatible with the generic symmetry
group of Kummer surfaces.}

In our construction we
work with the tetrahedral Kummer surface where $G=\m T_{192}$, a group of order
$192$ (see section \ref{tetrahedral}), and with the Kummer surface associated to the 
square torus where $G=\m T_{64}$,
a group of order $64$ (see section \ref{square}). In both cases, the lattice 
$(M_G)^\perp\cap \m K$, on which our
lift is not uniquely determined up to lattice automorphisms by imposing 
that it induces $i_G\colon M_G\hookrightarrow N(-1)$, has 
the minimal possible rank $r_G=4$. 

Let us discuss the compatibility conditions for our 
bijection ${\theta}\colon H_\ast(X,\Z)\longrightarrow N(-1)$
which already arise from 
the construction so far, imposing equivariance only with respect to 
the translational automorphism group $G_t$. In other words, let us discuss possible lifts of 
${\m K}^\ast/ {\m K}\longrightarrow \wt{\m K}_{n_0}^\ast/\wt{\m K}_{n_0}$ to $\m K^\ast$.
Recall from proposition \ref{fullk3} 
that ${\m K}^\ast/ {\m K}$ is generated by
the $\qu{{1\over2}\pi_\ast\lambda_{ij}}$ along with $\qu{{1\over2}(\upsilon_0+\upsilon)}$
which under the composition of $g\colon {\m K}^\ast/{\m K}\longrightarrow {\m P}^\ast/{\m P}$ 
with the map induced by our isometry $\m P\cong\wt{\m P}_{n_0}(-1)$ are mapped to the
$\qu{\hf\sum_{n\in \wt P_{ij}} f_n}$ and $\qu{\hf f_{n_0}}$, respectively. 
Here, $\wt P_{ij}$ is the quadruplet of labels in $\m I$ which under the map $I$ 
corresponds to the plane $P_{ij}\subset\F_2^4$ in \req{planes}.
Taking preimages under
the map $\wt g_{n_0}\colon  \wt{\m K}_{n_0}^\ast/\wt{\m K}_{n_0}\longrightarrow\wt{\m P}_{n_0}^\ast/\wt{\m P}_{n_0}$
of proposition \ref{fullniemeierprop}, we find that 
${\m K}^\ast/ {\m K}\longrightarrow \wt{\m K}_{n_0}^\ast/\wt{\m K}_{n_0}$
is given by
$$
\qu{{\textstyle\hf}\pi_\ast\lambda_{ij}} \mapsto \qu{q_{ij}}
\mbox{ for }
ij=12,\, 34,\, 13,\, 24,\, 14,\, 23, \quad\quad
\qu{{\textstyle\hf}(\upsilon_0+\upsilon)} \mapsto \qu{{\textstyle\hf} \sum_{n\in\m O_9\setminus\{n_0\}} f_n},
$$
where the $\qu{q_{ij}}$  can be represented, for example, by the $q_{ij}$ of \req{kijchoice}.
We impose the additional constraint that the $\hf\pi_\ast\lambda_{ij}$ and $\hf(\upsilon+\upsilon_0)$ 
should map to 
representatives of elements in $\wt{\m K}_{n_0}^\ast/\wt{\m K}_{n_0}$ of minimal length under
our lift. {This is an aesthetic choice which is not 
required mathematically. If this constraint is not imposed, then in equation \req{signchoices} 
we have to add terms
of the form $2\Delta_{ij}$ with $\Delta_{ij}\in\wt\KKK_{n_0}$ on the right hand side. 
We will see that in the examples studied in this work, it is possible to impose
$\Delta_{ij}=0$ for all $i,\, j$, and then} 
lifting to a bijection $\m K^\ast\rightarrow\wt{\m K}_{n_0}^\ast$ 
amounts to implementing 
\be\label{signchoices}
K\ni \pi_\ast\lambda_{ij}\mapsto \sum_{n\in Q_{ij}} (\pm f_n)\in\wt K,
\quad\quad
\upsilon_0+\upsilon\mapsto \sum_{n\in\m O_9\setminus\{n_0\}} (\pm f_n),
\ee
where $\qu{\hf\sum_{n\in Q_{ij}} (\pm f_n)}= \qu{q_{ij}}\in\wt K^\ast/\wt K$ and thus 
$\hf\sum_{n\in Q_{ij}}(\pm f_n)$ can be glued to $\hf\sum_{n\in \wt P_{ij}} f_n$ in $N$. 
In other words,  $Q_{ij}\subset \m O_9$ must be a quadruplet of labels such that 
$Q_{ij}\cup \wt P_{ij}$ gives an octad in the Golay code $\m G_{24}$. 
In fact, each such quadruplet $Q_{ij}$ must form an octad of the Golay code with every quadruplet of labels 
which under $I$ corresponds to a hypercube plane parallel to $P_{ij}$. This turns out to leave a choice of two complementary 
quadruplets in $\m O_9$ for each label $ij$:
\be\label{Qijchoice}
\begin{array}{l@{\mbox{ or $\;\;$}}l}
Q_{12}=\{3, 6, 15, 19\}& \{5,9,23,24\},\\[5pt]
Q_{13}=\{6, 15, 23, 24\}& \{3,5,9,19\},\\[5pt]
Q_{14}=\{3, 9, 15, 24\}& \{5,6,19,23\},\\[5pt]
Q_{23}=\{3, 9, 15, 23\}& \{5,6,19,24\},\\[5pt]
Q_{24}=\{15, 19, 23, 24\}& \{3,5,6,9\},\\[5pt]
Q_{34}=\{6, 9, 15, 19\}& \{3,5,23,24\},
\end{array}
\ee
where
\be\label{discriminantidentification}
K^\ast/ K \stackrel{\cong}{\longrightarrow} \wt K^\ast/\wt K,\quad\quad
\textstyle
 \qu{{1\over2}\pi_\ast\lambda_{ij}}  \mapsto   \qu{\hf\sum\limits_{n\in Q_{ij}} f_n}.
\ee
Note that in \req{kijchoice} we have chosen the first quadruplet listed in 
\req{Qijchoice} for the $Q_{ij}$, throughout. 
In the following subsections, we 
show that with this choice, surprisingly, {the resulting} bijection $\Theta$ induces 
isometric embeddings of the lattices $M_G$ for two \textsl{distinct} Kummer 
surfaces, the tetrahedral Kummer 
surface  $X_{D_4}$ with $G=\m T_{192}$ and the Kummer surface $X_0$ of the 
square torus with $G=\m T_{64}$.
Our choice of the $Q_{ij}${, up to a few choices of signs,} is in fact unique with the property
that for both groups $G$, the bijection
$\Theta$ induces a $G$-equivariant isometric embedding of  $M_G$ in $N(-1)$.
%
\subsection{The tetrahedral Kummer surface}\label{tetrahedral}
%
By the above, we are looking
for a linear bijection
$\Theta\colon H_\ast(X,\Z)\longrightarrow N(-1)$ which
induces an isometry between as large sublattices of $H_\ast(X,\Z)$ and of $N(-1)$
as possible. By theorem \ref{biglattice}, for a Kummer surface $X$ with its
induced dual K\"ahler class and with symmetry
group $G$, we already know that 
a primitive sublattice $M_G$ of $H_\ast(X,\Z)$, with $\rk(M_G)=\rk(L_G)+2$ for the
lattice $L_G$ found by Kondo, can be primitively embedded in 
$N(-1)$. Note that $\rk(M_{G_1})\geq\rk(M_{G_2})$ if $G_1\supseteq G_2$.
Our construction
therefore sets out from the study of the Kummer surface
whose holomorphic symplectic automorphism group has maximal order among
all Kummer surfaces with their induced dual K\"ahler classes. 
According to proposition \ref{sdproduct}, this amounts to a Kummer surface whose
underlying torus $T=T(\Lambda)$ has the largest group $G_T^\prime$ of 
non-translational holomorphic symplectic automorphisms. 
By Fujiki's results \cite[Lemma 3.1\&3.2]{fu88}, this is the torus $T(\Lambda_{D_4})$
which we call the \textsc{$D_4$-torus}, whose associated Kummer surface
$X_{D_4}$ we call the \textsc{tetrahedral Kummer surface} for reasons to
be explained below. Its group of 
holomorphic symplectic automorphisms\footnote{Note that the full symplectic automorphism group of the tetrahedral Kummer surface
(disregarding the dual K\"ahler class) is infinite, 
see e.g.\ \cite{shin77}.} is the finite group ${\m T}_{192}:=(\Z_2)^4\rtimes A_4$, with 
$A_4$ the group of even permutations on $4$ elements.

This subsection is devoted to the investigation of the tetrahedral Kummer surface
$X_{D_4}$. In particular, we describe its symmetry group
${\m T}_{192}$ as a subgroup of $F_{384}$, one of the $11$ subgroups of 
$M_{24}$ which have the property that every finite group of symplectic automorphisms 
of a K3 surface is isomorphic to a subgroup of one of them by Mukai's
theorem \ref{mukai}. Moreover, we construct a linear bijection 
$\Theta\colon H_\ast(X,\Z)\longrightarrow N(-1)$ 
which induces an isometry of the lattice $M_{{\m T}_{192}}$ of theorem \ref{biglattice}
with its image which is equivariant with respect to $\m T_{192}$.
\subsubsection{The $D_4$-torus and the tetrahedral Kummer surface}\label{d4}
As already mentioned above, we have the
\begin{definition}\label{tetrahedraldef}
Consider the lattice $\Lambda_{D_4}\subset\C^2$ which is 
generated by the four vectors
\be\label{d4gen}
\vec\lambda_1=(1, 0), \quad\vec\lambda_2=(i, 0),\quad 
\vec\lambda_3=(0,1),\quad\vec\lambda_4=\textstyle\hf( i+1,i+1)\quad\in\C^2.
\ee
We call $\Lambda_{D_4}$ the \textsc{$D_4$-lattice}, and 
$T=T(\Lambda_{D_4})$ is the \textsc{$D_4$-torus}. Moreover,
we call the Kummer
surface $X_{D_4}$ with underlying torus $T=T(\Lambda_{D_4})$,
equipped with the induced dual K\"ahler class, which in fact yields a polarization,
the \textsc{tetrahedral Kummer surface}.
\end{definition}
Note that the $D_4$-lattice is isomorphic to the root lattice of the simple Lie algebra 
$\mathfrak{d}_4$, thus our terminology. For the tetrahedral
Kummer surface $X_{D_4}$  our terminology is motivated by its symmetry group. 
Namely, according to \cite[Table 9]{fu88} we have 
\begin{prop}\label{tetrahedralsym}
The group $\m T$ of holomorphic
symplectic automorphisms of the $D_4$-torus has the maximal order $24$ among 
all translation-free groups of holomorphic symplectic automorphisms of 
complex K\"ahler tori. The group $\m T$ is the 
\textsc{binary tetrahedral group}, and its action on the universal cover
$\C^2$ of $T(\Lambda_{D_4})$ with standard
complex coordinates $(z_1,z_2)$ is generated by
\be \label{t4gen}
\begin{array}{rcl}
\gamma_1\colon (z_1, z_2)&\mapsto&(iz_1, -iz_2),\\[5pt]
\gamma_2\colon (z_1, z_2)&\mapsto&(-z_2, z_1),\\[5pt]
\gamma_3\colon (z_1, z_2)&\mapsto&\frac{i+1}{2}( i(z_1-z_2), -(z_1+z_2)).
\end{array}
\ee
\end{prop}
Note that the binary tetrahedral group $\m T$ is a $\Z_2$-extension of the 
group $A_4$ of orientation preserving symmetries of a regular tetrahedron in 
$\R^3$. Indeed, the group $A_4$  of even permutations
on $4$ elements
acts as a subgroup of $\SO(3)$, and $\m T$ is its 
lift to the universal cover $\SU(2)$ of $\SO(3)$.
Clearly, $\gamma_1$ and $\gamma_2$ in \req{t4gen}
have order $4$, while $\gamma_3$ has order $3$,
and as a cross-check one confirms 
that the standard holomorphic $(2,0)$-form 
$\widehat{\Omega}=dz_1 \wedge dz_2$  
and the standard K\"ahler form 
$\wh\omega_T=\frac{1}{2i}(dz_1\wedge d\bar{z}_1+dz_2\wedge d\bar{z}_2)$ 
of $T(\Lambda_{D_4})$ as in \req{cohomologyinv}, \req{standardKaehler} 
are invariant under $\m T$. 
Hence indeed, $\m T$ acts as a holomorphic symplectic automorphism group on $T(\Lambda_{D_4})$. 
Since $\gamma_1^2=\gamma_2^2$, for example, yields the $\Z_2$-orbifold action 
$(z_1,z_2)\mapsto (-z_1,-z_2)$ used in the Kummer construction of definition \ref{defKummer}, 
the action of $\m T$ on $T(\Lambda_{D_4})$ induces a symplectic action of 
the tetrahedral symmetry group
$A_4=\m T/\{\pm1\}$ on the corresponding Kummer surface 
$X_{D_4}=\wt{T(\Lambda_{D_4})/\Z_2}$, thus our terminology. 
Moreover, $\wh\omega_T$ induces the polarization $\omega$ of this K3 surface. 
We have the following
\begin{prop}\label{tetrahedralcpxstr}
Consider the tetrahedral Kummer surface $X_{D_4}$ of definition \mbox{\rm\ref{tetrahedraldef}}
with its induced polarization, and with notations as in definition \mbox{\rm\ref{defKummer}}.

The symmetry group of $X_{D_4}$
is the group $\m T_{192}:=(\Z_2)^4\rtimes A_4$, where $A_4$ denotes the 
group of even permutations on $4$ elements. 

By means of the Torelli Theorem \mbox{\rm\ref{torellithm}} and with respect to
our usual marking described at the end of section \mbox{\rm\ref{exampleK3}}, 
the complex structure of $X_{D_4}$
is specified by the $2$-dimensional oriented subspace $\Omega\subset H_2(X_{D_4},\R)$
with orthonormal basis 
$$
\begin{array}{rcl}
\Omega_1&=& -\pi_\ast\lambda_{12}+\pi_\ast\lambda_{13}+\pi_\ast\lambda_{23}-2\pi_\ast\lambda_{24},\\[5pt]
\Omega_2&=&-\pi_\ast\lambda_{12}-\pi_\ast\lambda_{13}+\pi_\ast\lambda_{23}+2\pi_\ast\lambda_{14},
\end{array}
$$
where $\lambda_{ij}=\lambda_i\vee\lambda_j$ as in \mbox{\rm\req{lambdaij}}
with generators $\lambda_1, \ldots, \lambda_4$ of $H_1(T_{D_4},\Z)$ represented
by the $\vec\lambda_1,\ldots,\vec\lambda_4$ of \mbox{\rm\req{t4gen}}.
Furthermore, the induced polarization is 
$$
\omega
=\pi_\ast\lambda_{12}+\pi_\ast\lambda_{13}+\pi_\ast\lambda_{23}+2\pi_\ast\lambda_{34}
\in H_2(X_{D_4},\Z).
$$
\end{prop}
\begin{proof}
By proposition \ref{tetrahedralsym}, the group $G^\prime_{{D_4}}$ of 
holomorphic symplectic automorphisms of $T(\Lambda_{D_4})$ that fix
$0\in\C^2/\Lambda_{D_4}$ is the binary tetrahedral group $G^\prime_{{D_4}}=\m T$
with $G_T=G^\prime_{{D_4}}/\Z_2=A_4$ by the above. Hence proposition \ref{sdproduct}
shows that $\m T_{192}=(\Z_2)^4\rtimes A_4$ is the holomorphic
symplectic automorphism group of the tetrahedral Kummer surface. One checks
that this group has order $192$, thus the notation.

As explained in 
section \ref{complex}, in terms of local coordinates 
the complex structure and polarization of $X_{D_4}$ are determined by
$\Omega_1,\, \Omega_2,\, \omega\in H_2(X_{D_4},\R)$ as given in \req{homologyinv1},
\req{homologyinv2},
where $\vec e_1,\ldots,\vec e_4$ are the standard 
Euclidean coordinate vectors of $\R^4\cong\C^2$.  
Here, indeed, we may use standard coordinates on the universal cover
of $T(\Lambda_{D_4})$ to induce local coordinates on $T(\Lambda_{D_4})/\Z_2$ away
from the singular points. The generators \req{d4gen} of the lattice 
$\Lambda_{D_4}$ are
$\vec {\lambda}_1=\vec e_1,\, \vec {\lambda}_2=\vec e_2, \,\vec {\lambda}_3=\vec e_3$ and 
$\vec {\lambda}_4 =\hf (\vec e_1+\vec e_2+\vec e_3+\vec e_4)$. 
Inserting these expressions in \req{homologyinv1},
\req{homologyinv2}, one obtains the formulas for $\Omega_1, \Omega_2, \omega$
in terms of the images $\pi_\ast\lambda_{ij}\in H_2(X_{D_4},\Z)$ of the 
$\lambda_{ij}$ as claimed.
\end{proof}

By proposition \ref{tetrahedralcpxstr},  we in particular have
$\Omega_1,\,\Omega_2,\,\omega\in H_2(X_{D_4},\Z)$, such that analogously to the 
example of the Kummer surface $X_0$ with underlying torus $T_0=\C^2/\Z^4$ discussed 
in section \ref{complex}, the results of Shioda and Inose \cite{shin77} apply: 
The quadratic form associated to the transcendental lattice 
$\Omega\cap H_2(X_{D_4},\Z)$ of $X_{D_4}$ uniquely determines the complex 
structure of this Kummer surface. Since generators of this lattice are given by
\be\label{homology}
\begin{array}{rcl}
I_1:=\hf (\Omega_1+\Omega_2)&=&-\pi_\ast\lambda_{12}+\pi_\ast\lambda_{14}+\pi_\ast\lambda_{23}-\pi_\ast\lambda_{24},\\[5pt]
I_2:=\hf (\Omega_1-\Omega_2)&=&\pi_\ast\lambda_{13}-\pi_\ast\lambda_{14}-\pi_\ast\lambda_{24},
\end{array}
\ee
the relevant quadratic form is $\langle I_i,I_j\rangle =4\delta_{ij}$, in agreement with
\req{standardKummer}. In other words, $X_{D_4}$ and the Kummer surface $X_0$ 
constructed from the square torus $T_0$ share the same complex structure, and 
according to the final remark of \cite{shin77}, they agree with the elliptic modular 
surface of level $4$ defined over $\Q\left(\sqrt{-1}\right)$ of \cite[p.~57]{sh72}. 
However, our tetrahedral Kummer surface comes equipped with the polarization 
$\omega$, which is invariant under the action of the group $\m T_{192}$
described in proposition \ref{tetrahedralcpxstr}. 
With $\Sigma:=\spann_\R\{\Omega_1,\Omega_2,\omega\}$, the lattice 
$\Sigma\cap H_2(X_{D_4},\Z)$ has generators $I_1,\, I_2$ as above and 
\be\label{homology3}
I_3:=\textstyle\hf (\Omega_2+\omega)=\pi_\ast\lambda_{14}+\pi_\ast\lambda_{23}+\pi_\ast\lambda_{34},
\ee
such that the associated quadratic form is
$$
\left( \begin{array}{rrr}4&0&2\\0&4&-2\\2&-2&4\end{array}\right),
$$
in contrast to \req{standardPolar}. 
In other words, the Kummer surfaces $X_0$ and $X_{D_4}$ carry 
different induced polarizations.

For later convenience we note that the 
following three vectors generate the lattice $\Sigma^\perp\cap\pi_\ast (H_2(T(\Lambda_{D_4}),\Z))$
of rank $3$:
\be \label{homologyperp}
\begin{array}{rcr}
I_1^{\perp}&:=&\pi_\ast\lambda_{14}+\pi_\ast\lambda_{24}-\pi_\ast\lambda_{23},\\[5pt]
I_2^{\perp}&:=&\pi_\ast\lambda_{13}+\pi_\ast\lambda_{24}+\pi_\ast\lambda_{34},\\[5pt]
I_3^{\perp}&:=&-\pi_\ast\lambda_{12}+\pi_\ast\lambda_{14}+\pi_\ast\lambda_{34}.
\end{array}
\ee
 \subsubsection{The holomorphic symplectic automorphism group of the tetrahedral 
 Kummer surface as a subgroup of $M_{24}$}
As was explained in section \ref{genericKummer}, our theorem \ref{biglattice} implies
that the holomorphic symplectic automorphism group $\m T_{192}$ of our
tetrahedral Kummer surface $X_{D_4}$ also acts faithfully on the Niemeier lattice
$N$ of type $A_1^{24}$. In the following, we calculate this action, and 
thereby we identify $\m T_{192}$ as a subgroup of the Mathieu group $M_{24}$. 
Throughout, we use the same notations as in theorem \ref{biglattice}.

Recall that the action of $\m T_{192}$ on $N$ can be induced by a primitive embedding
of the lattice $M_{\m T_{192}}$ in $N(-1)$ such that its image $\widetilde M_{\m T_{192}}$
contains one of the lattices $\wt{\m P}_{n_0}$ 
of proposition \ref{fullniemeierprop}, $n_0\in\m O_9$. Such an
embedding exists according to proposition \ref{translationonN}, which in particular states that
we may require that
the embedding of the Kummer lattice $\Pi\subset M_{\m T_{192}}$  in $N(-1)$ is induced by
$E_{I(n)}\mapsto f_n$ for every $ n\in\m I\setminus\m O_9$,
by means of the map $I$ of \req{map}. 
Moreover, according
to that proposition, the involutions $\iota_1, \ldots, \iota_4\in M_{24}$ of \req{genericc24} then
yield the action which the translational automorphism group $G_t\cong(\Z_2)^4$
in $\m T_{192}$ induces on $N$, if we impose equivariance of our embedding
$i_{\m T_{192}}\colon M_{\m T_{192}}\hookrightarrow N(-1)$ with respect to $\m T_{192}$.
By proposition \ref{tetrahedralcpxstr}, the symplectic action of the group
$\m T_{192}$ on the tetrahedral Kummer surface is generated by this translational 
automorphism group along with a symplectic action of $A_4$, which is induced by
the symplectic action of the binary tetrahedral group $\m T$ on $T(\Lambda_{D_4})$. 
With respect to standard Euclidean coordinates on $\R^4\cong\C^2$, by \req{t4gen} 
the generators of the latter group are
\be \label{t4action}
\begin{array}{rcl}
\gamma_1(x_1,x_2,x_3,x_4)&=&(-x_2,x_1,x_4,-x_3)\\[5pt]
\gamma_2(x_1,x_2,x_3,x_4)&=&(-x_3,-x_4,x_1,x_2)\\[5pt]
\gamma_3(x_1,x_2,x_3,x_4)&=&\hf([-x_1-x_2+x_3+x_4], [x_1-x_2-x_3+x_4], \\[5pt]
&&\qquad[-x_1+x_2-x_3+x_4], [-x_1-x_2-x_3-x_4]).
\end{array}
\ee
On the generators $\vec {\lambda}_1, \ldots,\vec\lambda_4$ of $\Lambda_{D_4}$ given in \req{d4gen} we thus have
\be\label{t24onli}
\begin{array}{lll}
\gamma_1\colon 
&\vec \lambda_1\mapsto \vec \lambda_2,& \vec \lambda_2\mapsto -\vec \lambda_1,
\quad\quad \vec \lambda_3\mapsto -2\vec \lambda_4+\vec \lambda_1+\vec \lambda_2+\vec \lambda_3,\\[5pt]
&&\vec \lambda_4\mapsto -\vec \lambda_4+\vec \lambda_2+\vec \lambda_3,\\[5pt]
\gamma_2\colon 
& \vec \lambda_1\mapsto \vec \lambda_3,
&\vec \lambda_2\mapsto 2\vec \lambda_4-\vec \lambda_1-\vec \lambda_2-\vec \lambda_3,\\[5pt]
&\vec \lambda_3\mapsto -\vec \lambda_1,
&\vec \lambda_4\mapsto \vec \lambda_4-\vec \lambda_1-\vec \lambda_2,\\[5pt]
\gamma_3\colon 
& \vec \lambda_1\mapsto \vec \lambda_2-\vec \lambda_4,
&\vec \lambda_2\mapsto \vec \lambda_3-\vec \lambda_4,\\[5pt]
&\vec \lambda_3\mapsto \vec \lambda_1-\vec \lambda_4,
&\vec \lambda_4\mapsto-2\vec \lambda_4+\vec \lambda_1+\vec \lambda_2+\vec \lambda_3.
\end{array}
\ee
These transformations induce permutations of the singular points in $T(\Lambda_{D_4})/\Z_2$ 
and thus affine linear maps $[\gamma_1], [\gamma_2], [\gamma_3]$ on the  
hypercube $\mathbb{F}_2^4$ which labels the $E_{\vec a}\in\Pi$ of 
proposition \ref{Kummerform}. For the induced permutations $\wh\gamma_k$ on 
$\m I\setminus\m O_9$, imposing equivariance of our embedding $i_{\m T_{192}}$ 
then amounts to
the condition $\wh\gamma_k(n) = [\gamma_k](I(n))$ for all $n\in \m I\setminus\m O_9$. We obtain
\be\label{gammahat}
\begin{array}{rcl}
\wh{\gamma}_1&=&(2, 8)(7,18)(10,22)(11,13)(12,17)(14,20),\\[5pt]
\wh{\gamma}_2&=&(2,18)(7,8)(10,17)(11,14)(12,22)(13,20),\\[5pt]
\wh{\gamma}_3&=&(2, 12, 13)(4, 16, 21)(7, 17, 20)(8, 22, 14)(10, 11, 18).
\end{array}
\ee
These permutations must be accompanied by appropriate permutations $\sigma_k$,
$k\in\{1,\,2,\,3\}$, of the
labels in $\m O_9$ in order to yield automorphisms of the Niemeier lattice $N$
or equivalently of the Golay code $\m G_{24}$. However, one checks that for each
$\wh\gamma_k$, there exists a unique permutation $\sigma_k$ of $\m O_9$ yielding
$\gamma_k:=\sigma_k\circ\wh\gamma_k\in M_{24}$. Altogether we
obtain
\begin{prop}\label{t192inM24}
Consider the tetrahedral Kummer surface $X_{D_4}$ of definition \mbox{\rm\ref{tetrahedraldef}} with its
induced polarization. Let $i_{\m T_{192}}\colon M_{\m T_{192}}\hookrightarrow N(-1)$ denote
a primitive embedding of the type constructed in proposition \mbox{\rm\ref{translationonN}}. 
This embedding is equivariant with respect to the 
holomorphic symplectic automorphism group $\m T_{192}$, where the
action of this group on $N$  is generated by
 \be\label{t192}
\begin{array}{rcl}
\iota_1&=& (1,11)(2,22)(4,20)(7,12)(8,17)(10,18)(13,21)(14,16),\\[5pt]
\iota_2&=& (1,13)(2,12)(4,14)(7,22)(8,10)(11,21)(16,20)(17,18),\\[5pt]
\iota_3&=& (1,14)(2,17)(4,13)(7,10)(8,22)(11,16)(12,18)(20,21),\\[5pt]
\iota_4&=& (1,17)(2,14)(4,12)(7,20)(8,11)(10,21)(13,18)(16,22),\\[5pt]
\gamma_1&=&(2,8)(7,18)(9,24)(10,22)(11,13)(12,17)(14,20)(15,19),\\[5pt]
\gamma_2&=&(2,18)(7,8)(9,19)(10,17)(11,14)(12,22)(13,20)(15,24),\\[5pt]
\gamma_3&=&(2,12,13)(4,16,21)(7,17,20)(8,22,14)(9,19,24)(10,11,18).
\end{array}
\ee
In particular, since $f_{n_0}$ is invariant under the action of $\m T_{192}$,
we find $n_0\in\{ 3, 5, 6, 23\}$.
\end{prop}
The above proposition identifies $\m T_{192}$ as a subgroup 
of $M_{24}$
which leaves the octad ${\m O}_9=\{9, 5, 24, 19, 23, 3, 6, 15\}$ invariant.
This group may be constructed as a succession of stabilizers of the Mathieu group $M_{24}$, 
starting with the stabilizer  in $M_{24}$ of the element $5\in\m I$, which is $M_{23}$, followed 
by the stabilizer in $M_{23}$ of the element $3\in\m I$, which is $M_{22}$. The next step is 
to construct the stabilizer in $M_{22}$ of the element $6\in\m I$, which is $\PSL(3,4)$, and 
the stabilizer in $\PSL(3,4)$ of the element $23\in\m I$, whose structure is 
$(\Z_2)^4\rtimes A_5$, and finally to obtain the stabilizer in that group of the set 
$\{9,15,19,24\}$. This last stabilizer group has order $192$ and coincides with the 
copy of ${\m T}_{192}$ generated above.
Equivalently, one obtains $\m T_{192}$ from $M_{23}$, the stabilizer of 
$5\in\m I$, by first stabilizing
the set $\m O_9$. This yields the maximal subgroup 
$(\Z_2)^4\rtimes A_7$ of $M_{23}$ which we will recover in 
section \ref{overmen}. Then
successively stabilizing $3,\, 6,\, 23\in\m I$ in that group one obtains the (maximal) subgroups 
$(\Z_2)^4\rtimes A_6,\, (\Z_2)^4\rtimes A_5,\, (\Z_2)^4\rtimes A_4=\m T_{192}$ of
$M_{22},\, \PSL(3,4),\, (\Z_2)^4\rtimes A_5$, respectively.

We remark that the Mathieu group $M_{24}$ may be generated from the set of $7$ 
permutations given in \req{t192}, augmented by one extra involution, for instance 
$$
\iota_5= (1,9)(2,5)(3,19)(4,15)(6,22)(7,18)(8,20)(10,17)(11,12)(13,16)(14,24)(21,23).
$$
This involution is one of the $7$ involutions that are seen on the Klein map, and that generate $M_{24}$ according to \cite{cu07}.

Now recall our discussion in section \ref{genericKummer} leading to 
\req{signchoices}--\req{discriminantidentification}, the consistency conditions 
for an extension of the $G_t$-equivariant embedding $i_{\m T_{192}}$
of $M_{\m T_{192}}$ in 
$N(-1)$ to a linear bijection $\Theta\colon H_\ast(X,\Z)\longrightarrow N(-1)$.
By proposition \ref{t192inM24}, the embedding $i_{\m T_{192}}$ is actually equivariant with respect
to the larger group $\m T_{192}$ and imposes an isometry of $M_{\m T_{192}}\cap \m K$
with its image. This amounts to the condition that the map
\req{signchoices} is equivariant with respect to 
the action of $A_4$ induced by \req{t4gen}
on the $\pi_\ast\lambda_{ij}\in H_2(X_{D_4},\Z)$ and by the 
$\gamma_k$ of \req{t192}
on our choices of $Q_{ij}$ from \req{Qijchoice}. 
Thus, next we need to determine the action of each $\gamma_k$ induced from \req{t24onli}
on the $\lambda_{ij}$: 
$$ 
\left( \begin{array}{c} \gamma_1(\lambda_{12})\\ \gamma_1(\lambda_{13})\\ \gamma_1(\lambda_{14})\\ \gamma_1(\lambda_{23})\\  \gamma_1(\lambda_{24})\\ \gamma_1(\lambda_{34})  \end{array}\right)  =
\left(  \begin{array}{cccccc} 1&0&0&0&0&0\\-1&0&0&1&-2&0\\0&0&0&1&-1&0\\-1&-1&2&0&0&0
\\-1&-1&1&0&0&0\\1&1&-1&0&1&1 \end{array}\right) \left( \begin{array}{c} 
\lambda_{12}\\  \lambda_{13}\\ \lambda_{14}\\ \lambda_{23}\\ \lambda_{24}\\\lambda_{34} \end{array}\right),
$$
$$
\left( \begin{array}{c} \gamma_2(\lambda_{12})\\ \gamma_2(\lambda_{13})\\ \gamma_2(\lambda_{14})\\ \gamma_2(\lambda_{23})\\  \gamma_2(\lambda_{24})\\ \gamma_2(\lambda_{34})  \end{array}\right) = 
\left(  \begin{array}{cccccc} 0&1&0&1&0&2\\0&1&0&0&0&0\\0&1&0&1&0&1\\-1&-1&2&0&0&0
\\0&-1&1&-1&1&-1\\1&0&-1&0&0&0 \end{array}\right) \left( \begin{array}{c} 
\lambda_{12}\\  \lambda_{13}\\ \lambda_{14}\\ \lambda_{23}\\ \lambda_{24}\\\lambda_{34} \end{array}\right),
$$
$$
\left( \begin{array}{c} \gamma_3(\lambda_{12})\\ \gamma_3(\lambda_{13})\\ \gamma_3(\lambda_{14})\\ \gamma_3(\lambda_{23})\\  \gamma_3(\lambda_{24})\\ \gamma_3(\lambda_{34})  \end{array}\right)  =
\left(  \begin{array}{cccccc} 0&0&0&1&-1&1\\-1&0&1&0&-1&0\\-1&0&1&1&-1&1\\0&-1&1&0&0&-1
\\0&-1&1&-1&1&-1\\1&1&-1&0&1&1 \end{array}\right) \left( \begin{array}{c} 
\lambda_{12}\\  \lambda_{13}\\ \lambda_{14}\\ \lambda_{23}\\ \lambda_{24}\\\lambda_{34} \end{array}\right).
$$
The task now is to make a choice for each $Q_{ij}$ in \req{Qijchoice}
such that in the map \req{signchoices} there exists a choice of signs making it
equivariant under $\gamma_1,\,\gamma_2,\,\gamma_3$.
A first step towards a solution to this task, though not uniquely determined at this stage, is 
\begin{prop}\label{Ibarsolution}
The following choices in \mbox{\rm\req{Qijchoice}} yield 
a map $\qu I(\pi_\ast\lambda_{ij}):= Q_{ij}$ with 
$\qu I(\lambda+\lambda^\prime):=\qu I(\lambda)+\qu I(\lambda^\prime)$ 
by means of symmetric differences of sets which for each $k=1,\,2,\,3$ obeys 
$\qu I(\pi_\ast\gamma_k(\lambda_{ij})) 
=\gamma_k (Q_{ij})$ for all labels $ij$:
\be\label{identifymuij}
\begin{array}{rclrcl}
Q_{12}&=&\{3, 6, 15, 19\}, &Q_{34}&=&\{6, 9, 15, 19\}, \\[5pt]
Q_{13}&=&\{6, 15, 23, 24\}, & Q_{24}&=&\{15, 19, 23, 24\}, \\[5pt]
Q_{14}&=&\{3, 9, 15, 24\}, & Q_{23}&=&\{3, 9, 15, 23\}. 
\end{array}
\ee
All these $Q_{ij}$ avoid the label $n_0:=5\in\m O_9$. 
\end{prop}
In section \ref{overmen} we shall show that there exists 
a linear bijection $\Theta\colon H_\ast(X,\Z)\longrightarrow N(-1)$
which is compatible with the overarching finite symmetry group $(\Z_2)^4\rtimes A_7$
of definition \ref{overarch}  in the following sense:
On two distinct sublattices $M_G$ of rank $20$ of $H_\ast(X,\Z)$ the map $\Theta$ restricts 
to isometric embeddings $M_G\hookrightarrow N(-1)$ that are equivariant 
with respect to appropriate holomorphic symplectic
automorphism groups $G$ that generate $(\Z_2)^4\rtimes A_7$. 
Under both groups $G$, the vector $\upsilon_0-\upsilon\in M_G$ is invariant.
It follows that in order for such a map $\Theta$ to exist, there must be
a label $n_0\in\m O_9$ which occurs in none of the $Q_{ij}$. Indeed,
$f_{n_0}$ is the image of $\upsilon_0-\upsilon$ under the
embedding $i_{\m T_{192}}$. Since the symmetry group of our Kummer surface stabilizes the generators
$\upsilon_0$ and $\upsilon$ of $H_0(X,\Z)$ and $H_4(X,\Z)$, and 
since $i_{\m T_{192}}$ is equivariant with respect to $\m T_{192}$, 
the root $f_{n_0}$ must be stabilized by $\m T_{192}$. If $f_{n_0}$
is stabilized by the induced action of the entire overarching finite symmetry 
group, then $f_{n_0}$ cannot occur in any of the
$q_{ij}$. In other words, the label $n_0\in\m O_9$ cannot occur in any of the $Q_{ij}$.
\subsubsection{A linear bijection between even unimodular lattices}
For the tetrahedral Kummer surface $X_{D_4}$ of definition \ref{tetrahedraldef}
with holomorphic symplectic automorphism group $\m T_{192}$, 
in the following we shall show that the primitive embedding
$i_{\m T_{192}}\colon M_{\m T_{192}}\hookrightarrow N(-1)$ of proposition
\ref{translationonN}, which yields $E_{\vec a}\mapsto f_{I^{-1}(\vec a)}$ for each $\vec a\in\F_2^4$
and $\upsilon_0-\upsilon\mapsto f_{n_0}$,
can be extended to a linear bijection $\Theta\colon H_\ast(X,\Z)\longrightarrow N(-1)$
which obeys \req{signchoices} with the choices of proposition \ref{Ibarsolution}
and with $n_0=5$.

First, by proposition \ref{Ibarsolution} the map $\qu I$  yields
\be\label{pairidentification}
\begin{array}{rcl}
\qu I(\pi_\ast\lambda_{12}+\pi_\ast\lambda_{34}+\pi_\ast\lambda_{14}) &=& \{ 15,24\},\\[5pt]
\qu I(\pi_\ast\lambda_{12}+\pi_\ast\lambda_{34}+\pi_\ast\lambda_{23}) &=& \{ 15,23\},\\[5pt]
\qu I( \pi_\ast\lambda_{13}+\pi_\ast\lambda_{24}+\pi_\ast\lambda_{12}) &=& \{ 3,15\},\\[5pt]
\qu I( \pi_\ast\lambda_{13}+\pi_\ast\lambda_{24}+\pi_\ast\lambda_{34}) &=& \{ 9,15\},\\[5pt]
\qu I( \pi_\ast\lambda_{14}+\pi_\ast\lambda_{23}+\pi_\ast\lambda_{13}) &=& \{ 6,15\},\\[5pt]
\qu I( \pi_\ast\lambda_{14}+\pi_\ast\lambda_{23}+\pi_\ast\lambda_{24}) &=& \{ 15,19\}.
\end{array}
\ee
Therefore, if a map $\Theta$ exists as claimed, then it obeys
$\Theta\colon \pi_\ast\lambda_{12}\pm \pi_\ast\lambda_{34}
\pm \pi_\ast\lambda_{14}\mapsto\pm f_{15}\pm f_{24}$, 
etc., with four signs to be chosen for each such identification. 
{If we need to relax the constraint posed in \req{signchoices}, namely
that every ${1\over2}\pi_\ast\lambda_{ij}$ is mapped to a representative of
an element in $\wt\KKK_{n_0}^\ast/\wt\KKK_{n_0}$ of minimal length, then 
additional summands of the form $2\Delta$ with $\Delta\in\wt K$ may occur.
Our} choices of signs 
are severely restricted by imposing \req{signchoices}, 
and one checks that the following yields a lift of 
\req{discriminantidentification} which is consistent with all gluing prescriptions, 
with \req{signchoices} as well as the above \req{pairidentification}
\be\label{doublef}
\begin{array}{rcrcl}
I_3^\perp&=&-\pi_\ast\lambda_{12}+\pi_\ast\lambda_{34}+\pi_\ast\lambda_{14}
&\longmapsto& f_{24}-f_{15},\\[5pt]
J_1^\perp-I_3^\perp&=&\pi_\ast\lambda_{12}-\pi_\ast\lambda_{34}-\pi_\ast\lambda_{23} 
&\longmapsto& f_{15}-f_{23},\\[5pt]
-J_3^\perp+I_2^\perp&=&\pi_\ast\lambda_{13}+\pi_\ast\lambda_{24}+\pi_\ast\lambda_{12}
&\longmapsto& f_3-f_{15},\\[5pt]
I_2^\perp&=&\pi_\ast\lambda_{13}+\pi_\ast\lambda_{24}+\pi_\ast\lambda_{34}
&\longmapsto& f_9-f_{15},\\[5pt]
J_2^\perp-I_1^\perp&=&-\pi_\ast\lambda_{14}+\pi_\ast\lambda_{23}+\pi_\ast\lambda_{13}
&\longmapsto& f_{15}-f_6,\\[5pt]
I_1^\perp&=&\pi_\ast\lambda_{14}-\pi_\ast\lambda_{23}+\pi_\ast\lambda_{24}
&\longmapsto& f_{19}-f_{15},
\end{array}
\ee
(see \req{homologyperp} and \req{doublefreduced} for the definitions of the 
$I_k^\perp$ and $J_k^\perp$), or equivalently
\be\label{lambdaq}
\begin{array}{rcrcr}
\Theta\colon \pi_\ast \lambda_{12}
&\longmapsto&  2q_{12}&=& f_3+f_6-f_{15}-f_{19},\\[5pt]
\Theta\colon \pi_\ast \lambda_{34}
&\longmapsto& 2q_{34}&=&f_6+f_9-f_{15}-f_{19},\\[5pt]
\Theta\colon \pi_\ast \lambda_{13} 
&\longmapsto& 2q_{13}&=&-f_6+f_{15}-f_{23}+f_{24},\\[5pt]
\Theta\colon \pi_\ast \lambda_{24}
&\longmapsto& 2q_{24}&=&-f_{15}+f_{19}+f_{23}-f_{24},\\[5pt]
\Theta\colon \pi_\ast \lambda_{14}
&\longmapsto& 2q_{14}&=&f_{3}-f_{9}-f_{15}+f_{24},\\[5pt]
\Theta\colon \pi_\ast \lambda_{23}
&\longmapsto& 2q_{23}&=&f_{3}-f_{9}-f_{15}+f_{23},
\end{array}
\ee
where $2q_{ij}=\sum_{n\in Q_{ij}} (\pm) f_n$ in accord with \req{signchoices}
and \req{identifymuij}.
In \req{tetracompare} we will see that this choice of signs allows us to isometrically
identify the lattice generated by $I_1^\perp,\, I_2^\perp,\, I_3^\perp\in M_{\m T_{192}}$ with its image in
$N(-1)$. 
In particular, the choices of signs on the left hand side of \req{doublef} are already
fixed by enforcing this property together with \req{signchoices}.
Similarly, in \req{squarecompare} we will see that our choice of signs
also allows us to isometrically
identify the lattice generated by $J_1^\perp,\, J_2^\perp,\, J_3^\perp\in M_{\m T_{64}}$ with its image in
$N(-1)$. 
Our choice of signs is unique with these properties, up to the freedom for each 
$n\in\m I$ to replace $f_n$ by $-f_n$ everywhere, which induces an isometry of $N$
according to proposition \ref{M24onN}. 
This freedom of choice is in accord with the fact that we are actually interested in
$M_{24}$, where by proposition \ref{M24onN} we have
$M_{24} = \Aut(N)/(\Z_2)^{24}$, with $(\Z_2)^{24}$ implementing precisely the freedom
of choice of replacing any $f_n$ by $-f_n$ everywhere.

Finally, with our choice $n_0=5$ from proposition \ref{Ibarsolution},
$\Theta(\upsilon_0-\upsilon)=f_5$, and \req{signchoices}
lifts to
$$
\Theta\colon \upsilon_0+\upsilon \longmapsto f_3+f_6+f_9-f_{15}-f_{19}-f_{23}-f_{24}.
$$
The signs on the right hand side are arbitrary, up to the fact 
that the right hand
side together with the $q_{ij}$ in \req{lambdaq} must generate the lattice
$\wt{\m K}_{n_0}$ of proposition \ref{fullniemeierprop} with $n_0=5$,
$\wt{\m K}_5=\spann_\Z\left\{ f_3,\, f_6,\, f_9,\-f_{15},\, f_{19},\, f_{23},\, f_{24}\right\}$
(see the proof of proposition \ref{fullniemeierprop}).
We thus have
\be\label{upsilonf}
\begin{array}{rcl}
\Theta\colon \upsilon_0 &\longmapsto& \hf \left(f_3+f_5+f_6+f_9-f_{15}-f_{19}-f_{23}-f_{24}\right),\\[5pt]
\Theta\colon \upsilon &\longmapsto&  \hf \left(f_3-f_5+f_6+f_9-f_{15}-f_{19}-f_{23}-f_{24}\right).
\end{array}
\ee
Collecting the various ingredients, we claim 
\begin{theorem}\label{Thetaok}
Consider the tetrahedral Kummer surface $X_{D_4}$ of definition \mbox{\rm\ref{tetrahedraldef}}
with its induced polarization and 
the symmetry group
$\m T_{192}$. 

There exists a linear bijection 
$\Theta\colon H_\ast(X,\Z)\longrightarrow N(-1)$ which is induced by
\mbox{\rm\req{lambdaq}} and \mbox{\rm\req{upsilonf}} along with the map $I$ of \mbox{\rm\req{map}} which induces
\be\label{finalpi}
\Theta\colon E_{\vec a}\longmapsto f_{I^{-1}(\vec a)} \quad \mbox{ for every } \vec a\in\F_2^4.
\ee
The map $\Theta$ isometrically embeds the lattice
$M_{\m T_{192}}$ of theorem \mbox{\rm\ref{biglattice}} into $N(-1)$, such that the
restriction $i_{\m T_{192}}$ of $\Theta$ to $M_{\m T_{192}}$ yields
an embedding $i_{\m T_{192}}\colon M_{\m T_{192}}\hookrightarrow N(-1)$
of the type constructed in proposition \mbox{\rm\ref{translationonN}}. Moreover, $i_{\m T_{192}}$ is 
equivariant with respect to the group
$\m T_{192}$, where on $H_\ast(X,\Z)$ the action of $\m T_{192}$ is induced by its action 
on $X_{D_4}$ as holomorphic symplectic automorphism group, and
on the Niemeier lattice $N$ of type $A_1^{24}$ its action is generated by 
$\iota_1, \ldots, \iota_4,\, \gamma_1,\, \gamma_2,\,\gamma_3$ as in \mbox{\rm\req{t192}}. 
The lattice $M_{\m T_{192}}$ has rank $20$.
\end{theorem}
\begin{proof}
We first need to show that \req{lambdaq}, \req{upsilonf}, \req{finalpi} can be extended
to a linear bijection $\Theta\colon H_\ast(X,\Z)\longrightarrow N(-1)$. With notations as in 
proposition \ref{fullniemeierprop} and using \req{doublef}
along with \req{upsilonf} we see that $\Theta( \m K ) = \wt{\m K}_5$. Similarly, \req{finalpi}
implies $\Theta(\Pi)=\wt\Pi$ according to the proof of proposition \ref{piinn}, and since
$\Theta(\upsilon_0-\upsilon) = f_5$ by \req{upsilonf} we find $\Theta( \m P ) = \wt{\m P}_5$.
By proposition \ref{fullk3} we can now glue $H_\ast(X,\Z)$ from the sublattices
$\m K$ and $\m P$, while by proposition \ref{fullniemeierprop} we can glue $N$ from
the sublattices $\wt{\m K}_5$ and $\wt{\m P}_5$. By the discussion following proposition 
\ref{translationonN}, the linear bijection 
$\Theta\colon \m K\oplus \m P\longrightarrow \wt{\m K}_5\oplus \wt{\m P}_5(-1)$
is compatible with this gluing if and only if it respects \req{signchoices}. We have
ensured that this is the case by construction, namely by \req{lambdaq} and \req{upsilonf}. 
Hence under gluing,
$\Theta$ extends to a linear bijection $\Theta\colon H_\ast(X,\Z)\longrightarrow N(-1)$
as claimed.

Next we claim that $\Theta$ induces a primitive embedding 
$i_{\m T_{192}}\colon M_{\m T_{192}}\hookrightarrow N(-1)$ of the type constructed in
proposition \ref{translationonN}. By proposition \ref{latticechar}, the lattice
$M_{\m T_{192}}$ is the orthogonal complement of the lattice 
$\pi_\ast(H_2(T(\Lambda_{D_4}),\Z)^{\m T})\oplus\spann_\Z\{\upsilon_0+\upsilon\}$,
where $\m T$ is the binary tetrahedral group acting as non-translational holomorphic 
symplectic automorphism group of the $D_4$-torus. 
By proposition \ref{tetrahedralcpxstr} together with \req{homology}, \req{homology3}, 
the lattice $\pi_\ast(H_2(T(\Lambda_{D_4}),\Z)^{\m T})$  is generated by  
$I_1$, $ I_2$, $ I_3$ as defined there. Hence $M_{\m T_{192}}$ consists of
$$
\spann_\Z\left\{  I_1^{\perp}, I_2^{\perp}, I_3^{\perp}\right\} \oplus\m P,
$$
with  $I_k^{\perp},\, k=1,2,3$ as in \req{homologyperp}, along with the appropriate 
rational combinations of contributions from $K$ and $\Pi$ obtained by our gluing for 
generic Kummer surfaces. In particular, $M_{\m T_{192}}$ has rank $20$.
We need to show that $\Theta$ maps this lattice isometrically 
to a primitive sublattice of $N(-1)$. We already know that 
$\Theta(\m P)=\wt{\m P}_5(-1)$ isometrically, due to the above and the proof of proposition 
\ref{piinn}. 
Since we also know that $\Theta$ is compatible with the gluing of
$H_\ast(X,\Z)$ and $N$ from $\m P$ and $\wt{\m P}_5$ and their orthogonal complements, 
it suffices to show that  the 
lattice generated by the $I_k^{\perp},\, k\in\{1, 2, 3\}$, is isometric to the lattice generated 
by the three vectors $\Theta(I_k^{\perp}),\, k\in\{1, 2, 3\}$, up to an inversion of signature. 
In fact, note that by \req{doublef} we have
 $$
I_1^{\perp}\longmapsto  f_{19}-f_{15} ,\quad 
I_2^{\perp} \longmapsto  f_{9}-f_{15},\quad
I_3^{\perp} \longmapsto f_{24}-f_{15}, 
 $$
 yielding the quadratic form of the corresponding lattices with respect to these generators as
 \be\label{tetracompare}
 \left( \begin{array}{rrr}
 -4&-2&-2\\-2&-4&-2\\-2&-2&-4
 \end{array}\right)
 \quad \longmapsto\quad
  \left( \begin{array}{rrr}
 4&2&2\\2&4&2\\2&2&4
 \end{array}\right)
 \ee
on both sides, as required. 

Finally, equivariance with respect to $\m T_{192}$ follows by construction for the isometric embedding 
$i_{\m T_{192}}\colon M_{\m T_{192}}\hookrightarrow N(-1)$ obtained by restricting
$\Theta$ to $M_{\m T_{192}}$,
see proposition \ref{t192inM24}.  
\end{proof}

Note that the comparison of the quadratic forms in \req{tetracompare} uniquely determines the 
relative signs in front of $f_{15}$ in the formulas for the $I_k^\perp$, $k\in\{1,\, 2,\, 3\}$,
in \req{doublef}. Note furthermore that the remaining choices of signs in \req{doublef} amount to 
choosing one sign for each $f_n$ separately in the formulas for the $I_k^\perp$ as well as in the identifications
\be\label{doublefreduced}
\begin{array}{rcrcl}
J_1^\perp&:=&\pi_\ast\lambda_{14}-\pi_\ast\lambda_{23}
& \longmapsto& f_{24}-f_{23},\\[5pt]
J_2^\perp&:=&\pi_\ast\lambda_{13}+\pi_\ast\lambda_{24}
& \longmapsto& f_{19}-f_{6},\\[5pt]
J_3^\perp&:=&\pi_\ast\lambda_{34}-\pi_\ast\lambda_{12}
& \longmapsto& f_9-f_3.
\end{array}
\ee
We wish to emphasize that our gluing strategy leading to theorem \ref{Thetaok}
differs from that used by Kondo in
\cite{ko98}, which was explained in section \ref{autos}. Indeed, due to our
theorem \ref{biglattice}, instead of basing our construction on the lattices 
$L_G$ and $N_G$ which Kondo uses, we can work with the lattice $M_G$ which contains
$L_G$ and exceeds it by rank $2$. In particular, $M_G$ contains the Kummer lattice $\Pi$,
which simplifies our analysis greatly. Since the vector 
$e=\hf \sum_{\vec a\in\F_2^4} E_{\vec a}\in\Pi$  is invariant under every
holomorphic symplectic automorphism of our Kummer surface (with its
induced dual K\"ahler class), the Kummer lattice can never be contained 
in Kondo's lattice $L_G$. In this sense, our techniques improve Kondo's construction
for all Kummer surfaces with induced dual K\"ahler class.
%
\subsection{The {square Kummer} surface}\label{square}
%
Let us now carry out a similar analysis as for the tetrahedral Kummer surface $X_{D_4}$
for another example of an algebraic Kummer surface:
\begin{definition}\label{squaretorus}
We call the standard torus $T_0=T(\Lambda_0)=\C^2/\Lambda_0$
the \textsc{square torus}, where the generators of the lattice $\Lambda_0$ are given by
\be\label{standardgenerators}
\vec{\lambda}_1 = (1,0), \quad \vec{\lambda}_2=(i,0),\quad \vec{\lambda}_3=(0,1), \quad \vec{\lambda}_4=(0,i).
\ee
The Kummer surface with underlying torus $T_0$ is denoted by $X_0:=\widetilde{T_0/\Z_2}$, and
we call it the \textsc{square Kummer surface}.
\end{definition}
In the example discussed around \req{standardKummer} and \req{standardPolar}
we already found
\begin{prop}\label{standardcpxstr}
Consider the square Kummer surface $X_0$ 
of definition \mbox{\rm\ref{squaretorus}}, where we are
working with our usual marking described at the end of section \mbox{\rm\ref{exampleK3}}.

The real homology classes $\Omega_1,\, \Omega_2,\, \omega$ which
determine the complex structure of $X_0$, by means of the Torelli Theorem \mbox{\rm\ref{torellithm}},
and the induced dual K\"ahler class, which in fact yields a polarization,  are given by
$$
\Omega_1=\pi_\ast\lambda_{13}-\pi_\ast\lambda_{24},\quad \Omega_2=\pi_\ast\lambda_{14}+\pi_\ast\lambda_{23},\quad
\omega=\pi_\ast\lambda_{12}+\pi_\ast\lambda_{34}\in H(X_0,\Z).
$$
Furthermore, the  $3$-dimensional subspace $\Sigma$ of  $H_2(X_0,\R)$ 
containing these three 2-cycles yields a lattice $\Sigma\cap H_2(X_0,\Z)$  of \mbox{\rm(}the maximal possible\mbox{\rm)}
rank $3$, with quadratic form
$$
\left( \begin{array}{ccc}4&0&0\\0&4&0\\0&0&4\end{array}\right).
$$ 
\end{prop}
One checks that the group of non-translational holomorphic symplectic 
automorphisms of the standard torus $T_0$ is 
the binary dihedral group $\m O$ of order $8$ mentioned in the proof of 
proposition \ref{sdproduct},
which is generated by
the symmetries
$$
\begin{array}{rcl}
\alpha_1\colon (z_1, z_2)&\mapsto&(iz_1, -iz_2),\\[5pt]
\alpha_2\colon (z_1, z_2)&\mapsto&(-z_2, z_1)
\end{array}
$$
on the universal cover $\C^2$ of $T_0$, and which induces an action of $(\Z_2)^2=\m O/\Z_2$ on
the Kummer surface $X_0$ obtained from $T_0$
(see \cite[Table 9]{fu88} or e.g.\ \cite[Thm.7.3.12]{diss}). By proposition \ref{sdproduct}
it then follows that the holomorphic symplectic automorphism group of 
$X_0$ is the group $G_t\rtimes (\Z_2)^2$ of order $64$, where $G_t\cong(\Z_2)^4$ 
is the translational automorphism group of definition \ref{defgeneric}.

Note that in terms of standard Euclidean coordinates
of $\C^2$, the symmetries $\alpha_1$ and $\alpha_2$ agree with $\gamma_1$, 
respectively $\gamma_2$
in (\ref{t4gen}). However, the induced actions on the lattice $H_2(T_0,\Z)$ and thereby on 
$H_\ast(X_0,\Z)$ are different, since the defining lattices of the underlying tori are different.
Indeed, one immediately checks that the induced actions of $\alpha_1$ and $\alpha_2$ on the
generators $\vec\lambda_1, \ldots, \vec\lambda_4$ of $\Lambda_0$ as in \req{standardgenerators} 
are the following:
\be\label{alphaaction}
\begin{array}{lllll}
\alpha_1\colon 
&\vec\lambda_1\mapsto \vec\lambda_2,&  \vec\lambda_2\mapsto -\vec\lambda_1,
&\vec\lambda_3\mapsto -\vec\lambda_4, &\vec\lambda_4\mapsto \vec\lambda_3,\\[5pt]
\alpha_2\colon 
& \vec\lambda_1\mapsto \vec\lambda_3,& \vec\lambda_2\mapsto \vec\lambda_4,
&\vec\lambda_3\mapsto -\vec\lambda_1, &\vec\lambda_4\mapsto -\vec\lambda_2.
\end{array}
\ee
We therefore obtain
\begin{prop}\label{squareautos}
Consider the square Kummer surface $X_0$ 
of definition \mbox{\rm\ref{squaretorus}}, equipped with the induced polarization. 

The symmetry group of $X_0$ is the group
$\m T_{64}:=G_t\rtimes (\Z_2)^2$ of order $64$, where 
$G_t\cong(\Z_2)^4$ is the translational automorphism group of
definition \mbox{\rm\ref{defgeneric}}. The induced action of the non-translational quotient group $(\Z_2)^2$ on
$H_2(T_0,\Z)$ with respect to the basis $\lambda_{ij}$ of \mbox{\rm\req{lambdaq}} is generated by
\be \label{alpha1H2T}
\left( \begin{array}{c} \alpha_1(\lambda_{12})\\ \alpha_1(\lambda_{13})\\ \alpha_1(\lambda_{14})\\ \alpha_1(\lambda_{23})\\  \alpha_1(\lambda_{24})\\ \alpha_1(\lambda_{34})  \end{array}\right)  =
\left(  \begin{array}{cccccc} 
1&0&0&0&0&0\\
0&0&0&0&-1&0\\
0&0&0&1&0&0\\
0&0&1&0&0&0\\
0&-1&0&0&0&0\\
0&0&0&0&0&1 \end{array}\right) 
\left( \begin{array}{c} 
\lambda_{12}\\  \lambda_{13}\\ \lambda_{14}\\ \lambda_{23}\\ \lambda_{24}\\\lambda_{34} \end{array}\right),
\ee
\be\label{alpha2H2T}
\left( \begin{array}{c} \alpha_2(\lambda_{12})\\ \alpha_2(\lambda_{13})\\ \alpha_2(\lambda_{14})\\ \alpha_2(\lambda_{23})\\  \alpha_2(\lambda_{24})\\ \alpha_2(\lambda_{34})  \end{array}\right) = 
\left(  \begin{array}{cccccc} 
0&0&0&0&0&1\\
0&1&0&0&0&0\\
0&0&0&1&0&0\\
0&0&1&0&0&0\\
0&0&0&0&1&0\\
1&0&0&0&0&0 
\end{array}\right) \left( \begin{array}{c} 
\lambda_{12}\\  \lambda_{13}\\ \lambda_{14}\\ \lambda_{23}\\ \lambda_{24}\\\lambda_{34} \end{array}\right).
\ee
\end{prop}
We now investigate the compatibility of the linear bijection
$\Theta$ of theorem \ref{Thetaok}
with the symmetry group 
$\m T_{64}$ of $X_0$:
\begin{prop}\label{alphaactionprop}
Consider the square Kummer surface $X_0$ with its holomorphic symplectic automorphism group 
$\m T_{64}$ according to proposition \mbox{\rm\ref{squareautos}}.
Let $M_{\m T_{64}}$ denote the sublattice of $H_\ast(X,\Z)$ introduced in theorem \mbox{\rm\ref{biglattice}},
and let $i_{\m T_{64}}\colon M_{\m T_{64}}\longrightarrow N(-1)$ denote the restriction of the map 
$\Theta$ of theorem \mbox{\rm\ref{Thetaok}} to $M_{\m T_{64}}$.

The map $i_{\m T_{64}}$ is equivariant with respect to the symmetry 
group $\m T_{64}$, where the action of $\iota_1,\, \ldots,\, \iota_4$
on $N$ is given by \mbox{\rm\req{genericc24}}, and the generators $\alpha_1,\, \alpha_2$ act by 
\be\label{alphapermutations}
\begin{array}{rcl}
\alpha_1&=&(4,8)(6,19)(10,20)(11,13)(12,22)(14,17)(16,18)(23,24),\\[5pt]
\alpha_2&=&(2,21)(3,9)(4,8)(10,12)(11,14)(13,17)(20,22)(23,24).
\end{array}
\ee
\end{prop}
\begin{proof}
The equivariance of $i_{\m T_{64}}$ with respect to  $\iota_1,\, \ldots,\, \iota_4$ follows immediately from
proposition \ref{translationonN}. 

For the generators $\alpha_1,\,\alpha_2$ of $\m T_{64}$, we find that 
the transformations (\ref{alphaaction}) induce permutations $[\alpha_k]$ 
on the singular points of $T_0/\Z_2$ and thus on the elements of our hypercube 
$\F_2^4$. For the induced permutations $\wh\alpha_k$ 
on $\m I\setminus\m O_9$, imposing equivariance of $i_{\m T_{64}}$ amounts to
the condition $\wh\alpha_k(n)=[\alpha_k](I(n))$ for all $n\in\m I\setminus\m O_9$,
where $I$ is  the map \req{map}. We obtain 
$$
\begin{array}{rcl}
\wh{\alpha}_1&=&(4,8)(10,20)(11,13)(12,22)(14,17)(16,18),\\[5pt]
\wh{\alpha}_2&=&(2,21)(4,8)(10,12)(11,14)(13,17)(20,22).
\end{array}
$$
From \req{alpha1H2T} we obtain the induced action of $\alpha_1$
on the $\pi_\ast\lambda_{ij}$, such that equivariance of \req{signchoices}
implies that the induced action of $\alpha_1$ on the $Q_{ij}$ must 
fix $Q_{12}$ and $Q_{34}$ and yield
$Q_{13}\leftrightarrow Q_{24}, \; Q_{14}\leftrightarrow Q_{23}$. Given
our choices for the $Q_{ij}$ of
proposition \ref{Ibarsolution} which lead to \req{lambdaq}, this 
action amounts to the permutation $\tau_1:=(6,19)(23,24)$.

Similarly, \req{alpha2H2T} implies that $\alpha_2$ must fix $Q_{13}$ and $Q_{24}$ 
and must induce
$Q_{12}\leftrightarrow Q_{34}, \; Q_{14}\leftrightarrow Q_{23}$, which amounts to
the permutation
$\tau_2:=(3,9)(23,24)$.
One now checks that the permutations $\alpha_k:=\wh\alpha_k\circ\tau_k$, 
which agree with the permutations $\alpha_1,\,\alpha_2$ in our claim \req{alphapermutations}, 
are elements of the Mathieu group $M_{24}$.
\end{proof}

We actually arrive at the surprising results of
\begin{theorem}\label{moreTheta}
Consider the square Kummer surface $X_0$
of definition \mbox{\rm\ref{squaretorus}} with its holomorphic symplectic automorphism group 
$\m T_{64}$.
Let $M_{\m T_{64}}$ denote the sublattice of $H_\ast(X,\Z)$ introduced in theorem \mbox{\rm\ref{biglattice}},
and let $i_{\m T_{64}}$ denote the restriction of the map 
$\Theta$ of theorem \mbox{\rm\ref{Thetaok}} to $M_{\m T_{64}}$.

The map $i_{\m T_{64}}$ yields an isometric embedding 
$i_{\m T_{64}}\colon M_{\m T_{64}}\hookrightarrow N(-1)$ of the type constructed in 
proposition \mbox{\rm\ref{translationonN}}. Moreover, $i_{\m T_{64}}$
is equivariant with respect to the group $\m T_{64}$, where on
$H_\ast(X,\Z)$ the action of $\m T_{64}$
is induced by its action on $X_0$ as holomorphic symplectic automorphism
group, and on the Niemeier lattice $N$ of type $A_1^{24}$ its action
is generated by $\iota_1,\,\ldots,\iota_4,\,\alpha_1,\,\alpha_2$
as in \mbox{\rm\req{genericc24}} and \mbox{\rm\req{alphapermutations}}. The lattice $M_{\m T_{64}}$ has rank $20$.
\end{theorem}
\begin{proof}
By proposition \ref{latticechar}, the lattice
$M_{\m T_{64}}$ is the orthogonal complement of the lattice 
$\pi_\ast(H_2(T_0,\Z)^{\m O})\oplus\spann_\Z\{\upsilon_0+\upsilon\}$,
where $\m O$ is the binary dihedral group of order $8$ acting as non-translational 
holomorphic symplectic automorphism group of the  square torus. 
By proposition \ref{standardcpxstr} 
the lattice $\pi_\ast(H_2(T_0,\Z)^{\m O})$  is generated by 
$J_1:=\pi_\ast\lambda_{13}-\pi_\ast\lambda_{24},\, 
J_2:=\pi_\ast\lambda_{14}+\pi_\ast\lambda_{23},\,
J_3:=\pi_\ast\lambda_{12}+\pi_\ast\lambda_{34}$.
Hence the lattice $M_{\m T_{64}}$ consists of
$$
\spann_\Z\left\{  J_1^{\perp}, J_2^{\perp}, J_3^{\perp}\right\} \oplus\m P,
$$
with  $J_k^{\perp},\, k=1,2,3$ as in \req{doublefreduced}, along with the appropriate 
rational combinations of contributions from $K$ and $\Pi$ obtained by our gluing for 
generic Kummer surfaces. In particular, $M_{\m T_{64}}$ has rank $20$.

We need to show that $\Theta$ maps this lattice isometrically 
to a primitive sublattice of $N(-1)$. We already know that 
$\Theta(\m P)=\wt{\m P}_5(-1)$ isometrically, due to the  proof of theorem 
\ref{Thetaok}. 
Since we also know that $\Theta$ is compatible with the gluing of
$H_\ast(X,\Z)$ and $N$ from $\m P$ and $\wt{\m P}_5$ and their orthogonal complements
by that same proof,
it suffices to show that  the 
lattice generated by the $J_k^{\perp},\, k\in\{1, 2, 3\}$, is isometric to the lattice generated 
by the three vectors $\Theta(J_k^{\perp}),\, k\in\{1, 2, 3\}$, up to an inversion of signature. 
In fact, note that by \req{doublef} we have
$$
J_1^\perp\longmapsto f_{24}-f_{23},
\quad J_2^\perp \longmapsto  f_{19}-f_6, \quad  J_3^\perp \longmapsto f_9-f_{3},
$$
 yielding the quadratic form of the corresponding lattices with respect to these generators as
 \be\label{squarecompare}
 \left( \begin{array}{rrr}
 -4&0&0\\0&-4&0\\0&0&-4
 \end{array}\right)
 \quad \longmapsto\quad
  \left( \begin{array}{rrr}
 4&0&0\\0&4&0\\0&0&4
 \end{array}\right)
 \ee
on both sides, as required. 

Finally, equivariance with respect to $\m T_{64}$ now follows from proposition \ref{alphaactionprop}.  
\end{proof}

Summarising, we have confirmed the surprising fact that the same bijection $\Theta$
of theorem \ref{Thetaok} is compatible with the symmetry groups of two \textsl{distinct} Kummer surfaces, namely the 
tetrahedral Kummer surface and the square Kummer surface.
%
\subsection{Generating the overarching finite symmetry group of Kummer surfaces}\label{overmen}
%
Let us now discuss the consequences of theorems \ref{Thetaok} and \ref{moreTheta},
which state that our map $\Theta$ induces two 
isometric embeddings $M_{\m T_{192}}\hookrightarrow N(-1)$
and $M_{\m T_{64}}\hookrightarrow N(-1)$. We can view $\Theta$ as an 
\textsc{overarching bijection} for these two embeddings. Theorems \ref{Thetaok} 
and \ref{moreTheta} also show that $\Theta$ is compatible with enforcing
equivariance on each  $i_{\m T_k}$, {$k\in\{64, 192\}$}, with respect to the symmetry group $\m T_k$,
inducing actions of these groups on the Niemeier lattice $N$ of type $A_1^{24}$ as 
subgroups of the Mathieu group
$M_{24}$. Hence the lattice $N$ serves as a device which carries both these actions
simultaneously, allowing us to make sense of combining these groups to a bigger 
group: 
\begin{theorem}\label{totalgroup}
Consider the subgroup $M_\Pi$ of the Mathieu group $M_{24}$ which is
obtained by the combined actions of $\m T_{192}$ and $\m T_{64}$
of theorems \mbox{\rm\ref{Thetaok}}
and \mbox{\rm\ref{moreTheta}} on the Niemeier lattice $N$ of type $A_1^{24}$. 
This group is isomorphic to the overarching finite symmetry
group $(\Z_2)^4\rtimes A_7$ of Kummer surfaces of definition \mbox{\rm\ref{overarch}}. 

If $X$ is an arbitrary Kummer surface with
induced dual K\"ahler class and holomorphic symplectic automorphism group $G$,
let $i_G\colon M_G\hookrightarrow N(-1)$ denote an isometric
embedding of the lattice $M_G$ of theorem \mbox{\rm\ref{biglattice}} according
to proposition \mbox{\rm\ref{translationonN}} such that $i_G(\m P)=\wt{\m P}_{n_0}(-1)$
with $n_0\in\m O_9$. 
By enforcing $G$-equivariance 
on $i_G$, the group $G$ acts faithfully on the Niemeier lattice
$N$. If $n_0=5$, then this realises $G$
 as a subgroup of the group $M_\Pi$. If $n_0\neq5$, then $G$ is a subgroup
 of a conjugate of $M_\Pi$ within the subgroup
 $(\Z_2)^4\rtimes A_8$ of $M_{24}$ which stabilizes the octad $\OOO_9$. 
 In this sense, $(\Z_2)^4\rtimes A_7$
contains the holomorphic symplectic automorphism group of
every Kummer surface with induced dual K\"ahler class.
\end{theorem}
\begin{proof} 
By theorems \ref{Thetaok} and \ref{moreTheta}, the group $M_\Pi$ 
is generated by $\iota_1, \ldots, \iota_4,\, \gamma_1,\, \gamma_2,\, \gamma_3$
as in \req{t192} and $\alpha_1,\, \alpha_2$ as in \req{alphapermutations}. 
Clearly, the permutations $\iota_1, \ldots, \iota_4$ generate a faithful action of
$(\Z_2)^4$ on $\m I=\{1,\ldots,24\}$ which leaves the labels in the special octad 
$\m O_9$ invariant pointwise. Furthermore, one readily sees that 
$\gamma_1,\, \gamma_2,\, \gamma_3,\, \alpha_1,\, \alpha_2$
induce a faithful action of $A_7$ on $\m O_9\setminus\{5\}$. 
Hence $M_\Pi\supset (\Z_2)^4\rtimes A_7$.

On the other hand, by construction both our groups
$\m T_{192}$ and $\m T_{64}$ act on $N$ as subgroups
of the stabilizer subgroup of $M_{24}$ of our special octad $\m O_9$.
According to \cite{tod66} this group is a maximal subgroup of $M_{24}$, and 
it is isomorphic to $(\Z_2)^4\rtimes A_8$, where $A_8$ is realised as the
group of permutations undergone by the $8$ points of $\m O_9$. Since
in addition, both our groups $\m T_{192}$ and $\m T_{64}$ stabilize the
label $5=n_0\in\m O_9$, we find that both of them are subgroups of
the corresponding maximal subgroup $(\Z_2)^4\rtimes A_7$ of $M_{23}$. 
In particular, $M_\Pi\subset (\Z_2)^4\rtimes A_7$, thus by
the above $M_\Pi\cong(\Z_2)^4\rtimes A_7$ is isomorphic
to the overarching finite symmetry group of
Kummer surfaces, as claimed.

Consider a Kummer surface $X$ with
induced dual K\"ahler class and holomorphic symplectic automorphism group $G$,
and the corresponding lattice $M_G$ that was constructed in
theorem \ref{biglattice}.
Let $i_G\colon M_G\hookrightarrow N(-1)$ denote
the isometric embedding with  $i_G(\m P)=\wt{\m P}_{n_0}(-1)$ according
to proposition \ref{translationonN}, where $n_0\in\m O_9$. Assume first that $n_0=5$.
Then the action of $G$ on $N$ obtained by 
enforcing $G$-equivariance 
on $i_G$ maps $\wt{\m K}_5=\wt{\m P}_5^\perp\cap N$ onto itself.
By the proof of proposition \ref{fullniemeierprop} we know that
$\wt{\m K}_5=\spann_\Z\left\{ f_3,\, f_6,\, f_9,\-f_{15},\, f_{19},\, f_{23},\, f_{24}\right\}$.
Hence $G$ stabilizes the special octad $\m O_9=\{3,\, 5,\, 6,\, 9,\ {15},\, {19},\, {23},\, {24}\}$
and fixes $n_0=5$. Thus, by the above, $G$ is a subgroup of the maximal subgroup
$M_\Pi\cong (\Z_2)^4\rtimes A_7$ of $M_{23}$.
If $n_0\neq 5$, then analogously one finds $G\subset M_\Pi^\prime\cong(\Z_2)^4\rtimes A_7$
where $M_\Pi^\prime$ is the subgroup of $M_{24}$ containing those automorphisms that stabilize
the special octad $\m O_9$ and that fix $n_0$. Using the structure of the stabilizer
group of $\m O_9$ described above, one finds that $M_\Pi^\prime$ is obtained from $M_\Pi$
by conjugation within the subgroup
 $(\Z_2)^4\rtimes A_8$ of $M_{24}$ which stabilizes the octad $\OOO_9$. 
\end{proof}

The group $M_\Pi\cong(\Z_2)^4\rtimes A_7$ obtained in theorem \ref{totalgroup}
has order $40320$. By the results of the theorem it is
also the largest group we can possibly obtain by combining several
holomorphic symplectic automorphism groups of Kummer surfaces
through equivariant embeddings overarched by some linear bijection 
$H_\ast(X,\Z)\longrightarrow N(-1)$ like our map $\Theta$ of theorem 
\ref{Thetaok}.

The construction of an \textsl{overarching} map $\Theta\colon H_\ast(X,\Z)\longrightarrow N(-1)$
for the two primitive embeddings $i_{G_k}\colon M_{\TTT_k}\hookrightarrow N(-1)$ makes 
sense, since we are choosing a fixed marking to identify $H_\ast(X,\Z)$, $H_2(X,\Z)$ with standard unimodular
lattices of signatures $(4,20)$, $(3,19)$, which is induced by the Kummer construction, as was explained
at the end of section \ref{exampleK3}. According to our remarks at the end of section \ref{complex},
this marking allows us to construct Kummer paths in
the smooth connected cover $\wt\MMM_{hk}$ of the moduli space of hyperk\"ahler structures.
Recall that each hyperk\"ahler structure is
represented by the $3$-dimensional subspace $\Sigma\subset H_2(X,\R)=H_2(X,\Z)\otimes\R$
generated by the $2$-dimensional $\Omega\subset H_2(X,\R)$, which specifies the complex structure according to the
Torelli Theorem \ref{torellithm}, and the dual K\"ahler class $\omega$. 
Along Kummer paths, each complex structure and (degenerate) dual
K\"ahler class is induced from the universal cover $\C^2$ of an underlying torus $T(\Lambda)=\C^2/\Lambda$,
{$\Lambda=\mbox{span}_\Z\{\vec\lambda_1,\ldots,\vec\lambda_4\}$,}
for a Kummer surface $\wt{T(\Lambda)/\Z_2}$. Smooth
variation of the {generators $\vec\lambda_1,\ldots,\vec\lambda_4$ of the}
underlying lattice $\Lambda\subset\C^2$
yields a smooth variation between the hyperk\"ahler structures of
any two Kummer surfaces. Our marking allows us
to view the lattice $H_\ast(X,\Z)$ (with generators $\upsilon_0,\, \upsilon;\;
{1\over2}\pi_\ast\lambda_{ij}+{1\over2}\sum_{\vec a\in P_{ij}}\varepsilon_{\vec a}E_{\vec a+\vec b},\,
ij\in\{ 12, 34, 13, 42, 14, 23\},\,\vec b\in\F_2^4,\,\varepsilon_{\vec a}\in\{\pm1\}$) as fixed, while
the $3$-dimensional space $\Sigma$
generated by $e_1\vee e_3-e_2\vee e_4,\, e_1\vee e_4+e_2\vee e_3,\,
e_1\vee e_2+e_3\vee e_4$ (see \req{homologyinv1}, \req{homologyinv2}) varies, since the expressions
of each $e_i$ in terms of $\vec\lambda_1,\ldots,\vec\lambda_4$ vary with $\Lambda$. In view of this
remark it is natural to ask whether our overarching map $\Theta$ is compatible with the embeddings
$i_{G^s}\colon M_{G^s}\hookrightarrow N(-1)$ of proposition \ref{translationonN} along some Kummer
path connecting the tetrahedral and the square Kummer surface. Indeed, this is the case:
\begin{theorem}\label{Kummerpath}
There exists a smooth path in the smooth connected cover $\wt\MMM_{hk}$ of the moduli space
of hyperk\"ahler structures of \mbox{\rm K3} with the following properties: 

Let $\Sigma_s$ denote the positive
definite oriented $3$-dimensional subspace of $H_2(X,\R)$ generated by the $2$-dimensional $\Omega_s$, which 
specifies the complex structure according to the Torelli Theorem \mbox{\rm\ref{torellithm}}, and $\omega_s$, the
dual K\"ahler class, at time $s\in[0,1]$ along the path. These data are specified by their relative positions with
respect to the lattice $H_2(X,\Z)$ which is given in terms of the even unimodular lattice of
signature $(3,19)$ arising from the Kummer construction, as explained at the end of section
\mbox{\rm\ref{exampleK3}}. Then for each $s\in[0,1]$\mbox{\rm:}
\begin{itemize}
\item
$\Omega_s,\,\omega_s$ give the complex structure and \mbox{\rm(}degenerate\mbox{\rm)} dual K\"ahler class
of a Kummer surface $\wt{T(\Lambda_s)/\Z_2}$ and are induced from the standard structures
on $\C^2$, the universal cover of $T(\Lambda_s)=\C^2/\Lambda_s$.
\item
$\Omega_0,\,\omega_0$ give the complex structure and \mbox{\rm(}degenerate\mbox{\rm)} dual K\"ahler class
of the square Kummer surface $X_0$.
\item
$\Omega_1,\,\omega_1$ give the complex structure and \mbox{\rm(}degenerate\mbox{\rm)} dual K\"ahler class
of the tetrahedral Kummer surface $X_{D_4}$.
\item
Let $G^s$ denote the symmetry group of the Kummer surface with data $\Omega_s,\, \omega_s$
and $M_{G^s}\subset H_\ast(X,\Z)$ the lattice of theorem \mbox{\rm\ref{biglattice}}. Then for the map 
$\Theta\colon H_\ast(X,\Z)\longrightarrow N(-1)$ constructed in theorem \mbox{\rm\ref{Thetaok}},
$i_{G^s}:=\Theta_{\mid M_{G^s}}\colon M_{G^s}\hookrightarrow N(-1)$ is an 
isometric, primitive embedding. Enforcing $G^s$-equivariance, $G^s$ acts as subgroup
of the overarching symmetry group $M_\Pi\cong(\Z_2)^4\rtimes A_7\subset M_{24}$ on $N$ which
was found in theorem \mbox{\rm\ref{totalgroup}}.
\end{itemize}
In particular, $s\mapsto\Sigma_s$, $s\in[0,1]$ describes a Kummer path along which
$i_G(\Pi)=\wt\Pi$ is constant.
\end{theorem}
\begin{proof}
We show that there exists a Kummer path $s\mapsto \Sigma_s$, $s\in[0,1]$,
in $\wt\MMM_{hk}$ which connects the square Kummer surface $X_0$ with
the tetrahedral Kummer surface $X_{D_4}$ such that for every
$s\in(0,1)$, the only holomorphic symplectic automorphisms are the translational
ones, i.e.\ $G^s\cong(\Z_2)^4$, $M_{G^s}=\Pi\oplus\spann_\Z\{\upsilon_0-\upsilon\}$.
Then the claims immediately follow by construction, since compatibility with 
the translational group $G^s=G_t\cong(\Z_2)^4$ is 
incorporated at the beginning of our construction of $\Theta$ in section \ref{genericKummer}.

Since the desired path is a Kummer path, we must ensure $\Sigma_s\subset \pi_\ast H_2(T,\R)$
for all $s\in [0,1]$ with the notations of section \ref{exampleK3}. Proposition \ref{sdproduct} implies
that  {$G^s=G_t\cong(\Z_2)^4$ if $\left(H_2(X,\Z)^{G^s}\right)^\perp\cap H_2(X,\Z)=\Pi$}
or equivalently $\left(H_2(X,\Z)^{G^s}\right)^\perp\cap \pi_\ast H_2(T,\Z)=\{0\}$. In general, 
$\Sigma_s\subset H_2(X,\Z)^{G^s}\otimes\R$ and thus 
$(\Sigma_s)^\perp \supset \left(H_2(X,\Z)^{G^s}\right)^\perp$;
hence it suffices to ensure $(\Sigma_s)^\perp\cap \pi_\ast H_2(T,\Z)=\{0\}$ for all $s\in(0,1)$. 
Using the $I_k^\perp$ and $J_k^\perp$ of \req{homologyperp} and \req{doublefreduced}
we set $\Sigma_{T,s}^\perp:=\spann_\R\left\{ I_{1,s}^\perp,\, I_{2,s}^\perp,\, I_{3,s}^\perp\right\}
\subset\pi_\ast H_2(T,\R)$,
$$
I_{k,s}^\perp := \left\{
\begin{array}{ll}
(1-s) J_k^\perp + s I_k^\perp + 2s\delta_k\cdot \pi_\ast \lambda_{12} &\mbox{if } s\in[0,{1\over2}],\\[3pt]
(1-s) J_k^\perp + s I_k^\perp + 2(1-s)\delta_k\cdot \pi_\ast \lambda_{12}&\mbox{if } s\in[{1\over2},1],
\end{array}
\right.
$$
with sufficiently small $\delta_1,\,\delta_2,\,\delta_3\in\R$, such that $\Sigma_{T,s}^\perp$ is negative definite,
which are linearly independent over $\Q$. One checks that 
$I_{k,0}^\perp = J_k^\perp$ and $I_{k,1}^\perp = I_k^\perp$ for $k\in\{1,2,3\}$, and
$\Sigma_{T,s}^\perp\cap\pi_\ast H_2(T,\Z) =\{0\}$ for all $s\in(0,1)$. Hence
$\Sigma_s:=\left(\Sigma_{T,s}^\perp\right)^\perp\cap\pi_\ast H_2(T,\R)$ defines a
Kummer path from the square Kummer surface to the tetrahedral Kummer surface
with the desired properties.
\end{proof}

Although we have not discussed the detailed proof of this statement here,
it is worthwhile mentioning that our map $\Theta$ is almost uniquely determined
by the requirement that it induces isometric embeddings of the type
constructed in proposition \ref{translationonN}
for both lattices
$M_{\m T_{192}}$ and $M_{\m T_{64}}$. 
Our explicit formulas, of course, first of all depend on our choice of $\m O_9$ as
our special octad, and more precisely on the particular embedding of the 
Kummer lattice $\Pi$ in the Niemeier lattice $N(-1)$ of type $A_1^{24}$
that we constructed in the proof of proposition \ref{piinn}. However, by
proposition \ref{uniquepi} any other embedding of $\Pi$ is related to the
one used by us by an automorphism of the Niemeier lattice $N$. On
the level of subgroups of the Mathieu group $M_{24}$, this amounts
to conjugating by some element of $M_{24}$. 

Once our embedding of $\Pi$ in $N(-1)$ has been chosen, our requirements on $\Theta$ turn out
to enforce the choices \req{identifymuij} for the quadruplets $Q_{ij}$
in our special octad $\m O_9$. Concerning the choices of signs in \req{signchoices},
as mentioned in the discussion of equation \req{lambdaq}, these are unique up
to the freedom for each $n\in\m I$ to replace $f_n$ by $-f_n$ everywhere,
and up to some choices of signs in the image of $\upsilon_0+\upsilon$.

Now recall once again that the involutions on the Niemeier lattice $N$ of type $A_1^{24}$
induced by $f_n\mapsto -f_n$ should not play a role in describing symmetries of Kummer
surfaces, or in fact their associated superconformal field theories:
On the one hand, by proposition \ref{M24onN} the Mathieu group 
$M_{24}$, which we are interested in, is obtained as factor group of $\Aut(N)$
by the normal subgroup $(\Z_2)^{24}$ generated by these involutions. 
On the other hand, consider a symplectic automorphism of a K3 surface
$X$ whose induced action on $H_\ast(X,\Z)$ is given by the lattice automorphism
$\alpha$. We claim that $\alpha(E)\neq-E$
for every $E\in H_\ast(X,\Z)$ with $\langle E,E\rangle=-2$. Indeed,
$\alpha(v)=-v$ for some $v\in H_\ast(X,\Z)$ immediately implies $v\perp L^\alpha_\bullet$, where
$L^\alpha_\bullet:=H_\ast(X,\Z)^\alpha$ is the sublattice of $H_\ast(X,\Z)$ which is invariant
under $\alpha$. Since $L^\alpha_\bullet=H_0(X,\Z)\oplus L^\alpha\oplus H_4(X,\Z)$ with
$L^\alpha:=H_2(X,\Z)^\alpha$,
this implies that $v\in L_\alpha:=(L^\alpha_\bullet)^\perp\cap H_\ast(X,\Z)=(L^\alpha)^\perp\cap H_2(X,\Z)$.
Hence by the very Torelli Theorem \ref{niktorelli}, $\langle v,v\rangle\neq-2$.
By arguments analogous to those used for symplectic automorphisms of K3 surfaces, 
one may deduce the following from the definition of 
\textsc{symmetries of non-linear $\sigma$-models on K3 preserving the $\m N=(4,4)$
superconformal algebra}
given in \cite{ghv11}: Firstly, any such  symmetry induces a lattice automorphism 
$\alpha$ on $H_\ast(X,\Z)$, which
in turn uniquely determines the symmetry. Furthermore, 
$\alpha(E)\neq-E$
for every $E\in H_\ast(X,\Z)$ with $\langle E,E\rangle=-2$. In particular, an involution 
of the Niemeier lattice $N$ induced by $f_n\mapsto -f_n$ can never correspond
to a symplectic automorphism of a K3 surface or to a symmetry of a non-linear $\sigma$-model
on K3 which preserves the $\m N=(4,4)$ superconformal algebra.
%
\section{Conclusions and outlook}
This work gives a novel perspective on holomorphic symplectic automorphisms of  Kummer surfaces
whose dual K\"ahler class is induced by the underlying complex torus. 
While finite symplectic automorphism groups 
of K3 surfaces, in general, have been classified and described by Mukai and Kondo in
their seminal works \cite{mu88, ko98}, we improve their description 
in the case of Kummer surfaces through an identification 
of lattice automorphisms. 
This  provides a concrete representation in terms of 
permutations of $24$ elements which are symmetries of the binary extended Golay code,
and thus in terms of subgroups of the Mathieu group $M_{24}$. 

For a K3 surface $X$ 
with holomorphic symplectic automorphism group $G$ and $L_G\subset H_\ast(X,\Z)$ 
the orthogonal complement of the $G$-invariant part of the K3-homology,  one finds by the results of
Kondo \cite{ko98} a Niemeier lattice $\wt N$, which carries a faithful $G$-action, such that $L_G$ can be
primitively and $G$-equivariantly embedded in $\wt N(-1)$. We improve this result for
all Kummer surfaces whose dual K\"ahler class is induced from the underlying complex torus:
We prove that in this case the Niemeier lattice $\wt N$, which is not further specified in Kondo's construction,
can be replaced by the Niemeier lattice $N$ of type $A_1^{24}$. Moreover, the lattice $L_G$
can be replaced by a lattice $M_G$ with $\rk(M_G)=\rk(L_G)+2$ which contains $L_G$,
that is, $M_G$ can be primitively and $G$-equivariantly 
embedded in $N(-1)$. We deduce that for every Kummer surface whose dual K\"ahler class
is induced from the underlying torus, one can construct a $\Z$-linear bijection 
${\theta\colon}H_\ast(X,\Z)\longrightarrow N(-1)$ whose restriction to $M_G$ yields an
isometric, $G$-equivariant primitive embedding. 

We explicitly construct a  map {$\Theta$} with the above-mentioned properties for the tetrahedral Kummer
surface $X_{D_4}$ with holomorphic symplectic automorphism group $\m T_{192}$
of order $192$. 
The nature of this map depends on the details of the complex structure 
and dual K\"ahler class, and its analog for other Kummer surfaces must  therefore be constructed 
case by case, 
but is not difficult to obtain in the framework set up in our work.

Our construction uses a description of the lattice $H_\ast(X,\Z)$ by means of lattice
generators naturally arising in the Kummer construction. This provides us with a common,
fixed marking for all Kummer surfaces. Varying the {generators $\vec\lambda_1,\ldots,\vec\lambda_4$
of the} defining lattice of the underlying torus
then amounts to a deformation along a path in the smooth connected cover of the
moduli space of hyperk\"ahler structures on K3. In particular, having fixed a common,
natural marking for all Kummer surfaces allows us to simultaneously study distinct
Kummer surfaces.
Surprisingly, our map {$\Theta$}, apart from restricting to an isometric,
$\m T_{192}$-equivariant embedding of $M_{\m T_{192}}$ in $N(-1)$, also restricts to 
an isometric, $G$-equivariant embedding of $M_G$ in $N(-1)$ for a different Kummer surface $X_0$,
namely the one obtained from the square torus $T_0$. The holomorphic symplectic
symmetry group $G=\m T_{64}$ of $X_0$ has order $64$. 
We find a smooth path in the smooth connected cover of the moduli space of hyperk\"ahler structures
on K3 between these two Kummer surfaces, such that {$\Theta$} is compatible with the symmetries
of all Kummer surfaces along the path. The latter have the translational group 
$G_t \cong (\mathbb{Z}_2)^4$ as symmetry group, except {for the Kummer surfaces} at the endpoints, which have
symmetry groups $\m T_{64}$ and $\m T_{192}$ as explained above.
{The translational group $G_t\cong(\Z_2)^4$, which is compatible with $\Theta$ by construction, 
is the symmetry group of generic Kummer surfaces.}

By construction, the lattice $M_G$ is
negative definite for every Kummer surface. Both for $X_{D_4}$ and for $X_0$ it turns out to have the maximal 
possible rank $20$.
The smallest primitive sublattice
of $H_\ast(X,\Z)$ containing both lattices $M_{\m T_{192}}$ and $M_{\m T_{64}}$ is
$H_2(X,\Z)\oplus\spann_\Z\{\upsilon_0-\upsilon\}$,
where $\upsilon_0\in H_0(X,\Z)$ and $\upsilon\in H_4(X,\Z)$ with $\langle\upsilon_0,\upsilon\rangle=1$. 

Our map {$\Theta$} thus yields an
action of the groups $\m T_{192}$ and $\m T_{64}$ on the Niemeier lattice $N$,
realising each of them  as a subgroup of the Mathieu group $M_{24}$. Again,
the nature of the respective permutations depends on the details of the complex structure 
and dual K\"ahler class, and their analog for other Kummer surfaces can be worked out
similarly.
There are further Kummer surfaces whose symmetries are consistently described by our
specific map {$\Theta$}.
For example, the $\Z_3$-symmetric torus $T_{(3)}=\C^2/\Lambda_{(3)}$ with $\Lambda_{(3)}$ generated by 
\be
\textstyle
\vec\lambda_1=(1, 0), \quad\vec\lambda_2=(i,0),\quad 
\vec\lambda_3=(\hf, {\sqrt3\over2}),\quad\vec\lambda_4=(-{i\over2},{i\sqrt3\over2})\quad\in\C^2 
\ee
yields a Kummer 
surface whose holomorphic symplectic automorphism group $G$ is also 
compatible with our map {$\Theta$}, if for the two-cycles yielding complex structure and polarization 
one uses
$$
\wt\Omega_1 = e_1\vee e_2+e_3\vee e_4,\quad
\wt\Omega_2 = -e_1\vee e_4-e_2\vee e_3, \quad
\wt\omega = e_1\vee e_3-e_2\vee e_4.
$$
In this case our map {$\Theta$} restricts to a primitive $G$-equivariant map
from $M_G$ to $N(-1)$, which however is not isometric. 

While {$\Theta$} is by no means compatible with all the symmetries of all Kummer surfaces,
the surprising observation is the fact that it yields a device which allows to
simultaneously realise several finite symplectic automorphism groups of different Kummer
surfaces in terms of subgroups of $M_{23}$ acting on $N$. 
The combined action of these groups yields the overarching finite symmetry group 
 $(\Z_2)^4\rtimes A_7$ of Kummer surfaces, which
contains every holomorphic symplectic automorphism group of a Kummer surface
with induced dual K\"ahler class as a subgroup, 
and it is the maximal group that our techniques can describe, so far. 

Although we do not present the proof in this work, we remark that 
for the above-mentioned $\Z_3$-symmetric torus $T_{(3)}$ we can show that 
no isometric embedding $M_G\hookrightarrow N(-1)$ exists which maps
$\upsilon_0-\upsilon$ to $f_5=\Theta(\upsilon_0-\upsilon)$. It follows that 
along a path from $X_0$, say, to $\wt{T_{(3)}/\Z_2}$, starting with the
embedding $M_{G_0}\hookrightarrow N(-1)$ constructed in this work, 
$\upsilon_0-\upsilon$ must change its image. This probably implies
that an extension of our techniques to more general paths in moduli space will yield subgroups of
$M_{24}$ rather than $M_{23}$, as will be explained in \cite{tawe12}.

A crucial observation which underlies our construction is the fact that the Niemeier lattice
$N$ of type $A_1^{24}$ contains a primitive sublattice $\wt\Pi$ which is isometric to the 
Kummer lattice $\Pi$, up to a total inversion of signature. This observation -- though not
very hard to prove -- is new, and it gives a geometric meaning to a fact known
to group theorists, namely that the complement of an octad in the extended binary
Golay code naturally carries the structure of a $4$-dimensional vector space over $\F_2$. 
In fact, the Kummer lattice $\Pi$
is a primitive sublattice of our lattice $M_G$ for every Kummer surface, while 
$\Pi$ can never be a sublattice of the lattice $L_G\subset M_G$ used
by Kondo in his construction. Our lattice $M_G$ can be generated by Kondo's lattice
$L_G$ together with the Kummer lattice $\Pi$ and the vector $\upsilon_0-\upsilon$
mentioned above.
This vector gives a geometric interpretation to the invariant root needed 
in Kondo's construction, whose role had been mysterious, so far.

The lattices $\Pi$ and $\wt\Pi$, however, govern the
symmetries in our holomorphic symplectic automorphism groups $G$, as we
shall explain now. By the known structure of $\Pi$ and $\wt\Pi$, 
the group of isometries for both lattices is isomorphic to $\Aff(\F_2^4)\ltimes(\Z_2)^{16}$,
where each factor of $(\Z_2)^{16}$ is induced by inverting the sign of one lattice vector
 on which the quadratic form takes value $\pm 2$. Hence restricting to the
analogs of \textsc{effective automorphisms} known from geometry, we can restrict
our attention to the affine linear maps in $\Aff(\F_2^4)=(\Z_2)^4\rtimes \GL(4,\F_2)$.
 
Consider the action of a 
lattice automorphism $\alpha$ on $\Pi\cong\wt\Pi(-1)$ with $\alpha\in \Aff(\F_2^4)$.
If  $\alpha$ is an element of the normal subgroup $(\Z_2)^4$, then it corresponds to a translation on the hypercube 
underlying the lattices $\Pi$ and $\wt\Pi$, and it induces a trivial action on the
discriminant groups $\Pi^\ast/\Pi$ and $\wt\Pi^\ast/\wt\Pi$. Thus $\alpha$ can be
trivially extended to a lattice automorphism of the 
K3-homology $H_\ast(X,\Z)$ and  of the Niemeier lattice $N$ of type $A_1^{24}$.
We claim that extensions to lattice automorphisms of $H_\ast(X,\Z)$ and $N$
exist for every $\alpha\in \Aff(\F_2^4)$ which is represented by an action as automorphism
of $\Pi\cong\wt\Pi(-1)$. For the K3-homology $H_\ast(X,\Z)$ this follows 
immediately from \cite[Thm.~1.7]{ni80b} (using \cite[Remark 1.2]{ni80b}), 
since $\Pi^\perp\cap H_\ast(X,\Z)\cong U^3(2)\oplus U$ 
is an even lattice of signature $(4,4)$ whose discriminant group has length $6=4+4-2$.
For the Niemeier lattice $N$ it is a consequence of the fact that the stabilizer
group of an octad in the Golay code is isomorphic to $(\Z_2)^4\rtimes A_8$,
where $A_8$ is realised as the group of permutations undergone by the $8$
points of that octad according to \cite{tod66}. Indeed, Conway's proof of \cite[Thm.~2.10]{co71}
shows that the induced action on $\wt\Pi$ is precisely our 
group of effective lattice automorphisms $\Aff(\F_2^4)\cong (\Z_2)^4\rtimes A_8$.
Hence every $\alpha\in \Aff(\F_2^4)\subset\Aut(\wt\Pi)$ is the restriction of some lattice
automorphism of $N$ to the sublattice $\wt\Pi$, which in fact  is uniquely determined. 
In the case of extending $\alpha$ to the full K3-homology, we do not have a uniqueness
statement in general. However, restricting to lattice automorphisms which are induced by some
holomorphic symplectic automorphism of a K3 surface, uniqueness follows from the 
Torelli Theorem. 

In summary, the lattice automorphisms that are of interest in this work are completely
determined by their restrictions to $\Pi$ and $\wt\Pi$, respectively. This explains why
extending Kondo's lattice $L_G$ to a lattice $M_G$ which contains $\Pi$ is such a useful
idea. That this idea works is one of our main results. 
{It in particular implies that a bijection $\theta$ between $H_\ast(X,\Z)$ and $N(-1)$ 
can be constructed which is compatible with the symmetry group $G_t\cong(\Z_2)^4$ of 
generic Kummer surfaces.}

While it thus may be tempting to expect some natural geometric interpretation of the group
$(\Z_2)^4\rtimes A_8$ in terms of holomorphic symplectic automorphisms 
of Kummer surfaces, we emphasize again that the overarching  
symmetry group  $(\Z_2)^4\rtimes A_7$ is the largest group which can be induced by a map $\Theta$
obeying all our assumptions. As mentioned above, for other pairs of Kummer surfaces one
will be forced to drop the assumption that an overarching map fixes the images of $\upsilon_0-\upsilon$
in $N(-1)$. Then one may be able to generate $(\Z_2)^4\rtimes A_8$ {\cite{tawe12}}.

It should be noted that our map {$\Theta$} is an isometry only
when restricted to the appropriate lattices $M_G$. Equivalently, it need not be $G$-equivariant
outside of $M_G$. The virtue of the specific map {$\Theta$} which we construct in this work  
is therefore its compatibility with two \textsl{distinct} Kummer surfaces, namely the tetrahedral one
and the one obtained from the square torus, and with all the Kummer surfaces along
a special path connecting the two. 
This property defines {$\Theta$} uniquely, up to
a few choices of signs discussed in section \ref{overmen}, and only this property allows us to combine 
the actions of the symmetry groups $\m T_{192}$ and $\m T_{64}$ of the two
Kummer surfaces to obtain the overarching symmetry group
$(\Z_2)^4\rtimes A_7$.

This group has $40320$ elements and thus is by orders of magnitude larger than the biggest
finite symplectic automorphism group of any K3 surface. Our techniques thus 
mitigate the ``order of magnitude'' problem that was mentioned in the introduction.
We are currently generalising
our techniques by replacing the Niemeier lattice $N$ by the Leech lattice. This way we hope to obtain the action of the 
entire Mathieu group $M_{23}$ or $M_{24}$ on the Leech lattice
merely using classical symmetries of K3 surfaces. This
may mean that ``Mathieu moonshine'' has nothing to do with symmetries
beyond the classical ones.

\section*{Acknowledgements}
We wish to thank Robert Curtis for sharing with us important insights on $M_{24}$ and infecting 
us with the MOG virus. 
We also benefitted from conversations with Geoffrey Mason and Michael Tuite, 
and many participants of the workshop on `Mathieu Moonshine' at ETH Zurich in July 2011, 
including its organiser, Matthias Gaberdiel,
whom we wish to thank for such an inspiring meeting. Matthias Sch\"utt deserves our gratitude for 
a number of helpful comments. Moreover, we are grateful to a 
referee for a correction in the proof of 
Proposition 2.3.3 and comments leading to the Remark 3.2.3.
We thank the London Mathematical Society for a Scheme 4 Collaborative Small Grant. 
AT is  supported in part by a Leverhulme Research Fellowship RF/2012-335. She  thanks Augsburg and Freiburg Universities for their hospitality.
KW is supported in part by ERC-StG grant No.~204757-TQFT. She thanks Durham University for its hospitality.
We  furthermore thank the Hausdorff Research Institute of Mathematics and the Max-Planck Institute for
Mathematics in Bonn
for hospitality -- a major part of
the final version of this work was completed at Bonn.

The free open-source mathematics software system  SAGE\footnote{W.A. Stein et~al., \emph{{S}age 
{M}athematics {S}oftware \mbox{\rm(}{V}ersion
  4.6.2}\mbox{\rm)}}, The Sage Development Team, 2010, \url{http://www.sagemath.org}}
 and in particular its 
component GAP4\footnote{The GAP~Group, \emph{GAP -- Groups, Algorithms, and Programming, 
  Version 4.4.12}; 
  2008,
  \url{http://www.gap-system.org}} were used in checking several assertions made in the group 
theoretical component of our work.

\appendix

 \setcounter{equation}{0}

 \section{The Mathieu group \boldmath{$M_{24}$}, the binary Golay code and the MOG}
\label{MOG}
The relation between Kummer lattices and the group $M_{24}$ can be made explicit by thinking of $M_{24}$ 
as the proper subgroup of  $A_{24}$ - the group of even permutations of $24$ objects 
labelled by the elements of ${\cal I}=\{1,2,3,\ldots,24\}$ - that preserves the 
extended binary Golay code ${\mathcal G}_{24}$. The latter is the dimension $12$ quadratic residue code 
of length $23$ over the field $\mathbb{F}_2$, extended in such a way that each element is augmented by 
a zero-sum check digit as described in \cite{cosl88}. The vector space ${\mathcal G}_{24}$ contains 
$2^{12}$ vectors called \textsc{codewords}, each being an element of $\mathbb{F}_2^{24}$, with the 
restrictions that their
\textsc{weight} (the number of non-zero entries) is a multiple of $4$, bar $4$ itself {and $20$}. 
The code contains exactly one codeword zero and one codeword where all digits are one, together 
with $759$ octads, $2576$ dodecads and $759$ complement octads.

Besides its description as a vector of $\mathbb{F}_2^{24}$ with components $c_k, k=1,\ldots,24$, a non-zero codeword 
is often represented in the main text
by a subset of ${\cal I}$ of cardinality 8, 12, 16 or 24, whose elements are the integers $k$ labelling $c_k \neq 0$. So for instance, the word
$(1,1,1,0,1,0,0,0,0,0,0,1,0,0,0,1,0,1,0,0,0,1,0,0)$ may be represented by the set \break$\{1,2,3,5,12,16,18,22\}$.
That such a collection of eight labels is actually a weight $8$ codeword
of the binary Golay code can be readily checked by using an extremely powerful (and playful!) technique devised by 
Robert Curtis in the course of his extensive study of $M_{24}$ \cite{cu74} and that we refer to as `mogging', as it uses the Miracle Octad Generator (MOG).
We have used a variant of the original technique, which was developed by Conway 
shortly after, and combines the hexacode ${\m H}_6$ with the Miracle Octad Generator 
(MOG). These tools are well-documented in the literature (see  \cite{cosl88} for instance), 
and we therefore confine ourselves to the bare essentials.

The hexacode ${\m H}_6$ is a $3$-dimensional code of length $6$ over the field of four 
element $\mathbb{F}_4=\{0, 1, \omega, \omega^2\}$, with $\omega^3=1$, $1+\omega=\omega^2$, and 
$\bar{\omega}:=\omega^2$.
It may be defined as
 $$
 {\m H}_6=\{ (a, b, \phi(0), \phi(1), \phi(\omega), \phi (\bar{\omega})) | a, b, \phi(0) \in \mathbb{F}_4, \, \phi(x):=ax^2+bx+\phi(0)\}.
 $$
The MOG is given by a $4 \times 6$ matrix whose entries are elements of 
$\mathbb{F}_2=\{0, 1\}$, and therefore provides binary words of length $24$. 
To check which among those words are Golay codewords, one proceeds in three steps.
\begin{enumerate}
\item 
Step 1: take a MOG configuration and calculate the parity of each $4$-column of the 
MOG and the parity of the top row (the parity of a column or a row being the parity of the 
sum of its entries); they must be all equal.
\item  
Step 2: to each 4-column with entries $\alpha, \beta, \gamma, \delta \in \mathbb{F}_2$, 
associate the $\mathbb{F}_4$ element $\beta +\gamma \omega +\delta \bar{\omega}$ 
called its \textsc{score}.
\item 
Step 3: check whether the set of six scores calculated from a given MOG form a 
hexacode word. If they do, then the original MOG configuration corresponds to a 
Golay codeword. One may take advantage of the fact that if $(a, b, c, d, e, f)$ is a 
hexacode 
word, then so are $(c, d, a, b, e, f), (a, b, e, f, c, d)$ and $(b, a, d, c, e, f)$.
\end{enumerate}
For instance, the MOG configuration
$$
\begin{array}{|rr|rr|rr|}
\hline
0&1&1&0&0&0\\
0&1&0&0&1&0\\
\hline
0&0&1&1&1&0\\
0&0&0&1&0&0\\
\hline \end{array}
$$
is such that all parities of columns and of top row are even, so the configuration passes 
Step 1. The ordered scores are $(0, 1, \omega, 1, \bar{\omega}, 0)$, and one must 
attempt to rewrite this $6$-vector as $(a, b, \phi(0), \phi(1), \phi(\omega), \phi (\bar{\omega}))$ 
for a quadratic function $\phi(x)=ax^2+bx+\phi(0)$. In the present case, 
$a=0, b=1$ and $\phi(0)=\omega$, so we see that $\phi(x)=x+\omega$, and hence 
$\phi(1)=\bar{\omega}$, which differs from the fourth entry of the ordered scores vector. 
The latter is therefore not a hexacode word, and the MOG configuration does not yield a 
Golay codeword. The power of the MOG in this context resides in the fact that
all Golay codewords can be obtained as MOG codewords.  

The connection between subsets of ${\m I} =\{1,\ldots,24\}$ and Golay codewords is 
made possible through the use of a special $4 \times 6$ array whose entries are the 
elements of ${\m I}$, distributed in one of two ways, according to
\ba
\begin{array}{|rr|rr|rr|}
\hline
24&23&11&1&22&2\\
3&19&4&20&18&10\\
\hline
6&15&16&14&8&17\\
9&5&13&21&12&7\\
\hline \end{array}_{\,M}&&\mbox{ or }\qquad\quad
\begin{array}{|cc|cc|cc|}
\hline
23&24&1&11&2&22\\
19&3&20&4&10&18\\
\hline
15&6&14&16&17&8\\
5&9&21&13&7&12\\
\hline \end{array}_{\,M'}
\label{MOGarray}
\ea
The distribution $M$ is the original Curtis configuration, while the mirror 
distribution $M'$ is due to Conway. Our labelling conventions for the codewords are compatible with the second 
version $M'$, but our results could be rederived using the version $M$, provided an appropriate
relabelling. 

Starting with a subset of eight distinct elements of ${\m I}$, one constructs a MOG 
configuration using $M'$, where entries corresponding to elements in the subset are 
$1$'s and the $16$ other entries are $0$'s. It remains to apply Steps 1 to 3 to 
conclude whether or not the initial set corresponds to a Golay codeword. 
For instance, the set $\{1, 2, 3, 4, 5, 6, 7, 8\}$
corresponds to the MOG configuration
$$
\begin{array}{|rr|rr|rr|}
\hline
0&0&1&0&1&0\\
0&1&0&1&0&0\\
\hline
0&1&0&0&0&1\\
1&0&0&0&1&0\\
\hline \end{array}_{\,M'}\,\,,
$$
which fails the parity test (Step 1), and therefore does not yield a Golay codeword. 
The same technique may be used to check whether a subset of $12$ elements in 
${\m I}$ is a dodecad.

As an application of the MOG technique, one can check that the following
twelve codewords form a basis of the Golay code:

\begin{eqnarray}
{\m O}_1=\{1, 2, 16, 18, 5, 12, 22, 3\}, &&{\m O}_6=\{6, 3, 22, 4, 21, 17, 15, 9\},  \nonumber\\
{\m O}_7=\{7, 4, 19, 10, 12, 15, 8, 16\}, && {\m O}_8=\{8, 7, 15, 17, 11, 23, 18, 16\},  \nonumber\\ 
 {\m O}_{16}=\{16,1, 2, 21, 4, 7, 8, 18\},&& {\m O}_{18}=\{18, 1, 16, 8, 23, 13, 14, 5\},  \nonumber \\
{\m O}_{20}=\{20, 10, 11, 17, 14, 13, 22, 12\},&&{\m O}_{23}=\{23, 3, 9, 19, 13, 18, 8, 11\},  \nonumber\\
 {\m O}_{24}=\{24, 2, 10, 19, 9, 5, 14, 21\}, &&{\m D}_1=\{ 8, 7, 15, 9, 19, 23, 4, 22, 13, 18, 1, 16\},  \nonumber\\
 {\m D}_2=\{18, 23, 13, 22, 3, 1, 11, 10, 2, 16, 7, 8\},&&
{\m D}_3=\{16, 1, 2, 10, 12, 7, 5, 9, 15, 8, 23, 18\}. \nonumber\\
\label{Golaybasis}
\end{eqnarray}

We now indicate how one may show that the 
group ${\m T}_{192}$ preserves the Golay code. Act with each generator in 
\req{t192} on a basis of the Golay code, for instance, the basis introduced in 
\req{Golaybasis},
%
%
and show that the resulting sets of eight or twelve elements correspond to  Golay codewords.
For instance, take the first generator  from \req{t192},
$$
\iota_1=(1, 11)(2, 22)(4, 20)(7, 12)(8, 17)(10, 18)(13, 21)(14, 16),
$$
acting on the first basis vector ${\m O}_1$,
$$
\iota_1({\m O}_1)=\iota_1 (\{1, 2, 16, 18, 5, 12, 22, 3\})=\{11, 22, 14, 10, 5, 7, 2, 3\}.
$$
The corresponding MOG configuration is
$$
\begin{array}{|rr|rr|rr|}
\hline
0&0&0&1&1&1\\
0&1&0&0&1&0\\
\hline
0&0&1&0&0&0\\
1&0&0&0&1&0\\
\hline \end{array}_{\,M'}\quad,
$$
which passes the parity test. Furthermore, the score vector is 
$(\bar{\omega}, 1, \omega, 0, \omega, 0)$ and for the MOG configuration to 
correspond to a Golay codeword, one needs to  identify the score vector with 
$(\bar{\omega}, 1, \phi(0), \phi(1), \phi(\omega), \phi(\bar{\omega})$ where 
$\phi(x)=\bar{\omega}x^2+x+\omega$. Since 
$\phi(0)=\omega, \phi(1)=\bar{\omega}+1+\omega=0,
 \phi(\omega)=\bar{\omega}\omega^2+\omega+\omega= \omega$ 
 and $\phi(\bar{\omega})=\bar{\omega}^3+\bar{\omega}+\omega=0$, we are through: the set
$\iota_1(\{1, 2, 16, 18, 5, 12, 22, 3\})$ is an octad.

A related technique used in this work consists in constructing the unique octad associated with
$5$ given elements of ${\m I}$ via the MOG. Suppose we choose the set $A=\{3, 6, 14, 17, 18\}$ 
and wish to complete $A$ so that one obtains an octad. First, one constructs a MOG start 
configuration where one replaces the elements belonging to $A$ by $1$, and all elements in 
${\m I} \setminus A$ by nothing in the Conway MOG array $M'$ of \req{MOGarray},
$$
\begin{array}{|rr|rr|rr|}
\hline
&&&&&\\
&1&&&&1\\
\hline
&1&1&&1&\\
&&&&&\\
\hline \end{array}_{\,M'}.
$$
Then one observes that, were all the blanks replaced by $0$'s,  $3$ columns would have odd parity, 
and $3$ would have even parity, while the top row also would have even parity. One has three 
extra entries of $1$ to distribute in such a way that the configuration passes the parity test. 
If the solution corresponds to odd parity, columns $3$, $5$ and $6$ cannot accommodate 
more entries of $1$, so the score vector is
partially known and reads $(a, b, \omega, \phi(1), \omega, 1)$, with $a, b \in \mathbb{F}_4$ 
and $\phi(x)=ax^2+bx+c$. So $\phi(0)=c=\omega$ and the system  
$\phi(\omega)=a\bar{\omega}+b\omega+\omega=\omega; \phi(\bar{\omega})=a\omega+b\bar{\omega}+\omega=1$ 
has no solution for $a,b \in \mathbb{F}_4$. One thus tries a solution corresponding to even parity. 
In this case, columns $1$, $2$ and $4$ cannot accommodate more entries of $1$  if one hopes to pass the 
parity test. The partial score is $(0, \bar{\omega}, \omega+n, 0, \phi(\omega), \phi(\bar{\omega}))$, 
with $\phi(x)=\bar{\omega}x+\omega+n$ for $n=0, 1$ or $\bar{\omega}$. The equation 
$\phi(1)=\bar{\omega}+\omega+n=0$ implies that $n=1$. Thus 
$\phi(\omega)=1+\omega+1=\omega$ and $\phi(\bar{\omega})=\omega+\omega+1=1$.
The reconstructed hexacode word is $ (0, \bar{\omega}, \bar{\omega}, 0, \omega, 1)$ 
and the   octad MOG configuration  is thus,
$$
\begin{array}{|rr|rr|rr|}
\hline
0&0&0&0&1&1\\
0&1&1&0&0&1\\
\hline
0&1&1&0&1&0\\
0&0&0&0&0&0\\
\hline \end{array}_{\,M'}.
$$
In other words, the unique octad formed from the partial knowledge encoded in the set 
$\{3, 6, 14, 17, 18\}$ is given by $\{2, 3, 6, 14, 17, 18, 20, 22\}$.

\section{Tables}\label{latticeappendix}

The following tables provide a roadmap through our notations. They list the 
relevant lattices that we introduce in sections \ref{lattice} and \ref{k3geometry} 
and which we use throughout this
work. As a shorthand we introduce the following notation: Consider an even unimodular
lattice $\Gamma$.  To state that $\Gamma$ can be obtained by the gluing construction 
of proposition \ref{glueperp}
from a pair of primitive sublattices $\Lambda,\,\m V\subset\Gamma$, where 
$\Lambda$ and $\m V$ are orthogonal complements of one another in $\Gamma$,
we write $\Gamma=\Lambda\Join\m V$.
\begin{landscape}
\small
$$
\begin{array}{|c||c|c|c||c|c|c|}
\hline
\mbox{rank}&\mbox{sign. \&}&\mbox{lattice}&\mbox{ref.} &\mbox{sign. \&}&\mbox{lattice}&\mbox{ref.} \\
&\mbox{disc. grp.}&&\mbox{in text}&\mbox{disc. grp.}&&\mbox{in text}\\
\hline
&&H_\ast(X,\Z)=H_0(X,\Z)\oplus H_2(X,\Z)\oplus H_4(X,\Z)&&&&\\
24&(4,20)&\cong U^4\oplus E_8^2(-1)&&(24,0)&N=\{\nu \in (A_1^{24})^\ast\mid\bar{\nu} \in {\m G}_{24}\}&\mbox{Prop.~\ref{niefromroot}}\\
&&\mbox{full integral homology lattice of K3 surface } X& &&\mbox{Niemeier lattice of type } A_1^{24}&\\
&\{\mbox{id}\}&H_\ast(X,\Z)={\m K}\Join{\m P}&\mbox{Prop.~\ref{fullk3}}&\{\mbox{id}\}&N=\widetilde{K}\Join \widetilde{\Pi}=\widetilde{{\m K}}_{n_0}\Join \widetilde{\m P}_{n_0}&\mbox{Prop.~\ref{fullniemeierprop}}\\
\hline
&(3,19)&H_2(X,\Z)\cong U^3\oplus E_8^2(-1)&\mbox{after}&&&\\
22&&=K\Join \Pi&\mbox{Def.~\ref{k3definition}}&&&\\
&\{\mbox{id}\}&\mbox{K3-lattice}&\mbox{before \req{lambdaij}}&&&\\
\hline
&&&&&&\\
17&(0,17)&{\m P}:=\Pi\oplus \spann_\Z\{\upsilon_0-\upsilon\}&&(17,0)&\widetilde{{\m P}}_{n_0}=\widetilde{\Pi}\oplus \spann_\Z\{f_{n_0}\}\cong {\m P}(-1)&\\
&(\Z_2)^7&\mbox{primitive sublattice of }H_\ast(X,\Z)&\mbox{Prop.~\ref{fullk3}}&(\Z_2)^7&\mbox{primitive sublattice of }N
&\mbox{Prop.~\ref{fullniemeierprop}}\\
\hline
&&&&&\widetilde{\Pi}\cong\Pi(-1)&\\
16&(0,16)&\Pi=\spann_\Z\{E_{\vec{a}},\vec{a}\in \F_2^4, \hf \sum_{\vec{a}\in H_i}E_{\vec{a}}\} 
&&(16,0)&\widetilde{\Pi}=\spann_\Z\{f_n, n\notin {\m O}_9, \hf \sum_{n \in {\m H}_i}f_n\}&\mbox{Prop.~\ref{piinn}}\\
&&H_i, i=1,..,5 \mbox{ affine hyperplanes in }\F_2^4&&&{\m H}_i, i=1,..,5 \mbox{ octads, }{\m H}_i \cap {\m O}_9=\emptyset&\req{GolayKummer}\\
&(\Z_2)^6&\mbox{Kummer lattice; primitive sublattice of }H_2(X,\Z)&\mbox{Prop.~\ref{Kummerform}}&(\Z_2)^6&\mbox{primitive sublattice of }N&\mbox{Prop.~\ref{piprim}}\\
\hline
&&&&&&\\
8&&&&(8,0)&\widetilde{K}=\spann_\Z\{f_n, n \in {\m O}_9, \hf \sum_{n \in {\m O}_9}f_n\}&\\
&&&&(\Z_2)^6&\mbox{primitive sublattice of }N&\mbox{Prop.~\ref{piinn}}\\
\hline
&&&&&&\\
&(4,3)&{\m K}:={\m P}^{\perp}\cap H_\ast(X,\Z)&&(7,0)&\widetilde{\m K}_{n_0}:=\wt{\m P}_{n_0}^{\perp}\cap N&\\
7&&=K\oplus \spann_\Z\{\upsilon_0+\upsilon\}&&&=\spann_{\Z} \{f_n, n \in {\m O}_9\setminus \{n_0\}\}&\\
&(\Z_2)^7&\mbox{primitive sublattice of }H_\ast(X,\Z)&\mbox{Prop.~\ref{fullk3}}&(\Z_2)^7&\mbox{primitive sublattice of }N
&\mbox{Prop.~\ref{fullniemeierprop}}\\
\hline
&&K:=\pi_\ast(H_2(T,\Z))\cong U^3(2)&&&&\\
6&(3,3)&=\Pi^{\perp} \cap H_2(X,\Z)&\mbox{Prop.~\ref{Kummerform}}&&&\\
&&K=\spann_\Z\{\pi_\ast\lambda_{ij}, i < j, i,j=1,..,4\}&&&&\\
&(\Z_2)^6&\mbox{primitive sublattice of }H_2(X,\Z)&\mbox{after \req{Kummerglue}}&&&\\
\hline
2&(1,1)&H_0(X,\Z)\oplus H_4(X,\Z) \cong U&&&&\\
&\{\mbox{id}\}&\mbox{hyperbolic lattice}&\req{hyperbolic}&&&\\
\hline
\end{array}
$$
\refstepcounter{table}{1}
\label{latticetable}
\begin{center}
{\textbf{Table \thetable:} Notations for lattices introduced in section \ref{lattice} and used throughout the text.}
\end{center}
\end{landscape}

%
\begin{table}[ht]
$$
\begin{array}{|c|c|c|}
\hline
\mbox{signature}&\mbox{lattice}&\mbox{Ref.  in text}\\
&&\\
\hline
&&\\
(4,m_-),&L^G \mbox{ sublattice of }H_\ast(X,\Z) \mbox{ invariant}&\\
m_-\ge 1&\mbox{under a group } G \mbox{ of symplectic}&\mbox{Prop.~\ref{mukaifinite}}
\\
&\mbox{automorphisms preserving the}&\\
&\mbox{dual K\"ahler class}&\\

&&\\
\hline 
&&\\
(0,20-m_-),&L_G=(L^G)^{\perp}\cap H_\ast(X,\Z)&\mbox{Prop.~\ref{mukaifinite}}\\
m_-\ge 1&&\\
\hline
&&\\
(0,22-l_-),&M_G:=L_G\oplus \spann_\Z\{e,\upsilon_0-\upsilon\}&\mbox{Prop.~\ref{latticechar}}\\
l_- \ge 2&&\\
\hline
&&\\
(4,l_- -2),&M'_G:=\pi_\ast((H_2(T,\Z))^{G'_T})\oplus \spann_{\Z}\{\upsilon_0+\upsilon\}&\mbox{Prop.~\ref{latticechar}}\\
l_-\ge 2&=M_G^{\perp}\cap H_\ast(X,\Z)&\\
&&\\
\hline
&&\\
(0,1)&\Pi^G:=\Pi \cap H_\ast(X,\Z)^G&\mbox{Prop.~\ref{latticechar}}\\
&=\spann_\Z\{e\}&\\ 
&&\\
\hline
\end{array}
$$

\caption{\label{latticetableG}Notations for lattices introduced in section \ref{k3geometry} and used throughout the text.
$G$ denotes the group of  dual K\"ahler class preserving symplectic automorphisms  of a Kummer surface $X$
with complex structure and dual K\"ahler class induced from the underlying torus $T$.  
 The group of non-translational holomorphic symplectic automorphisms of $T$
 is $G_T^\prime$.}
\end{table}



%
\end{document}